\documentclass[USenglish]{article}	
\usepackage[utf8]{inputenc}				
\usepackage[big,online]{dgruyter}	
\usepackage{lmodern} 
\usepackage{microtype}
\usepackage[numbers,square,sort&compress]{natbib}
\usepackage{comment}
\usepackage{xcolor}
\usepackage{enumitem} 

\RequirePackage{amsthm,amsmath,amsfonts,amssymb}
\RequirePackage{graphicx}

\newtheorem{theorem}{Theorem}[section]
\newtheorem{lemma}[theorem]{Lemma}
\theoremstyle{remark}

\newcommand{\nl}{\newline}
\newcommand{\pl}{\parallel}
\newcommand{\openr}{\hbox{${\rm I\kern-.2em R}$}}
\newcommand{\openn}{\hbox{${\rm I\kern-.2em N}$}}
 
\newcommand{\argmin}[1]{\underset{#1}{\operatorname{argmin}}\;} 
\usepackage{mathtools}
\DeclarePairedDelimiterX{\set}[1]{\lbrace}{\rbrace}{#1}
\newcommand{\indicator}[1]{\mathbb{I}_{\set*{#1}}}
\DeclarePairedDelimiterX{\norm}[1]{\lVert}{\rVert}{#1}
\newcommand\numberthis{\addtocounter{equation}{1}\tag{\theequation}}

\newcommand{\newbrace}[1]{\left\{#1\right\}}
\newcommand{\newbracket}[1]{\left[#1\right]}
\newcommand{\newparethensis}[1]{\left(#1\right)}
\newcommand{\ex}{\mathbb{E}}

\theoremstyle{dgthm}

\theoremstyle{dgdef}
\newtheorem{remark}{Remark}

\def\acknowledgementsname{Acknowledgments}
\newenvironment{acks}[1][\acknowledgementsname]{\section*{#1}}{\par}

\makeatletter
  \AtBeginDocument{%
    \setbox\dg@wordmark=\hbox{}%
  }
\makeatother

\usepackage{etoolbox}   
\makeatletter
  \patchcmd{\ps@plain}
    {\rlap{\vrule\@width\textwidth\@height-6\p@\@depth7.5\p@}}
    {}
    {}{}
\makeatother

\begin{document}


\title{Super Ensemble Learning Using the Highly-Adaptive-Lasso}
\runningtitle{Meta-Learning using Highly-Adaptive-Lasso}

\author*[1]{Zeyi Wang}
\author[2]{Wenxin Zhang}
\author[3]{Brian S Caffo}
\author[3]{Martin Lindquist}
\author[2]{Mark van der Laan} 
\runningauthor{Z. Wang et al.}
\affil[1]{\protect\raggedright 
Division of Biostatistics, School of Public Health, University of California, Berkeley, e-mail: wangzeyi@berkeley.edu; Department of Statistics, Oklahoma State University, e-mail: zeyi.wang@okstate.edu}
\affil[2]{\protect\raggedright 
Division of Biostatistics, School of Public Health, University of California, Berkeley, e-mail: wenxin\_zhang@berkeley.edu, laan@berkeley.edu}
\affil[3]{\protect\raggedright 
Division of Biostatistics, Johns Hopkins Bloomberg School of Public Health, e-mail: bcaffoweb@jhu.edu,	mlindquist@jhu.edu}

\abstract{
We introduce the Meta Highly-Adaptive-Lasso Minimum Loss Estimator (M-HAL-MLE), a novel ensemble approach for estimating functional parameters of realistically modeled data distribution from independent and identically distributed observations. Given $J$ initial estimators, candidate ensembles are generated by finite-sectional-variation cadlag functions. Using $V$-fold cross-validation, the M-HAL-MLE selects the optimal cadlag ensemble minimizing the cross-validated empirical risk, with the sectional variation bound as a tuning parameter. The final estimator, M-HAL super-learner, is obtained by averaging ensemble compositions across folds. In contrast, the oracle ensemble and oracle estimator are defined by minimizing the population excess risk relative to the true function. We establish following theoretical properties: 1) the M-HAL super-learner converges to the oracle estimator at rate $n^{-2/3}$ in excess risk, up to log-n factors; 2) by appropriate undersmoothing, target features of the M-HAL super-learner are asymptotically linear for corresponding target features of the oracle estimator; 3) the excess risk between the oracle estimator and true function, along with the difference between their target features, is generally second-order. Simulations validate the theoretical results, demonstrating effectiveness in high-dimensional settings. We further illustrate the method in a real-data application involving mediation analysis of functional MRI from human pain studies.
}


\maketitle
	
\section{Introduction}

We consider estimation of a functional  parameter of a 
realistically modeled data distribution based on observing independent and identically distributed observations of a $d$-dimensional Euclidean valued random variable.
Suppose that the true function is a $k$-variate real valued function defined 
as the minimizer over its parameter space of the expectation of a specified uniformly bounded loss function. 

{\bf HAL-MLE:}
In our previous work \citep{vanderLaan15,vanderLaan17,Benkeser&vanderLaan16,Bibaut&vanderLaan19} we showed that if the parameter space consists of $k$-variate real valued cadlag functions with a universal bound on the sectional variation norm \citep{Gill&vanderLaan&Wellner95,vanderLaan15}, then  the MLE over this parameter space, selecting the variation norm bound with cross-validation, converges to the true function at a rate $n^{-2/3}(\log n)^d$ w.r.t. the loss-based dissimilarity, also called excess risk. Moreover, computation of this estimator corresponds with minimizing an empirical risk over a linear combination of spline-basis functions under the constraint that the $L_1$-norm of the coefficient vector is bounded by this sectional variation norm bound, making it a high dimensional lasso estimation problem.

A general theory for undersmoothing sieve based estimators is developed in \citet{shen97,shen07}, and a powerful demonstration of such an estimator also presented in \citep{Newey94}.  
We have shown that an undersmoothed HAL-MLE that selects the $L_1$-norm  larger than the cross-validation selector (but still bounded) according to a specified global undersmoothing criterion  is asymptotically efficient \citep{Bickeletal97} for any pathwise differentiable target parameter, under weak regularity conditions \citep{vanderLaan&Benkeser&Cai19}.
Thus, the smooth target features of the HAL-MLE were shown to be asymptotically efficient for the target features of the true function,  if the sectional variation norm satisfies a global undersmoothing criterion, even though this HAL-MLE is not targeted towards that target feature.  

{\bf Super learner with small family of ensembles:}
In other past research, we have proposed super-learning as a general optimal approach to learn a true function \citep{vanderLaan&Dudoit03,vanderVaart&Dudoit&vanderLaan06,vanderLaan&Dudoit&vanderVaart06,vanderLaan&Polley&Hubbard07,Chpt3}. The super-learner selects a library of estimators, defines a collection of ensembles such as all convex combinations, and chooses the ensemble that minimizes the cross-validated empirical risk of the ensemble specific candidate estimator. For simplicity, let's consider the case that we use $V$-fold sample splitting, so that the cross-validated empirical risk is defined as the average over the $V$ sample splits in training and validation sample of the empirical mean over the validation sample of the loss function at the ensemble specific estimator trained on the training sample. In particular, the discrete super-learner (the collection of  ensembles is trivially defined as the set of estimators in the library) simply selects the estimator in the library that minimizes the cross-validated empirical risk. Given the cross-validated selected ensemble, one could either rerun the selected ensemble of estimators based on the whole data set, or one can simply take the average over the $V$ sample splits of the selected ensemble  of the estimators based on the training sample only. The later is immediate available as a by product of the cross-validated empirical risk of the cross-validation selector of the ensemble. In this article we will represent the super-learner as this average across $V$ sample splits.

{\bf Asymptotic equivalence of super-learner (with small family of ensembles) with oracle selected ensemble:}
Under some constraints on the size of the family of ensembles, the excess risk of the super-learner divided by the excess risk of the oracle selected ensemble converges to 1 as sample size increases (see references above): we say that the cross-validation selected ensemble (i.e, super-learner) is asymptotically equivalent with the oracle selected ensemble. The oracle ensemble minimizes the excess risk among all ensemble specific estimators applied to  data sets of the size of the training sample. In particular, if one uses the discrete super-learner, then it is asymptotically equivalent with the oracle selected estimator that minimizes the excess risk, as long as the number of candidate estimators is polynomial in sample size. When the family of ensembles is defined as all convex combinations, the asymptotic equivalence will require the number of estimators in the library to grow slowly with sample size (i.e., $\log n$).
By including the HAL-MLE in its library, it follows that the discrete  super-learner converges at least as fast as $n^{-2/3}(\log n)^d$, while being adaptive to unknown structure of the true function if other candidate estimators are tailored to such structure. However, this gain of adaptivity of this super-learner relative to HAL-MLE comes also at a price in the sense that pathwise differentiable target features of this super-learner (even when it includes the HAL-MLE in its library) are not asymptotically linear estimators of the true features, so that one cannot provide formal statistical inference. 
 
 {\bf Meta-HAL super-learner:}
In this article, we consider an aggressive  super-learner that further extends the family of ensembles to a class of multivariate real valued cadlag functions with a bound on its sectional variation norm. We refer to this super-learner as the Meta-HAL  super-learner since the computation of the cross-validation selector of the ensemble corresponds with applying the HAL-MLE at a meta-level data set in which each observation $O_i$ is coupled with a cross-fitted realization of the library of estimators. 

{\bf Size of ensemble family controlled by sectional variation norm bound:}
Due to the large size of this family of ensembles, asymptotic equivalence of the cross-validation selector with the oracle selector has not been established. This makes our results novel within the current super-learner literature. In fact, if the $J$ estimated functions represent a zero-loss transformation of the coordinate system for true function, then  the oracle estimator equals the true function. 
Potential  overfitting of the M-HAL-MLE in finite samples is effectively  controlled by the choice of sectional variation norm, making it a robust and powerful super-learner. 
In particular, 
we can select this sectional variation norm bound $C$  with a discrete super-learner with a library of $C$-specific M-HAL super-learners across a range of values for $C$, thereby guaranteeing that our cross-validation selector $C_n=C_{n,cv}$ will be asymptotically equivalent with the oracle selector of $C$. 

{\bf Due to large family of ensembles, candidate estimators can be a flexible data-adaptive coordinate-transformation:}
By selecting such a large family of ensembles, the library of estimators does not need to be restricted to good estimators of the functional parameter, but could include  fixed functions, simple parametric model based estimators, and intermediate layer outputs from externally pretrained networks,   not necessarily approximating or even aiming to approximate the true function. That is, one could also view the library of estimators as a proposed data-adaptive transformation of the coordinates $x\in \openr^k$ for the true function. For example, one extreme choice of transformation would be to simply propose $k$ fixed functions $f_1(x),\ldots,f_k(x)$ of $x$ so that the coordinate-transformation is invertible and thereby represents a zero-loss transformation. Adding a super-learner as the $k+1$-th function potentially allows a relatively simple (low sectional variation norm) yet zero-loss transformation. 
A fit from a previous study can also be added, transferring external model knowledge. 
Of course, one still has the option to define the library as $k$ highly adaptive estimators targeting the the true function. 
Therefore, we more generally refer to the library of estimators as a data-adaptive coordinate-transformation for the true function, emphasizing that they need not directly estimate the true function itself. 

{\bf M-HAL super-learner behaves as an HAL-MLE for a transformed and potentially simpler data problem 
}
Our results demonstrate that the M-HAL super-learner of the true function will generally behave as an HAL-MLE  of the $J$-dimensional oracle  cadlag function for a transformed data problem in which the $J$ coordinates for this oracle function are the cross-fitted $J$  estimated functions applied to the original $x$. Due to the transformed coordinates, the true oracle cadlag function can be a simpler function of the $J$ variables   than the true function is of the original $k$-dimensional input variables (potentially in much higher dimensions), which can be formally compared by contrasting the sectional variation norms of the two functions. This allows our method to reduce the complexity of the estimation problem, leading to a more robust and potentially superior estimator. For example, in plug-in estimation of target features of the true function, this strategy reduces the classical Donsker class conditions on functions of the original $k$-dimensional input variables to similar conditions on estimated ensembles of the $J$ transformed coordinates, leading to relaxed conditions and improved performance in complex, high-dimensional data settings

{\bf Target features of M-HAL super-learner:}
 In this article, analogue to our work on the HAL-MLE   \citep{vanderLaan&Benkeser&Cai19},  we will also analyze the undersmoothed M-HAL super-learner, selecting the sectional variation norm bound in the meta-HAL-MLE larger than the value suggested by cross-validation. In this case, we establish that under the meta-level analogue of the global  undersmoothing criterion for the HAL-MLE in \citep{vanderLaan&Benkeser&Cai19},  target features of the M-HAL-SL are either asymptotically efficient, or potentially super-efficient, depending on the coordinate-transformation implied by library of $J$ estimators. Either way, we will establish asymptotic linearity with known influence curve implied by canonical gradient of the pathwise derivative of the target parameter, and thereby allowing for formal statistical inference.

Thus, contrary to the regular discrete or convex ensemble super-learner, this highly aggressive super-learner using an undersmoothed HAL-MLE in the meta-learning step results in asymptotically  linear plug-in estimators of target features. 
Therefore, the M-HAL super-learner is not only at least as powerful as a regular (small family of ensembles) super-learner, but, when undersmoothed, its smooth features  are asymptotically linear, super-efficient or efficient.

\subsection{Organization of article}
In Section \ref{section2} and Section \ref{section3} we formally define the statistical estimation problem, the relevant quantities, and the M-HAL super-learner. We also provide a transformation/reduction of the observed data implied by the library of $J$ estimators, and corresponding statistical estimation problem addressed by the meta-learning step, where the latter treats the $J$-dimensional vector of cross-fitted estimates as a fixed $J$-dimensional coordinate-transformation. The latter reduced data statistical estimation problem will essentially make our study of the M-HAL-MLE cross-validation selector equivalent with our previous study of the HAL-MLE, and will therefore naturally guide our analysis as a meta-level analogue of our work on the HAL-MLE 
In Section \ref{sec:rate} we analyze the excess risk of the M-HAL super-learner. The excess risk will be  decomposed as a sum of the excess risk of the M-HAL super-learner relative to the oracle estimator (i.e., the best possible ensemble/cadlag function of the $J$ candidate estimators among all cadlag functions with sectional variation norm bounded by our bound),
 and the excess risk of the oracle estimator  relative to the true function (which would be zero if the coordinate-transformation is a zero-loss transformation).
In Section \ref{section5} we analyze a target feature of  M-HAL super-learner as estimator of the  target feature of  true function. The difference of this plug-in estimator with the true target estimand is decomposed as the sum of 1) the difference of plug-in M-HAL super-learner with  the plug-in oracle estimator, and 2) the difference of the plug-in of  oracle estimator  and the true target estimand. The latter is analyzed separately in Section \ref{sec:difference}, and is shown to be a second order difference, while the first is analyzed analogue to our previous analysis of the plug-in HAL-MLE \citep{vanderLaan&Benkeser&Cai19}.
In Section \ref{sec:TSM},
we demonstrate our general results of the M-HAL super learner and its plug-in estimation, by applying them in the nonparametric estimation of a treatment specific mean. 
In Section \ref{new_section6} and \ref{sec:data} we present numerical experiment results and a real-world data application of the proposed method to high-dimensional mediation analysis with fMRI data in pain studies. 
We conclude with a discussion in Section \ref{section8}. 

Basic outline of proofs are presented in the main article, while more technical results are presented in a  self-contained way in the Appendix. Appendix \ref{AppendixA1} and Appendix \ref{AppendixA2} presents a notation index that should help the reader, even though notation will be introduced in main article. 
Appendix \ref{sec:undersmoothing}  provides that undersmoothing makes the M-HAL super-learner solve the cross-validated empirical mean of the efficient influence curve equation, representing the only real challenge for establishing asymptotic linearity. In particular, we discuss in detail the two undersmoothing conditions (\ref{suffassumptionatarget}) and (\ref{a3}).
Appendix \ref{AppendixD} generalizes the consistency and asymptotic linearity results to the targeted M-HAL SL. Finally, Appendix \ref{AppendixE} provides deeper understanding of the undersmoothing condition (\ref{suffassumptionatarget}), and that it can be easily  achieved with bounded selectors of the sectional variation norm.

\section{Statistical Model}\label{section2}

Suppose we observe $n$ independent and identically distributed copies $O_1,\ldots,O_n$ with probability measure $P_0$ that is known to be an element of  a statistical model ${\cal M}$. We assume that $O\in \openr^d$ is a $d$-variate bounded random variable. Let $P_n$ be the empirical probability measure that puts mass $1/n$ on each $O_i$. 
We consider a functional parameter $Q:{\cal M}\rightarrow {\cal Q}\equiv \{Q(P): P\in {\cal M}\}$, where the parameter space ${\cal Q}$ represents a set of multivariate $[0,1]$-valued functions; that is, $Q(P): \openr^k\rightarrow [0,1], \forall P \in \mathcal{M}$. In practice, the target function is often defined on a bounded Euclidean set; therefore, without loss of generality we assume that variables can be standardized so that $Q(P): [0, 1]^k\rightarrow [0,1], \forall P \in \mathcal{M}$. 
Let  $L:{\cal Q}\rightarrow L^2(P_0)$ be a mapping from the parameter space ${\cal Q}$ into a set of $d$-variate real valued  functions of $O$ satisfying $Q_0=\arg\min_{Q\in {\cal Q}}P_0L(Q)$. We refer to $L$ as a loss function, where  $L(Q)(o)$ evaluates a loss of candidate $Q$ at observation $o$. 
Let $d_0(Q,Q_0)\equiv P_0L(Q)-P_0L(Q_0)$ be the loss-based dissimilarity, which denotes the excess risk of an estimator $Q$ with respect to the minimum risk of the class ${\cal Q}$.
It is assumed that $M_1=\sup_{Q,o}\mid L(Q)(o)\mid<\infty$ and $M_{20}=\sup_{Q\in {\cal Q}}\frac{P_0(L(Q)-L(Q_0))^2}{d_0(Q,Q_0)}<\infty$, so that the cross-validation selector is well behaved and generally asymptotically equivalent with the oracle selector \citep{vanderLaan&Dudoit03,vanderVaart&Dudoit&vanderLaan06}.

We will consider a pathwise differentiable target parameter $\Psi:{\cal M}\rightarrow\openr$. It is assumed that it is pathwise differentiable at $P$ with canonical gradient $D^*(P)$, and that $\Psi(P)$ only depends on $P$ through $Q(P)$. Let $G:{\cal M}\rightarrow{\cal G}$ be a functional parameter so that $D^*(P)$ only depends on $P$ through $Q(P)$ and $G(P)$. 
We will also denote $D^*(P)$ with $D^*(Q,G)$ and $\Psi(P)$ with $\Psi(Q)$.  Let $R_{20}(Q,G,Q_0,G_0)\equiv \Psi(Q)-\Psi(Q_0)+P_0D^*(Q,G)$ be the exact second order remainder implied by the canonical gradient. Let $L_1(G)$ be a loss function for $G_0=\arg\min_{G\in {\cal G}}P_0L_1(G)$, and $d_{01}(G,G_0)=P_0L_1(G)-P_0L_1(G_0)$. 

Our goal is to construct an estimator of $Q_0$, and to also use it to construct asymptotically linear estimators of $\Psi(Q_0)$, possibly for arbitrary $\Psi$ in a large class of smooth features.

\subsection{Data-adaptive Coordinate-Transformation with $V$-fold Cross-Validation}

We split the $n$ observations $O_1,\ldots,O_n$ in $V$ blocks, and for each choice $v$ of a block, let $P_{n,v}^1$ be the empirical measure of the observations $O_i$ in that block, and let $P_{n,v}$ be the empirical measure of the observations $O_i$ in the other $V-1$ blocks. We refer to $P_{n,v}^1$ and $P_{n,v}$ as the empirical measures of the validation and training sample for the $v$-th sample split, respectively. 

Let $\hat{{\bf Q}}=(\hat{Q}_j: j=1,\ldots,J)$ be a collection of $J$ algorithms, possibly a library of estimators of $Q_0$, but it is only required that $\hat{Q}_j$ maps data into an element of the parameter space ${\cal Q}$. For an empirical measure $P_{n,v}$ of a training sample extracted from $\{O_1,\ldots,O_n\}$, ${\bf Q}_{n,v}=\hat{\bf Q}(P_{n,v})\in {\cal Q}^J$ represents its realization applied to the empirical probability measure $P_{n,v}$. Let ${\cal M}_{np}$ denote the set of discrete empirical measures based on an arbitrary subset of $\{O_1,\ldots,O_n\}$ so  that $\hat{\bf Q}:{\cal M}_{np}\rightarrow {\cal Q}^J$ represents a vector of estimators that can be applied to arbitrary training samples extracted from $O_1,\ldots,O_n$. 
Let ${\bf Q}_n$ be a function of $(v,x)$ defined by ${\bf Q}_n(v,x)={\bf Q}_{n,v}(x)$.
Note that $\pmb Q_{n, v}: [0, 1]^k \to [0, 1]^J$ constructs a data-adaptive coordinate-transformation for $Q_0$, which is a realization of the algorithm $\hat{\bf Q}$ applied to the $v$-th training sample.

For observation $O_i$, let $v_i$ be the index of the block that contains the $i$-th observation $O_i$, $i=1,\ldots,n$. We represent our data set with $(v_1,O_1),\ldots,(v_n,O_n)$. One can  represent this sample as an  i.i.d. sample from the true distribution $P_0^{\bar{V}}$  of  a random variable $(\bar{V},O)$, where the conditional distribution of $O$, given $\bar{V}=v$, equals $P_0$, and $\bar{V}\sim U(1,\ldots,V)$ is uniform on $\{1,\ldots,V\}$. In this manner, for a function $f(v,o)$ of $(v,o)$, we can write $P_0^{\bar{V}} f=\frac{1}{V}\sum_{v=1}^V \int f(v, o)dP_0(o)$.
Let $P_n^{\bar{V}}$ be the empirical measure putting  mass $1/n$ on each $(v_i,O_i)$, $i=1,\ldots,n$.

In the next section, we will construct HAL-MLE with meta-level data, in which each observation $O_i$ is coupled with the cross-fitted algorithm realizations $\pmb Q_{n, v_i}$. 

\subsection{Family of Ensembles}

We define the set of ensembles as a class of cadlag functions with a bound on the sectional variation norm. 
That is, consider a collection ${\cal Q}^r\subset {\cal D}_{C^u}[0,1]^J$ of $J$-variate real valued cadlag functions $Q^r:[0,1]^J\rightarrow [0,1]$, with a uniform bound $C^u$ on its sectional variation norm. 
We have $Q^r(x)=Q^r(0)+\sum_{s\subset\{1,\ldots,J\}}\int \phi_{s,u_s}(x)dQ^r_s(u_s)$, where $\phi_{s,u_s}(x)=I(x_s\geq u_s)$, $Q^r_s(u_s)=Q^r(u_s,0_{-s})$ is the $s$-specific section that sets the coordinates in $s^c$ equal to zero \citep{vanderLaan17}.  The sectional variation norm of $Q^r$ is defined as $\pl Q^r\pl_v^*=\mid Q^r(0)\mid+\sum_{s\subset\{1,\ldots,J\}}\int \mid dQ^r_s(u_s)\mid$. 

For any $Q^r\in {\cal Q}^r$ and ${\bf Q}\in {\cal Q}^J$, a $Q^r$-ensemble of ${\bf Q}$ is given by 
$x \mapsto Q^r\circ {\bf Q}(x)=Q^r({\bf Q}(x))$. 
We assume that any ensemble estimator constructed by $Q^r \in {\cal Q}^r$ respects the parameter space ${\cal Q}$.  
Specifically, for each $v=1,\ldots,V$, 
\begin{equation}\label{respectparspace}
\{Q^r\circ{\bf Q}_{n,v}:Q^r\in {\cal Q}^r\}\subset {\cal Q}.
\end{equation} 

In practice, one can consider additive HAL models with respect to a subset of the most nonparametric ensemble space, ${\cal Q}^r$, for potentially better finite sample performance. For example, a hyperparameter can restrict the size of the $s$ section so that $\{Q^r \in {\cal Q}^r: Q^r(x) = Q^r(0)+\sum_{s\subset\{1,\ldots,J\}, |s| \leq U }\int \phi_{s,u_s}(x)dQ^r_s(u_s)\}\subset {\cal Q}^r$ involves only up to $U$-th order interactions. 
Such hyperparameters, that further restrict the class of ensembles for the additive HAL models, can be decided jointly along with $C^u$ by a cross-fitted discrete super learner similar to Section \ref{sec:selector}.

\subsection{Data Reduction Implied by Cross-fitted Coordinate-Transformation}\label{sec:data_reduction}

Note we can view $L(Q^r\circ {\bf Q}_n): (v, O) \mapsto L(Q^r \circ {\bf Q}_{n, v})(O)$ as a function of $(v, O)$. 
In many settings $X$ is a subvector of $O$, and $L(Q)(O)$ involves evaluating the function $Q$ at $X$, and $Q^r$ is only a function of original coordinates $X \subset O$ through transformations ${\bf Q}_{n,v}(X)$. 
In general, for any given $L$ and ${\bf Q}_n$, there exists a data reduction of $(\bar{V},O)$, 
$$O^r=O^r(\bar{V},O),$$ 
such that $L(Q^r\circ{\bf Q}_n)(v,O)$ depends on $(v,O)$ or ${\bf Q}_{n,v}(X)$ only through $(v,O^r(v,O))$. 
This allows us to define a loss for $Q^r$ with the reduced data, 
\begin{align}
L^r(Q^r)(v,O^r(v,O)) \equiv L(Q^r\circ{\bf Q}_n)(v,O). \label{eq:dimr_loss}
\end{align}
For example, if $L(Q)(X,Y)=(Y-Q(X))^2$, $O=(X,Y)$, then $L(Q^r\circ{\bf Q}_n)(v,O)=(Y-Q^r\circ{\bf Q}_{n,v}(X))^2$ depends only on $(v,O)$ or ${\bf Q}_{n,v}(X)$ through $(v,O^r(v,O))$ with $O^r(v,O)=(X^r_v\equiv {\bf Q}_{n,v}(X),Y)$; so we can define $L^r(Q^r)(v,O^r)=(Y-Q^r(X^r_v))^2$.

Note that $(\bar{V},O^r)$ represents a resulting data reduction of $(\bar{V},O)$ implied by the cross-fitted coordinated transformations ${\bf Q}_n$. Let $d^r$ be the dimension of $O^r$.

Let $P_0^r$ be the distribution of $(\bar{V},O^r)$ implied by $P_0^{\bar{V}}$, treating ${\bf Q}_n$ as a fixed function. Similarly, let ${\cal M}^r=\{P^r: P\in {\cal M}\}$ be the statistical model for the distribution $P_0^r$ of $(\bar{V},O^r)$ implied by the statistical model ${\cal M}$ for $P_0$. 
Let $P_n^r$ be the empirical probability measure of $(v_i,O^r_i)$, where $O^r_i=Q^r(v_i,O_i)$ is the data reduction of $O_i$ implied by ${\bf Q}_{n,v_i}$.
We also use notation ${\cal M}^r_v=\{P^r_v:P\in {\cal M}\}$, where $P^r_v$ denotes the distribution of $O^r_v=O^r(v,O)$ implied by $O\sim P$. 
Let $Q^r: {\cal M}^r \to {\cal Q}^r$ be a parameter of $P^r \in {\cal M}^r$ such that $Q^r(P^r) = \argmin{Q^r \in {\cal Q}^r} P^r L^r(Q^r)$. Note that $Q^r(P^r_0)$ is an excess risk minimizer for the data reduction $(\bar{V},O^r) \sim P^r_0$ when treating the coordinate transformations ${\bf Q}_n$ as fixed, similar to $Q_0$ for the full data $O \sim P_0$; we denote it as the oracle ensemble, $Q^r_{0, n} = Q^r(P^r_0)$.

\section{Meta-Level Learning Using Highly Adaptive Lasso} \label{section3}

\subsection{M-HAL-MLE}

For a given ensemble $Q^r\in {\cal Q}^r$, define the cross-validated empirical risk of corresponding $Q^r$-specific candidate estimator $Q^r\circ \hat{\bf Q}:{\cal M}_{np}\rightarrow {\cal Q}$  of $Q_0$ by
\[
P_n^{\bar{V}} L(Q^r\circ{\bf Q}_n)=\frac{1}{V}\sum_{v=1}^V P_{n,v}^1 L(Q^r\circ{\bf Q}_{n,v}).\]
Define the $C_n$-specific meta-level HAL-MLE (M-HAL-MLE) as the cross-validation selector among all ensembles of $J$ estimators, 
\[
Q_n^r=\argmin{Q^r\in {\cal Q}^r,\pl Q^r\pl_v^*<C_n}P_n^{\bar{V}}L(Q^r\circ{\bf Q}_n)=\argmin{Q^r\in {\cal Q}^r,\pl Q^r\pl_v^*<C_n}\frac{1}{V}\sum_{v=1}^V P_{n,v}^1 L(Q^r\circ{\bf Q}_{n,v}). \]
The oracle ensemble targeted by the M-HAL-MLE is given by 
\[
Q^r_{0,n}=\argmin{Q^r\in {\cal Q}^r}P_0^{\bar{V}} L(Q^r\circ {\bf Q}_n)=\argmin{Q^r\in {\cal Q}^r}\frac{1}{V}\sum_{v=1}^V P_0 L(Q^r\circ{\bf Q}_{n,v}).\]
Note that $Q^r_n$ is the empirical estimator of $Q^r_{0,n}$ defined by replacing $P_0^{\bar{V}}$ by its empirical counterpart $P_n^{\bar{V}}$, so that $Q^r_n$ is a regular HAL-MLE of $Q^r_{0,n}$.

\subsection{M-HAL-SL}

For each coordinate-transformation ${\bf Q}_{n,v}$, the M-HAL-MLE ensemble defines an estimator $Q_{n,v}\equiv Q_n^r\circ{\bf Q}_{n,v}$ of $Q_0$, $v=1,\ldots,V$. We also use the notation $Q_n\equiv Q_n^r \circ {\bf Q}_n$ for the function $Q_n(v,x)\equiv Q_{n,v}(x)$ which codes each of these $V$ estimates of $Q_0$.
The $C_n$-specific meta-level HAL super-learner (M-HAL-SL) of $Q_0$ refers to either $Q_n$ or to its average across splits given by \[
 \bar{Q}_n(x)\equiv P_n^{\bar V} Q_n(\bar V, x) =  \frac{1}{V}\sum_{v=1}^V Q_n(v,x).\]
Sometimes, to emphasize its dependence on the selector $C_n$, we denote these estimators with $Q_n^{C_n}$ and $\bar{Q}_n^{C_n}$, respectively.
Note that, if the parameter space ${\cal Q}$ is convex, $\bar{Q}_n\in {\cal Q}$.

Similarly, the oracle ensemble defines an oracle estimator, $Q_{0,n}\equiv Q^r_{0,n}\circ {\bf Q}_{n}$, so that $Q_{0,n}(v,x) = Q_{0, n, v}(x) \equiv Q^r_{0,n}\circ{\bf Q}_{n,v}(x)$,  $v=1,\ldots,V$. 
We will view this function $Q_{0,n}$ as a parameter of the distribution $P_0^{\bar{V}}$ of $(\bar{V}, O)$.
In addition, we use the notation $\bar{Q}_{0,n}(x)\equiv \frac{1}{V}\sum_{v=1}^V Q_{0,n}(v,x)$.
We note that $Q_n$ and $\bar{Q}_n$ are estimators of $Q_{0,n}$ and $\bar{Q}_{0,n}$, respectively.

The excess risk of the M-HAL-SL $Q_n$ is defined as \[
d_0^{\bar{V}}(Q_n,Q_0)=P_0^{\bar{V}}L(Q_n)-P_0^{\bar{V}}L(Q_0)=\frac{1}{V}\sum_{v=1}^V\{ P_0 L(Q_{n,v})-P_0L(Q_0)\}.\]
The excess risk of the M-HAL-SL $\bar{Q}_n$ is given by $d_0(\bar{Q}_n,Q_0)=P_0L(\bar{Q}_n)-P_0L(Q_0)$, and for convex loss functions $L(Q)$, we have \[
d_0(\bar{Q}_n,Q_0)\leq {d}_0^{\bar{V}}(Q_n,Q_0).\]
Therefore, it suffices to analyze the excess risk $d_0^{\bar{V}}(Q_n,Q_0)$ of the M-HAL-SL $Q_n$, which is decomposed as
\begin{eqnarray*}
d_0^{\bar{V}}(Q_n,Q_0)&=&P_0^{\bar{V}} L(Q_n^r\circ {\bf Q}_n)-P_0^{\bar{V}} L(Q_{0,n}^r\circ{\bf Q}_{n})\\
&&+P_0^{\bar{V}} L(Q_{0,n}^r\circ{\bf Q}_n)-P_0^{\bar{V}} L(Q_0)\\
&\equiv& d_0^{\bar{V}}(Q_n,Q_{0,n})+d_0^{\bar{V}}(Q_{0,n},Q_0).
\end{eqnarray*}
The first term, $d_0^{\bar{V}}(Q_n,Q_{0,n})$, involves comparing M-HAL-MLE $Q_n^r$ with the oracle ensemble  $Q_{0,n}^r$. Viewing it as a function of a given set of cross-fitted coordinate-transformations, we also denote this as $d_0(Q_n^r, Q_{0, n}^r)$, the loss-based dissimilarity of an HAL-MLE $Q_n^r$ of $Q_{0, n}^r$. 
The second term, $d_0^{\bar{V}}(Q_{0,n},Q_0)$, represents the dissimilarity between the oracle selected ensemble of   ${\bf Q}_n$ and the true function $Q_0$. 

\subsection{M-HAL-SL Plug-in Estimation}

We will see that for estimation of smooth features $\Psi(Q_0)$ of $Q_0$, one can either use the smooth feature of the average $\bar{Q}_n$,   $\Psi(\bar{Q}_n)$, or use the average across the  sample splits of $Q_{n,v}$, 
$$\Psi^{\bar V}(Q_n) \equiv \frac{1}{V}\sum_v \Psi(Q_{n,v}), $$
as the difference is second order.
Just as our decomposition above, we will also decompose $\Psi^{\bar{V}}(Q_n)-\Psi(Q_0)$ as the sum of the difference of the target feature of $Q_n$ and the oracle estimator $Q_{0,n}$, $\Psi^{\bar{V}}(Q_n)-\Psi^{\bar{V}}(Q_{0,n})$, and the difference of the target feature of the oracle estimator and true target estimand, $\Psi^{\bar{V}}(Q_{0,n})-\Psi(Q_0)$. Similarly, $\Psi(\bar{Q}_n)-\Psi(Q_0)=\Psi(\bar{Q}_n)-\Psi(\bar{Q}_{0,n})+\Psi(\bar{Q}_{0,n})-\Psi(Q_0)$. Both terms will be analyzed separately, where the ``bias'' (conditional on the training sample, it is truly a bias) of the target feature of the oracle estimator, $\Psi(Q_{0,n})-\Psi(Q_0)$, will be shown to be second order, or even zero when $\pmb Q_{n, v}$ are zero-loss coordinate-transformations.

\subsection{Equivalent Formulation of Statistical Parameters Using Reduced Data}

Given a data reduction $(\bar V, O^r) \sim P^r_0$ of $(\bar V, O) \sim P^{\bar V}_0$ as specified in Section \ref{sec:data_reduction}, M-HAL-MLE of oracle ensemble is regular HAL-MLE with the reduced data. Specifically, with the corresponding reduced data loss function $L^r: \mathcal{Q}^r \to L^2(P_0^r)$ that satisfies (\ref{eq:dimr_loss}), we have that the oracle ensemble 
\begin{align*}
    Q^r_{0, n} 
    = & \argmin{Q^r \in {\cal Q}^r} P_0^r L^r(Q^r)
\end{align*}
can be represented as a functional of the reduced data distribution $P^r_0$, and our M-HAL-MLE of $Q^r_{0, n}$, 
\begin{align*}
    Q_n^r
    = & \argmin{Q^r\in{\cal Q}^r,\pl Q^r\pl_v^*<C_n}P_n^r L^r(Q^r)
    , 
\end{align*}
is given by the HAL-MLE fitted with the reduced data $(v_i, O^r_i) \sim P_n^r$. 
This demonstrates that $Q_n^r$ can be implemented as a standard HAL-MLE based on the reduced data $(v_i,O^r_i)$, $i=1,\ldots,n$, and loss function $L^r(Q^r)$. Denote the loss-based dissimilarity with reduced data as $d_0^r(Q^r,Q_{0,n}^r) = P_0^r L^r(Q^r) - P_0^r L^r(Q_{0,n}^r)$.
In our model ${\cal M}^r$ for reduced data $(\bar{V},O^r)$ implied by the cross-fitted coordinate-transformation ${\bf Q}_n$, $\Psi^{\bar{V}}(Q^r \circ {\bf Q}_n)$ can be viewed as a parameter $\Psi^r:{\cal Q}^r\rightarrow\openr$ defined by 
$$\Psi^r(Q^r)=\frac{1}{V}\sum_{v=1}^V \Psi(Q^r\circ{\bf Q}_{n,v}).$$ 
So now we have both $d_0^{\bar{V}}(Q_n,Q_{0,n})=d_0^r(Q^r_n,Q_{0,n}^r)$ and $\Psi^{\bar{V}}(Q_n)=\Psi^r(Q^r_n)$. 
Therefore, the performance of the averaged plug-in estimator $\Psi^{\bar V}(Q_n)$ is decided by $Q_n^r$ as HAL-MLE of oracle ensemble $Q_{0,n}^r$, treating the coordinate-transformation ${\bf Q}_n$ as fixed.

Let $D^r(P^r)(\bar{V},O^r)$ be the canonical gradient of $\Psi^r$ at $P^r$. 
Let $G^r(P^r) \in {\cal G}^r$ be a nuisance parameter such that $D^r(P^r)=D^r(Q^r,G^r)$. 
We assume that $G^r(\cdot)$ given $\bar V = v$ defines an element $G_v^r$ in $\cal G$, the parameter space for $G$. 
For example, if $G$ is a function of $X$ and $G^r$ is a function of ${\bf Q}_n(\bar V, X)$, then conditional on $\bar V = v$ we may define $G_v^r(x) = G^r({\bf Q}_{n, v}(x)) \in \mathcal{G}$; although in practice the definition is flexible depending on $G^r$ and $G$. 
We assume the following link between $D^r(P^r)$ and $D^*(P)$:  
\begin{equation}\label{linkDrDstar}
D^r(Q^r,G^r)(v,O^r(v,O))=D^*(Q^r\circ{\bf Q}_{n,v},G^r_v)(O).
\end{equation}
Condition (\ref{linkDrDstar}) essentially states that for a fixed $v$, the reduced data structure $O^r_v$ has the same structure as $O$, 
and as a result $P^r_v$ has same model structure as $P$, 
so that the pathwise derivative of $Q^r\rightarrow \Psi(Q^r\circ{\bf Q}_{n,v})$ has the same structure as $Q\rightarrow \Psi(Q)$.  For example, the general formula for the canonical gradient of $EE(Y\mid A=1,W)$ with a nonparametric model remains the same in terms of the (conditional) densities of $Y$, $A$, and $W$, regardless of the dimension or definition of $W$. 
Similarly, the canonical gradient of a treatment specific mean $EY_{\bar{a}}$  for general longitudinal data structure $O=(L(0),A(0),\ldots,L(K),A(K),Y=L(K+1))$ has the same general form in terms of all the conditional, densities regardless of the dimension of definition of $L(k)$.

Recall that the true $Q^r_{0,n} = Q^r(P_0^r)$ and $G^r_{0,n} \equiv G^r(P_0^r)$ are indexed by a  subscript $n$ to emphasize dependence on coordinate-transformation ${\bf Q}_n$. Let $L_1^r(G^r)$ be a loss function for $G_{0,n}^r=\arg\min_{G^r\in {\cal G}^r}P_0^rL_1^r(G^r)$ and let $d_{01}^r(G^r, G^r_{0, n}) = P_0^r L_1^r(G^r) - P_0^r L_1^r(G^r_{0, n})$, where ${\cal G}^r$ is the parameter space for $G^r$. Let $R_{20}^r(Q^r,G^r,Q_{0,n}^r,G_{0,n}^r)=\Psi^r(Q^r)-\Psi^r(Q_{0,n}^r)+P_0^r D^r(Q^r,G^r)$ be the exact second order remainder. By (\ref{linkDrDstar}) it follows that
\[
R_{20}^r(Q^r,G^r,Q_{0,n}^r,G_{0,n}^r)=\frac{1}{V}\sum_{v=1}^V R_{20}(Q^r\circ{\bf Q}_{n,v},G^r_v,Q_{0,n}^r\circ{\bf Q}_{n,v},G^r_{0,n,v}).\]

\subsection{Cross-Validation Selector of the Sectional Variaition Norm}\label{sec:selector}
Note that M-HAL-MLE and M-HAL-SL are fitted as functions of $P_n$ across a $V$-fold sample splitting $(P_{n, v}, P_{n, v}^1: v = 1, \dots, V)$ of $P_n$, and are indexed by a hyperparameter $C$ for the sectional variation norm bound enforced on the ensembles $Q^r\in {\cal Q}^r$. Therefore, we can denote the algorithms by $Q_n^{r, C}: {\cal M}_{np}\rightarrow {\cal Q}^r$ and $\bar Q_n^C: {\cal M}_{np}\rightarrow {\cal Q}$, such that 
$$\bar Q_n^C(P_n) = \frac{1}{V}\sum_{v=1}^V Q^{r, C}_n(P_n) \circ {\bf Q}_n(P_{n, v}).$$ 
We can define the cross-validation selector of $C$ by the performance of M-HAL-SL, 
\[
 C_{n,cv}=\argmin{C} \frac{1}{V}\sum_{v=1}^V P_{n,v}^1 L(\bar Q^C_n(P_{n,v}) ).\] 
Note that this involves double cross-validation similar to the cross-validated risk of a regular super-learner. This cross-validation selector $C_{n,cv}$ is asymptotically equivalent with the oracle selector of $C$ optimizing excess risk. 

If we fix the realized coordinate-transformations ${\bf Q}_{n,v}$, $v=1,\ldots,V$, then the M-HAL-SL can be defined as an algorithm using fixed functions, independent of what data is provided. 
This defines an approximation of the M-HAL-SL as 
$$\bar Q_n^{C, fast}(P_n) = \frac{1}{V}\sum_{v=1}^V Q^{r, C}_n(P_n) \circ {\bf Q}_{n, v}.$$ 
The cross-validation selector for this algorithm is then given by:
   \[
   C_{n,cv}^{fast}=\argmin{C} \frac{1}{V}\sum_{v=1}^V P_{n,v}^1 L({\bar{Q}}_n^{C,fast}(P_{n,v}) ).
   \]
This algorithm still requires $V$ times running the HAL-MLE ${Q}^r_n$ at the meta-level on the training samples $P_{n, v}$,
but it does not require rerunning the estimators $\hat{Q}_j$, $j=1,\ldots,J$. 
This could be used as a fast approximation of the double cross-validation selector with the risk of slightly overfitting the univariate hyperparameter $C$. We believe this criterion will still provide a good ranking and thereby selector $C_{n,cv}^{fast}$, especially as an initial value of an undersmoothing selector for the purpose of plug-in estimations. 

Lastly, one can directly optimize the performance of the ensemble function $Q^r\in \mathcal{Q}^r$, rather than that of the resulting M-HAL-SL. While this is less similar to a classical super-learner, it enables an even faster variation-bound selector. Specifically, when  the realized coordinate-transformations ${\bf Q}_{n,v}$, $v=1,\ldots,V$ are fixed, this selector is equivalent to a standard cross-validated $\ell_1$-norm selector (such as \texttt{glmnet::cv.glmnet}), according to the equivalent formulation with meta-level reduced data $(v_i, O^r_i), i = 1, \dots, n$. We refer to this as the fast M-HAL-MLE-based selector, denoted as $\tilde C_{n,cv}^{fast}$, with the corresponding M-HAL-SL denoted as $\bar Q_n^{\tilde C_{n,cv}^{fast}, fast}(P_n)$. 
Similarly, after the final refitting with the full data and the selected bound,  the M-HAL-SL estimators constructed by the two fast selectors and the honest selector (slower double cross-validated) only differ in the choice of the univariate hyperparameter $C$. In practice, we expect this difference to lead to only modest overfitting and similar overall performance.

\subsection{Summary of Assumptions}

We assume the analogue of $M_{20}<\infty$ for loss $L(Q)$ for the reduced data loss $L^r(Q^r)$. 
For that purpose, define $Q^r_{0,{\bf Q}}=\arg\min_{Q^r\in {\cal Q}^r}P_0 L(Q^r\circ{\bf Q})$, and let
\begin{equation}\label{M2r}
M_2^r\equiv
\sup_{{\bf Q}\in {\cal Q}^J}\sup_{Q^r\in {\cal Q}^r}\frac{P_0^r\{ L(Q^r\circ{\bf Q})-L(Q^r_{0,{\bf Q}}\circ{\bf Q})\}^2}{P_0^r\{ L(Q^r\circ{\bf Q})-L(Q^r_{0,{\bf Q}}\circ {\bf Q})\}}<\infty.\end{equation}
We also need that various classes functions of $O^r_v$ are contained in the class of cadlag functions with bound on sectional variation norm. Therefore, it is convenient to let ${\cal D}_{d^r}[0,\tau^r]$ represent a class of $d^r$-variate real valued cadlag functions on a cube $[0,\tau^r]$ with a universal bound (also over realizations of $\{{\bf Q}_{n, v}: v\}$) on the sectional variation norm.
Let $d_0^r((Q^r,G^r),(Q_{0,n}^r,G_{0,n}^r))\equiv d_0^r(Q^r,Q_{0,n}^r)+d_{01}^r(G^r,G_{0,n}^r))$ be the loss-based dissimilarity for the joint $(Q^r,G^r)$. We make the usual assumption that the exact second order remainders can be bounded in terms of this loss-based dissimilarity. 

The summarize the key assumptions throughout this article as follows
\begin{eqnarray}
D^r(Q^r,G^r)(v,O^r_v=O^r(v,O))&=&D^*(Q^r\circ{\bf Q}_{n,v},G^r_v)(O)\label{r1}\\
M_2^r&<&\infty \nonumber \\
\{L^r(Q^r):Q^r\in {\cal Q}^r\}&\subset& {\cal D}_{d^r}[0,\tau^r]\nonumber\\
\{D^r(Q^r,G_{0,n}^r):Q^r\in {\cal Q}^r\}&\subset&  {\cal D}_{d^r}[0,\tau^r]\nonumber \\
R_{20}(Q,G,Q_0,G)&=&O(d_0(Q,Q_0))\nonumber \\
P_0\{D^*(Q,G)-D^*(Q_0,G)\}^2&=&O(d_0(Q,Q_0))\nonumber \\
R_{20}^r(Q^r,G^r,Q_{0,n}^r,G_{0,n}^r)&=&O(d_0^r((Q^r,G^r),(Q^r_{0,n},G^r_{0,n})) )\nonumber \\
P_0^r\{D^r(Q^r,G^r)-D^r(Q_{0,n}^r,G_{0,n}^r)\}^2&=&O(d_0^r((Q^r,G^r),(Q_{0,n}^r,G_{0,n}^r)) )\nonumber \\
\sup_{\{{\bf Q}_{n,v}: v\}, v} P_0 \{D^*(Q_{0,n,v},G_{0,n,v}^r) - D^*(Q_0, \tilde G_0)\}^2  &\rightarrow_p& 0 \nonumber
\end{eqnarray}
for a limit $\tilde{G}_0\in {\cal G}$  of $G^r_{0,n,v}$, $v=1,\ldots,V$.
The last assumption is a universal consistency condition for the asymptotic normality of $n^{1/2}P_n^r D^r(Q_{n}^r,G_{0,n}^r)$. 
The other assumptions in general hold true once we enforce strong positivity so that $D^*(P_0)$ is a uniformly bounded function on a support of $O$. 
We will refer to this whole set of assumptions as assumption (\ref{r1}).

\section{Convergence Rate of M-HAL-SL} \label{sec:rate}

Consider $d_0^{\bar{V}}(Q_n,Q_0)=d_0^{\bar{V}}(Q_n,Q_{0,n})+d_0^{\bar{V}}(Q_{0,n},Q_0)$.
By assuming that $L(Q)$ is a convex loss function so that $P_0L(\sum_j \alpha_j Q_j)\leq \sum_j \alpha_j P_0L(Q_j)$ for $\alpha$-vectors with $\alpha_j\geq 0$, $\sum_j \alpha_j=1$, 
we have $d_0(\bar{Q}_n,Q_0)\leq d_0^{\bar{V}}(Q_n,Q_0)$, so that our results imply the same rate result for $d_0(\bar{Q}_n,Q_0)=P_0L(\bar{Q}_n)-P_0L(Q_0)$. 

The following lemma establishes the rate of convergence result for $d_0^{\bar{V}}(Q_n,Q_{0,n})$ (proof in Appendix \ref{sec:proofs_rate}). 
\begin{lemma}\label{lemmarateQnr}
Recall assumption (\ref{r1}).
We have
\[
d_0^{\bar{V}}(Q_n,Q_{0,n})=O_P(n^{-2/3}(\log n)^{d^r}).\]
\end{lemma}

This yields the following result for $d_0^{\bar{V}}(Q_n,Q_0)$ and thereby  for $d_0(\bar{Q}_n,{Q}_0)$. 
\begin{theorem} \label{theoremrateQn}
Recall assumption (\ref{r1}). We have
\begin{eqnarray*}
d_0^{\bar{V}}(Q_n,Q_0)
&=& d_0^{\bar{V}}(Q_{0,n},Q_0)+d_0^{\bar{V}}(Q_n,Q_{0,n})\\
&=&\min_{Q^r\in {\cal Q}^r}\frac{1}{V}\sum_{v=1}^VP_0\{L(Q^r\circ{\bf Q}_{n,v})-L(Q_0)\}
+O_P(n^{-2/3}(\log n)^{d^r}). 
\end{eqnarray*}
If $\hat{\bf Q}=(\hat{Q}_j: j)$ includes an estimator $\hat{Q}_j$ such that $d_0(\hat{Q}_j(P_{n,v}),Q_0)=O_P(n^{-2/3}(\log n)^d)$, then it follows that
\[
d_0^{\bar{V}}({Q}_n,Q_0)=O_P(n^{-2/3}(\log n)^{\max\{d, d^r\}}).\]
\end{theorem}
The leading term in $d_0^{\bar{V}}({Q}_n,Q_0)$ represents the dissimilarity between the oracle estimator $Q^r_{0,n}\circ{\bf Q}_{n,v}$ and $Q_0$. 
Since ${\cal Q}^r$ includes the functions $f(x)=x_j$, $j=1,\ldots,J$, this leading term can be bounded by \[\min_j \frac{1}{V}\sum_{v=1}^V P_0\{L(\hat{Q}_j(P_{n,v}))-L(Q_0)\}.\] 
However, note that the oracle estimator is generally a much better estimator than one of the candidates in the library of $J$ estimators.
In fact, $d_0^{\bar{V}}(Q_{0,n},Q_0)$ even equals zero for many coordinate-transformations.

\section{Asymptotic Linearity of Target Features of Undersmoothed M-HAL-MLE}\label{section5}

We can estimate $\Psi(Q_0)$ with $\Psi(\bar{Q}_n)$ or $\frac{1}{V}\sum_{v=1}^V\Psi(Q_{n,v})$. 
Under regularity conditions, the Taylor expansion at $\bar{Q}_n$ gives that the difference between these two plug-in estimators will generally be second order
\begin{eqnarray*}
\frac{1}{V}\sum_{v=1}^V \Psi(Q_{n,v})-\Psi(\bar{Q}_n)&=&\frac{1}{V}\sum_{v=1}^V d\Psi(\bar{Q}_n)(Q_{n,v}-\bar{Q}_n)+O_P(d_0^{\bar{V}}(Q_n, Q_0))\\
&=&0+ O_P(d_0^{\bar{V}}(Q_n, Q_0)) = O_P(d_0^{\bar{V}}(Q_n, Q_0)),
\end{eqnarray*} 
where $d\Psi(\bar{Q}_n)(h)=\left . \frac{d}{d\epsilon}\Psi(\bar{Q}_n+\epsilon h)\right |_{\epsilon =0}$ is the directional derivative of $\Psi$ at $\bar{Q}_n$ in direction $h$, and $O_P(\norm{Q_{n, v} - \bar Q_n}_{P_0}^2) = O_P(\norm{\bar Q_n - Q_0}_{P_0}^2 + \norm{Q_{n, v} - Q_0}_{P_0}^2)= O_P(d_0^{\bar{V}}(Q_n, Q_0))$ under mild assumptions (Section \ref{sec:rate}). Note, the first terms equals zero since $\bar{Q}_n=\frac{1}{V}\sum_v Q_{n,v}$. Theorem \ref{theoremrateQn} establishes that $d_0^{\bar{V}}(Q_n, Q_0)=O_P(n^{-2/3}(\log n)^d)$ under reasonable conditions, so that this will indeed be $o_P(n^{-1/2})$.
Therefore, it suffices to analyze the target feature $\Psi^r(Q_n^r)=\frac{1}{V}\sum_{v=1}^V \Psi(Q_{n,v})$  of the undersmoothed M-HAL-MLE $Q_n^r$.

Furthermore, we have 
\begin{eqnarray*}
\Psi^r(Q_n^r)-\Psi(Q_0)=\Psi^r(Q_n^r)-\Psi^r(Q^r_{0,n})+ (\Psi^r(Q^r_{0,n})-\Psi(Q_0)).
\end{eqnarray*}
The second term represents a bias term in our reduced data model  that treats the cross-fitted transformation ${\bf Q}_n$ as fixed. If the coordinate-transformation ${\bf Q}_n$ is zero loss, then we would have that $Q_{0,n,v}=Q_0$ for each $v$, so that $\Psi^r(Q_{0,n}^r)=\Psi(Q_0)$. 
There also exist examples of reductions ${\bf Q}_n$  for which $Q_{0,n,v}\not =Q_0$, but nonetheless $\Psi^r(Q_{0,n}^r)=\Psi(Q_0)$ (Section \ref{sec:TSM}). 
In Section \ref{sec:difference} we will generally establish that this bias term $\Psi^r(Q^r_{0,n})-\Psi(Q_0)$ is  second order, due to either the cross-fitted transformation ${\bf Q}_n$ being zero-loss w.r.t. $\Psi(Q_0)$, or due to ensembles of ${\bf Q}_n$ being $n^{-1/4}$-consistent estimators of $Q_0$.  This condition will not require $C_n$ to undersmooth.
The first estimation term, $\Psi^r(Q_n^r)-\Psi^r(Q^r_{0,n})$, will be analyzed in Section \ref{sec:theorem}. 

\subsection{Asymptotic Linearity Theorem}\label{sec:theorem}

Under an undersmoothing selector $C_n>C_{n,cv}$ chosen large enough  \citep{vanderLaan&Benkeser&Cai19}, we have (see Appendix \ref{sec:undersmoothing})
\begin{equation}\label{effscoreeqnr}
P_n^r D^r(Q_n^r,G_{0,n}^r)=o_P(n^{-1/2}).
\end{equation}
 Once this efficient score equation (\ref{effscoreeqnr}) is solved, then we obtain 
\[
\Psi^r(Q_n^r)-\Psi^r(Q^r_{0,n})=(P_n^r-P_0^r)D^r(Q_n^r,G_{0,n}^r)+R_2^r(Q_n^r,G_{0,n}^r,Q^r_{0,n},G_{0,n}^r).\]
This can be represented as
\begin{eqnarray*}
\Psi^r(Q_n^r)-\Psi^r(Q^r_{0,n})&=&\frac{1}{V}\sum_{v=1}^V (P_{n,v}^1-P_0)D^*(Q_n^r\circ{\bf Q}_{n,v},G_{0,n, v}^r)\\
&& +\frac{1}{V}\sum_{v=1}^V R_{20}(Q_n^r\circ{\bf Q}_{n,v},G_{0,n, v}^r,Q^r_{0,n}\circ{\bf Q}_{n,v},G_{0,n, v}^r).\end{eqnarray*}
By assumption (\ref{r1}),  $d_0^{\bar{V}}(Q_n,Q_{0,n})=O_P(n^{-2/3}(\log n)^{d^r})$ implies that the second order remainder is $O_P(n^{-2/3}(\log n)^{d^r})$. The empirical process term will be controlled in the following asymptotic linearity theorem (proof in Appendix \ref{sec:proofs_AL}).  

Note that conditional on $P_{n, v}$ or for fixed $\pmb Q_{n, v}$, the Donsker class condition over ${\cal D}^* = \{D^*(Q, G): Q\in\mathcal{Q}, G\in\mathcal{G}\}$, typically required for the original data problem,  is avoided. Instead, it suffices to assume a Donsker class condition driven by $Q_n^r$ only, which is satisfied if the sectional variation norms in $\mathcal{Q}^r, \mathcal{G}^r$ are universally bounded with probability tending to $1$. With the reduced data dimensions, this meta-level regularity condition is easier to hold, especially when the original data problem is complex and constructs highly varying initial estimators. Moreover, (\ref{effscoreeqnr}) may be satisfied for not only one specific target, in which case the following theorem applies for arbitrary $\Psi$ in a large class of smooth features. 

\begin{theorem}\label{theoremaslinmhalmle}
Recall assumption (\ref{r1}).
Assume $C_n$ is chosen large enough so that
\[
\frac{1}{V}\sum_{v=1}^V P_{n,v}^1 D^*(Q_n^r\circ{\bf Q}_{n,v},G^r_{0,n,v})=o_P(n^{-1/2}).\]
 Then, under conditions 1-5 of Lemma \ref{lemma:AL1} and \ref{lemma:AL2},
 \begin{eqnarray*}
\Psi^r(Q_n^r)-\Psi^r(Q_{0,n}^r)&=& \frac{1}{V}\sum_{v=1}^V (P_{n,v}^1-P_0)D^*(Q_{0,n,v},G_{0,n,v}^r) + O_P(n^{-2/3}(\log n)^{d^r})\\
&=&P_n D^*(Q_{0},\tilde{G}_{0})+o_P(n^{-1/2}).
\end{eqnarray*}
 \end{theorem}

Combined with Section \ref{sec:difference}, we conclude that the target feature of M-HAL-MLE,  $\Psi^r(Q^r_n)$, is root-$n$-consistent for the true target feature $\Psi(Q_0)$, it has known influence curve  conditional on training samples (so that variance estimation follows), and it is asymptotically normally distributed, without any Donsker class assumption on  ${\cal D}^*$.
In addition, the target feature of M-HAL-SL  is  an asymptotically linear estimator of $\Psi(Q_0)$ with influence curve $D^*(Q_0,\tilde{G}_0)$. 
For zero-loss transformations ${\bf Q}_n$, and certain types of reductions ${\bf Q}_n$ under which $G_{0, n, v}^r$ converges to true $G_0$, then we will have that $D^*(Q_0,\tilde{G}_0)=D^*(Q_0,G_0)$, in which case $\Psi^r(Q_n^r)$ behaves as an asymptotically efficient estimator of $\Psi(Q_0)$. If $\tilde{G}_0\not =G_0$, then $\Psi^r(Q_n^r)$ will  typically  end up being super-efficient.

\subsection{Difference Between Target Feature of Oracle Estimator and Target Estimand}\label{sec:difference}

To establish the asymptotic linearity for the fixed parameter $\Psi(Q_0)$, it remains to establish that $\frac{1}{V}\sum_v \Psi(Q^r_{0,n}\circ{\bf Q}_{n,v})-\Psi(Q_0)
=o_P(n^{-1/2})$. The following two theorems address the scenarios with or without the nuisance parameter (see proofs in Appendix \ref{sec:proofs_diff}).  

The first result applies to the case in which $D^*(P) = D^*(Q)$ so that there is no nuisance parameter $G$. 
\begin{theorem}\label{theoremoracleensemble1}
Assume (\ref{r1}).
Suppose that $D^*(P)=D^*(Q(P))$. We have
\[
\Psi^r(Q_{0,n})-\Psi(Q_0)=\frac{1}{V}\sum_{v=1}^V R_{20}(Q_{0,n,v},Q_0).\]
By (\ref{r1}), the latter is bounded by $O(d_0(Q_{0,n},Q_0))$. Thus, $\Psi^r(Q_{0,n})-\Psi(Q_0)=O(d_0(Q_{0,n},Q_0))$. 
\end{theorem}
Even if there exists a nuisance parameter, one could redefine $Q$ as a joint parameter $Q = (Q^s, G)$ including both the sufficient $Q^s$ and the nuisance parameter $G$, such that $\Psi(Q) = \Psi(Q^s, G) = \Psi(Q^s)$ and $D^*(P) = D^*(Q)$. This strategy simplifies the conditions required for $G_{0, n}^r$ but involves M-HAL-SL of both $Q_0^s$ and $G_0$. 

The following general theorem handles the nuisance parameter using an approximation of $G_{0, n}^r$. In practice, this approximation can be chosen such that the residual $r_n$ is also bounded by $(d_0^{\bar V}(Q_{0, n}, Q_0))^{1/2}$ (Section \ref{sec:TSM}). 

\begin{theorem}\label{theoremoracleensemble} \ \newline
{\bf Definitions:}
Let $G_{0,n,v}^{*}\in {\cal G}$ be a functional that approximates $G_{0,n,v}\equiv G^r_{0,n,v}$ and for which 
\[
\frac{1}{V}\sum_v P_0D^*(Q_{0,n}^r\circ{\bf Q}_{n,v},G_0)=\frac{1}{V}\sum_v
P_0D^*(Q_{0,n}^r\circ{\bf Q}_{n,v},G_{0,n,v}^{*}),\]
or equivalently,
\[R_{20}(Q_{0,n}^r\circ{\bf Q}_{n,v},G_{0,n,v}^{*},Q_0,G_0)=R_{20}(Q^r_{0,n}\circ{\bf Q}_{n,v},G_0,Q_0,G_0).\]
Let \[
r_n\equiv \frac{1}{V}\sum_v \{R_{20}(Q_{0,n}^r\circ{\bf Q}_{n,v},G_{0,n,v},Q_0,G_0)-R_{20}(Q_{0,n}^r\circ{\bf Q}_{n,v},G_{0,n,v}^{*},Q_0,G_0)\}.\] We note that this represents a second order term that generally can be bounded in terms of squares or products of $d_0^{\bar{V}}(Q_{0,n},Q_0)^{1/2}$ and a norm $\pl G_{0,n,v}^{*}-G_{0,n,v}\pl$ (such as $L^2(P_0)$-norm).
\nl
{\bf Conclusion:}
We have
\[\left\{ \frac{1}{V}\sum_v \Psi(Q_{0,n}^r\circ{\bf Q}_{n,v})-\Psi(Q_0)\right\}=r_n+O(d_0^{\bar{V}}(Q_{0,n},Q_0)).\]
\end{theorem}

\section{Treatment Specific Mean Example} \label{sec:TSM}

In this section, we will go through the definitions and conditions of the theorems in the context of a concrete example of treatment specific means. 

Let $O=(W,A,Y)\sim P_0\in \mathcal{M}$ be a vector random variable in which $W$ are baseline covariates, $A\in \{0,1\}$ is a binary treatment, and $Y\in [0,1]$ a bounded continuous outcome. 
Suppose that we observe $n$ i.i.d. copies $O_1,\ldots,O_n$ of $O$. 
Let $d$ be the dimension of $O$, and assume $O$ is a bounded random variable. 

For all $P \in \mathcal{M}$, let $Q(P)=E_P(Y\mid A=1,W)$ be the functional parameter of interest, and let $G(P)=E_P(A\mid W)$. 
Let the statistical model be given by ${\cal M}=\{P: G(P)>\delta>0\mbox{ for some $\delta>0$}\}$, thereby enforcing a positivity assumption.
Then $Q:{\cal M}\rightarrow {\cal Q}=\{Q(P): P\in {\cal Q}\}$, where ${\cal Q}$ is its parameter space. Note that each realization is a $k=d-2$-dimensional real valued measurable function of $W$. 
For $Q \in {\cal Q}$, we can choose the squared error loss function $L(Q)(O)=A(Y-Q(W))^2$, so that $Q_0= Q(P_0) = \arg\min_{Q}P_0L(Q)$.
Note that the loss-based dissimilarity is given by $d_0(Q,Q_0)=P_0L(Q)-P_0L(Q_0)=
P_0 G_0(Q-Q_0)^2$, and is thus a square of a  weighted $L^2$-norm. Since $G$ is bounded away from zero, this loss-based dissimilarity is equivalent with $\pl Q-Q_0\pl_{P_0}^2$, where $\pl Q-Q_0\pl_{P_0}=\sqrt{P_0 (Q-Q_0)^2}$.

We will define $\Psi(P)=P_0 Q(P)$ as the target parameter, so that the treatment specific mean is given by $P_0Q_0= E_0E_0(Y\mid A=1,W)$ at $P = P_0$. We also denote $\Psi(P)$ with $\Psi(Q)$. 
The canonical gradient of $\Psi(P)$ at $P$ is given by $D^*(G,Q)=A/G(W)(Y-Q(W))$ and the exact second order remainder $R_{20}(Q,G,Q_0,G_0)=\Psi(Q)-\Psi(Q_0)+P_0D^*(G,Q)$ is given by 
$R_{20}(Q,G,Q_0,G_0)=P_0 (G-G_0)/G (Q-Q_0)$ (e.g., \citep{vanderLaan&Rose11}).

Let ${\bf Q}_{n,v}=\hat{\bf Q}(P_{n,v})$ be a collection of $J$ estimators of $Q_0$ based on training sample $P_{n,v}$, $v=1,\ldots,V$, which can also be viewed more generally as a $J$-dimensional data-adaptive transformation ${\bf Q}_{n,v}(W)$ of the $k$-dimensional $W$. 
Recall ${\bf Q}_n(v,W)={\bf Q}_{n,v}(W)$. 
Let ${\cal Q}^r$ be a class of $J$-variate real valued cadlag functions with a universal bound $C^u$ on its sectional variation norm.
For a  given cadlag function (also called ensemble) $Q^r\in {\cal Q}^r$, we can define the composition $Q^r\circ{\bf Q}_n$ by $Q^r\circ{\bf Q}_n(v,W) = Q^r\circ{\bf Q}_{n, v} (W)=Q^r({\bf Q}_{n,v}(W))$.

\subsection{Reduced data estimation problem treating ${\bf Q}_n$ as fixed}

Treating ${\bf Q}_n$ as fixed, we can reduce the observed data $(\bar{V},O)$ to $(\bar{V},O^r=(W^r,A,Y))$, where $W^r\equiv {\bf Q}_{n,\bar{V}}(W)= \sum_{v = 1}^V \indicator{\bar V = v} {\bf Q}_{n,v}(W)$. 
Define $W^r_v\equiv{\bf Q}_{n,v}(W)$. 
Let $d^r=J+2$ be the dimension of the reduced data $O^r$. 
Recall $(\bar{V},O^r)$ follows the joint distribution  $P_0^r \in \mathcal{M}^r$, where $\bar{V}$ is uniform $\{1,\ldots,V\}$, and $O^r$ given $\bar{V}=v$ follows the distribution of $(W^r_v,A,Y) \sim P_{0, v}^r$ which is implied by $({\bf Q}_{n, v}(W),A,Y)$ under $P_0$. 
For all $P^r \in \mathcal{M}^r$, if $(\bar V, O^r)\sim P^r$, then there exists $P \in \mathcal{M}$ such that $(W^r_v,A,Y) \sim P_{v}^r$ follows the distribution implied by $({\bf Q}_{n, v}(W),A,Y)$ under $O \sim P$.

Define the reduced data loss as $L^r(Q^r)(\bar{V},O^r)=A(Y-Q^r(W^r))^2$. 
This satisfies condition (\ref{eq:dimr_loss}): $L(Q^r\circ{\bf Q}_{n,v})(O)= A(Y - Q^r({\bf Q}_{n, v}(W)))^2 =L^r(Q^r)(v,O^r(v,O))$, which depends on $O$ only through $({\bf Q}_{n, v}(W), A, Y) = O^r(v,O)$.

 Let ${\cal M}^r=
\{P^r(P): P\in {\cal M}, G^r(P^r) > \delta > 0, W^r \mapsto E_{P^r(P)}(Y | A = 1, W^r) \in {\cal Q}^r\}$ be the model for the distribution $P_0^r$ of $(\bar{V},O^r)$ implied by ${\cal M}$. 
Define the functional parameter $Q^r:{\cal M}^r\rightarrow {\cal Q}^r$ by  $Q^r(P^r)\equiv \arg\min_{Q^r\in {\cal Q}^r}P^r L^r(Q^r)$ which equals $E_{P^r}(Y | A = 1, W^r)$ due to the assumption on ${\cal M}^r$. 
$Q^r_{0,n} = Q^r(P^r_0)$ represents the optimal ensemble.
Let $d^r_0(Q^r,Q^r_{0,n})=P_0^rL^r(Q^r)-P_0^rL^r(Q^r_{0,n})$ be the loss-based dissimilarity. 
Let $G_{0,n}^r(W^r)=E_0(A\mid W^r)$. 
Note also that, since $G_0>\delta>0$, we also have $G_{0,n}^r>\delta>0$.
The loss-based dissimilarity is given by
\begin{eqnarray*}
d^r(Q^r,Q^r_{0,n})&=&P_0^r G_{0,n}^r (Q^r-Q^r_{0,n})^2=E_0 G_{0,n}^r(W^r)(Q^r-Q^r_{0,n})^2(W^r)\\
&=&\frac{1}{V}\sum_{v=1}^V E_0 G_{0,n}^r(W^r_v)(Q^r-Q^r_{0,n})^2(W^r_v).
\end{eqnarray*}

We define  $\Psi^r:{\cal M}^r\rightarrow\openr$  by $\Psi^r(P^r)=\Psi^r(Q^r(P^r))= E_{P_0^r}E_{P^r}(Y\mid A=1,W^r)$.  It can be verified that $\Psi^r(P^r) = \frac{1}{V}\sum_{v = 1}^V E_{P_{0}^r}( Q^r(W^r)|\bar V = v) 
= \frac{1}{V}\sum_{v = 1}^V E_{P_{0, v}^r}( Q^r(W^r_v))
= \frac{1}{V}\sum_{v = 1}^V E_{P_{0}}Q^r({\bf Q}_{n, v}(W))
= \frac{1}{V}\sum_{v = 1}^V \Psi(Q^r \circ {\bf Q}_{n, v}) $.
Note that, as a special case, $\Psi^r(Q_{0,n}^r)=\Psi(Q_0)$ if, for each $v$, $W^r_v$ is such that $A$, given $W$, only depends on $W^r_v={\bf Q}_{n,v}(W)$. 
  The canonical gradient of $\Psi^r$ at $P^r$ is given by:
 \[
 D^r(P^r)(\bar{V},O^r)=\frac{A}{E(A\mid W^r)}(Y-E_{P^r}(Y\mid A=1,W^r)),\]
 or
 \[
 D^r(Q^r,G^r)(\bar{V},O^r)=\frac{A}{G^r(W^r)}(Y-Q^r(W^r)).\]
 Condition (\ref{linkDrDstar}) holds since $D^r(Q^r,G^r)(v,O^r_v(o))=D^*(Q^r({\bf Q}_{n,v}),G^r_v)(o)$, where $G^r_v(W) \equiv G^r({\bf Q}_{n,v}(W))$.
  We have \[
 P_0^r D^r(Q^r,G^r)=\frac{1}{V}\sum_{v=1}^V P_0 D^*(Q^r({\bf Q}_{n,v}),G^r_v);\]
 and $R_2^r(P^r,P_0^r)=P_0^r (G^r-G_{0,n}^r)/G^r (Q^r-Q^r_{0,n})$. It follows that
 \[R_{20}^r(Q^r,G^r,Q^r_{0,n},G_{0,n}^r)=\frac{1}{V}\sum_{v=1}^V P_0 (G^r-G_{0,n}^r)/G^r ({\bf Q}_{n,v})
  (Q^r({\bf Q}_{n,v})-Q^r_{0,n}({\bf Q}_{n,v}).\]
 Thus,
  \[
  R_{20}^r(Q^r,G^r,Q^r_{0,n},G_{0,n}^r)=\frac{1}{V}\sum_{v=1}^V R_{20}(Q^r({\bf Q}_{n,v}),G^r_v,Q^r_{0,n}({\bf Q}_{n,v}),G^r_{0,n,v}) .\]
  We also have the equivalences $d_0^r(Q_n^r,Q_{0,n}^r)=d_0^{\bar{V}}(Q_n,Q_{0,n}) 
  $.

\subsection{Convergence of M-HAL-MLE}

We have that the M-HAL-MLE of the oracle ensemble $Q_{0,n}^r$ is given by
$Q_n^r=\arg\min_{Q^r\in {\cal Q}^r, \pl Q^r\pl_v^*<C_n}P_n^r L^r(Q^r)$.
This is just a regular HAL-MLE of $E(Y\mid A=1,W^r)$ based on the reduced data set $O^r_i=(W^r_i,A_i,Y_i)$, $i=1,\ldots,n$, where $W^r_i={\bf Q}_{n,v_i}(W_i)$. It corresponds with a linear least squares regression under an $L_1$-constraint $\pl \beta^r\pl_1<C_n$, and it results in a fit $Q_n^r=\sum_{s,j}\beta_n^r(s,j)\phi_{s,j}$ for a rich collection of spline basis functions. 
Given $Q_n^r$, we can compute the M-HAL-SL $Q_n$ (collection of $V$ estimators) by $Q_n(v,\cdot)=Q_n^r\circ{\bf Q}_{n,v}$, and corresponding average $\bar{Q}_n(\cdot)=\frac{1}{V}\sum_{v=1}^V Q_n(v,\cdot)$.
The target of $Q_n$ is the oracle estimator
defined by $Q_{0,n}(v,\cdot)=Q_{0,n}^r\circ{\bf Q}_{n,v}$, and the target of $\bar{Q}_n$ is accordingly given by $\bar{Q}_{0,n}=\frac{1}{V}\sum_{v=1}^V Q_{0,n}^r\circ{\bf Q}_{n,v}$.
Note that $Q_{0,n}^r\circ{\bf Q}_{n,v} (w)= E_0(Y\mid A=1,W^r_v={\bf Q}_{n,v}(w))$.

Due to $O$ being bounded, and ${\cal Q}$ being bounded functions, we have that $M_1<\infty$ and $M_2^r<\infty$.
By assumption, ${\cal Q}^r$ consists of $J$-dimensional real valued cadlag functions on $[0,1]^J$ with sectional variation norm bounded by a universal $C^u$. Let $d^r=J+2$ be the dimension of $O^r=(W^r,A,Y)$. Therefore, it follows that $\{L(Q^r\circ{\bf Q}_{n,v}):Q^r\in {\cal Q}^r\}$ represents a class of $d^r$-valued cadlag functions with a universal bound on its sectional variation norm. This verifies all conditions of Lemma \ref{lemmarateQnr} so that $d^{\bar{V}}_0(Q_n, Q_{0, n}) = O_P(n^{-2/3}(\log n)^{d^r})$. In addition, we have $L(Q^r_{0,n} \circ {\bf Q}_{n, v}) = L(Q_0)$ so long as $A$ depends on $W$ only through $W^r_v = {\bf Q}_{n, v}(W)$, in which case Theorem \ref{theoremrateQn} implies $d^{\bar{V}}_0(Q_n, Q_0) = O_P(n^{-2/3}(\log n)^{d^r})$. In the case that one of the algorithms $\hat Q_1$ satisfies $d_0({\bf Q}_{n, v, 1}, Q_0) = O_P(n^{-2/3}(\log n)^{d}))$, it follows that $d^{\bar{V}}_0(Q_n, Q_0) = O_P(n^{-2/3}(\log n)^{d})$ by Theorem \ref{theoremrateQn}. 

\subsection{Plug-in Estimation with Undersmoothed M-HAL-SL}

To apply Theorem \ref{theoremaslinmhalmle}, we verify that with large enough $C_n > C_{n, cv}$ the undersmoothed M-HAL-SL solves efficient score equations for the target feature such that $P^r_n D^r(Q^r_n, G^r_{0, n}) = o_P(n^{-1/2})$ (Appendix \ref{sec:TSM_undersmoothing}). Thus, 
$$\Psi^r(Q_{n}^r) - \Psi^r(Q_{0,n}^r) = P_n D^*(Q_0, \tilde G_0) + o_P(n^{-1/2}), $$ 
where $\tilde G_0$ is the limit of $G^r_{0, n, v}$ as defined in (\ref{r1}). 
To achieve the asymptotic linearity of $\Psi^r(Q_{n}^r)$ for $\Psi(Q_0)$, it is left to be verified the conditions required for $\Psi^r(Q_{0,n}^r)-\Psi(Q_0) = o_P(n^{-1/2})$ as in Section \ref{sec:difference}. 

There are  important cases in which  $\Psi^r(Q^r_{0,n})=\Psi(Q_0)$ exactly. For example, suppose  that $A$, given $W$, equals $A$, given $W_1$ for a lower dimensional vector $W_1$. In that case, we could define $({\bf Q}_{n,v}:v)$ as a collection of estimators of $E_0(Y\mid A=1,W)$, but augmented with the fixed function $W_1$ of $W$. 
Then, 
$\Psi(Q_0)=E_0E_0(Y\mid A=1,W)=E_0E_0(Y\mid A=1,W_1)$, and, similarly, 
$\Psi(Q^r_{0,n})=E_0 E_0(Y\mid A=1,W^r)=E_0E_0(Y\mid A=1,W_1)$. So in this case, we have $\Psi^r(Q^r_{0,n})-\Psi(Q_0)=0$. 

This generalizes to any causal estimation problem in which the intervention mechanism is known to only be affected by a low dimensional summary of all measured (baseline and time-dependent) covariates, and these are included in ${\bf Q}_n$ as fixed functions. (With longitudinal data and iterative conditional expectations, ${\bf Q}_n$ can have a nested structure, sequentially defining a set of coordinate transformations for each intervention time-point.) In particular, this means that this type of  M-HAL-SL  yields efficient estimators of causal effects of single time point and multiple time point interventions based on (sequentially) randomized trials and well understood observational studies, in which one knows, at each intervention time-point, a low dimensional summary measure of the past that predicts the intervention. Here, the intervention can have both a treatment and a censoring component.

In general, the difference of the random and fixed parameter is given by:
\begin{eqnarray*}
\Psi^r(Q_{0,n}^r)-\Psi(Q_0)&=&
E_0E_0(Y\mid W^r,A=1)-E_0E_0(Y\mid W,A=1), \end{eqnarray*}
which, under similar conditions as Theorem \ref{theoremoracleensemble}, can be shown to behave as $d_0(Q_n, Q_0)$ and is thus second order (Appendix \ref{sec:TSM_diff}).

\subsection{Double Robustness}

\begin{lemma}\label{lemma:DR_TSM}
    If ${\bf Q}_{n, v}(W)$ includes the correct propensity score model $G_0(W)$ (or if $G_0$ depends on $W$ only through ${\bf Q}_{n, v}(W)$) for all $v  = 1, \dots, V$, or if ${\bf Q}_{n, v}(W)$ includes the correct outcome model $Q_0(W)$ for all $v  = 1, \dots, V$, then we have $\Psi^r(Q_{0, n}^r) - \Psi(Q_0) = 0$. Moreover, if ${\bf Q}_{n, v}(W)$ includes a consistent estimator for either $Q_0$ or $G_0$, then we have that $\Psi^r(Q_{0, n}^r)$ is  consistent for $\Psi(Q_0)$. 
\end{lemma}

\begin{proof}
Define $G_{0, n, v}(W) = E_0(A | {\bf Q}_{n, v}(W))$. 
Note that $P_0 D^*(Q_{0, n}^r \circ {\bf Q}_{n, v}, G_{0, n, v}) = 0$ by iterated conditional expectation. Utilizing the double robustness structure of $R_{20}$ and the positivity assumption, we have 
\begin{align*}
    |\Psi(Q_{0, n}^r \circ {\bf Q}_{n, v}) - \Psi(Q_0) |
    & = |P_0 D^*(Q_{0, n}^r \circ {\bf Q}_{n, v}, G_{0, n, v}) -P_0 D^*(Q_{0, n}^r \circ {\bf Q}_{n, v}, G_{0})| \\
    & \leq |P_0 (G_{0, n, v} - G_0)(Q_{0} - Q_{0, n}^r \circ {\bf Q}_{n, v})| \\
    & \leq \norm{G_{0, n, v} - G_0}_{P_0}\norm{Q_{0, n}^r \circ {\bf Q}_{n, v} - Q_0}_{P_0}. 
\end{align*}
If ${\bf Q}_{n, v}(W)$ includes estimators $Q_1(W)$ and $G_1(W)$, then $Q_1(W)$ and $G_1(W)$ are measurable functions of ${\bf Q}_{n, v}(W)$. 
Note that by tower properties, $E_0\{G_0(W) | {\bf Q}_{n, v}(W)\} = E_0(A | {\bf Q}_{n, v}(W)) = G_{0, n, v}(W)$, and $E_0\{Q_0(W) | A = 1, {\bf Q}_{n, v}(W)\} = E_0(Y | A = 1, {\bf Q}_{n, v}(W)) = Q_{0, n}^r\circ {\bf Q}_{n, v}(W)$. 
By projection properties,   
\begin{align*}
    |\Psi(Q_{0, n}^r \circ {\bf Q}_{n, v}) - \Psi(Q_0) | \leq \norm{G_{1} - G_0}_{P_0}\norm{Q_1 - Q_0}_{P_0}.
\end{align*}
This proves the claims when $Q_1 = Q_0$, or $G_1 = G_0$, or either $Q_1$ or $G_1$ is a consistent estimator ($A$ and $Y$ are both bounded). 
Lastly, if $G_0$ depends on $W$ only through ${\bf Q}_{n, v}(W)$, then $G_0(W) = G_{0, n, v}(W)$ directly. 
\end{proof}

Note that a similar result can be achieved following Lemma \ref{lemmaatesecondorder}.

When Lemma \ref{lemma:DR_TSM} holds, 
the target feature of a properly undersmoothed M-HAL-SL $\Psi^r(Q_n^r)$, that satisfies the undersmoothing condition of Theorem \ref{theoremaslinmhalmle}, is asymptotically linear for $\Psi(Q_0)$. Therefore, the typical double robustness property, that full-data TMLE is asymptotic linear 
with a correct 
model for either $Q$ or $G$,  can be preserved in meta-learning estimators for treatment specific means. 

\section{Numerical Experiments}\label{new_section6}

\subsection{Prediction Performance of M-HAL-SL}\label{sec:pred}

In this section, we evaluate the prediction performance of M-HAL-SL relative to regular super-learners. 
The tuning parameter lambda in M-HAL-SL is selected by cross-validation with \texttt{glmnet}, using honest or fast selectors.
Other super-learners for comparisons include: non-negative least square superlearner without normalization (NNLS-SL), convex super-learner which restricts the weights to be positive with sum equal to one (Convex-SL), and the simple ensemble that takes the average of base learners' predictions (Average). Each super-learner uses the following five base learner algorithms: intercept-only model (\texttt{Lrnr\_mean}), simple linear regression (\texttt{Lrnr\_glm}), Xgboost (\texttt{Lrnr\_xgboost}), support vector machine (\texttt{Lrnr\_svm}) and random forest (\texttt{Lrnr\_rf}).

We first generate data with simple distributions of five covariates $X_1,...X_5$ and a continuous outcome $Y$. The distribution of variables are as follows:
\begin{align*}
X_1 \sim U(-4,4), X_2 \sim U(-4,4), X_3 \sim Bernoulli(0.5), \\
X_4 \sim N(0,1), X_5 \sim Gamma(2,1). \numberthis \label{eq:simple}
\end{align*}
The outcome Y is generated by the ``jump'' regression function \citep{Benkeser&vanderLaan16}:\\
\begin{align*}
\psi_{0}(x)= & -I\left(x_{1}<-3\right) x_{3}+0.5 I\left(x_{1}>-2\right)-I\left(x_{1}>0\right) + 2 I\left(x_{1}>2\right) x_{3} \\
& -3 I\left(x_{1}>3\right)+1.5 I\left(x_{2}>-1\right)-5 I\left(x_{2}>1\right) x_{3}+2 I\left(x_{2}>3\right) \\
& + 2 I\left(x_{4}<0\right)- I\left(x_{5}>5\right)-I\left(x_{4}<0\right) I\left(x_{1}<0\right)+2 x_{3}, \\
Y = & \psi_{0}(X) + \epsilon, \epsilon \sim N(0,1). 
\end{align*}
In each iteration, a dataset is generated with sample sizes $n = 200, 500, 1000,$ and $2000$, respectively. The measure of performance is the mean of squared error (MSE) on an external test dataset of 5000 samples generated from the same distribution. We also provide their relative MSEs using Convex-SL as the baseline. 
In this scenario, M-HAL-SL performs as good as other super-learners (Table \ref{tab:simResult_d5_diversed}).

In the second scenario, we generate data with more complicated distributions of 20 covariates ($X_1,...X_{20}$). These covariates are divided in four equal sized groups, and the conditional mean of $Y$ is an additive model of four functions of the corresponding clusters of covariates. Each function involves two-way intersections between the five covariates. 
The data-generating distribution is as follows: \\
\[
\begin{array}{l}
\mbox{$X_1 \sim U(-4,4)$, $X_2 \sim U(-4,4)$, $X_3 \sim Bernoulli(0.5)$, $X_4 \sim N(0,1)$,} \\
\mbox{$X_5 \sim Gamma(2,1)$, $X_6 \sim Pois(2)$, $X_7 \sim Exp(3)$, $X_8 \sim Beta(1,1)$,} \\
\mbox{ $X_9 \sim {\chi}^{2}(2)$, $X_{10} \sim Geom(0.6)$, 
$X_{11} \sim U(-4,4)$, $X_{12} \sim U(-4,4)$, }\\
\mbox{ $X_{13} \sim Bernoulli(0.5)$, $X_{14} \sim N(0,1)$, $X_{15} \sim Gamma(2,1)$,}\\
\mbox{ $X_{16} \sim Pois(1)$, $X_{17} \sim Exp(1)$, $X_{18} \sim Beta(2,1)$, }\\
\mbox{ $X_{19} \sim {\chi}^{2}(1)$, $X_{20} \sim Geom(0.8)$.}
\end{array}
\]
\[
\begin{array}{l}
G_{1}(x)=-I\left(x_{1}<-3\right) x_{3}+0.5 I\left(x_{1}>-2\right)-I\left(x_{1}>0\right)+ 
2 I\left(x_{1}>2\right) x_{3}\\-3 I\left(x_{1}>3\right)+1.5 I\left(x_{2}>-1\right)- 
5 I\left(x_{2}>1\right) x_{3}+2 I\left(x_{2}>3\right)\\+2 I\left(x_{4}<0\right) x_{2}- 
I\left(x_{5}>5\right) x_{1}-I\left(x_{4}<0\right) I\left(x_{1}<0\right)+2 x_{3}
\\
\\
G_{2}(x)=-I\left(x_{10}>3\right) x_{8}+0.5 I\left(x_{10}>2\right)-I\left(x_{10}>0\right)+ 
2 I\left(x_{10}>2\right) x_{8}\\-3 I\left(x_{10}>3\right)+
1.5 I\left(x_{9}>5\right) - 
5 I\left(x_{9}>1\right) x_{8}+
2 I\left(x_{9}>3\right)\\+
I\left(x_{7}>4\right) x_{9}- 
I\left(x_{6}>5\right) x_{10}-
I\left(x_{7}>1\right) I\left(x_{10}<2\right)+
2 x_{8}
\\
\\
G_{3}(x)=-I\left(x_{11}<-3\right) x_{13}+
0.5 I\left(x_{11}>-2\right)-
I\left(x_{11}>0\right)+ 
2 I\left(x_{11}>2\right) x_{13}\\
-3 I\left(x_{11}>3\right)+
1.5 I\left(x_{12}>-1\right)- 
5 I\left(x_{12}>1\right) x_{13}+
2 I\left(x_{12}>3\right)\\+
I\left(x_{14}<0\right) x_{12}- 
I\left(x_{15}>5\right) x_{11}-
I\left(x_{14}<0\right) I\left(x_{11}<0\right)+
2 x_{13}
\\
\\
G_{4}(x)=-I\left(x_{19}>3\right) x_{17}+
0.5 I\left(x_{19}>2\right)-
I\left(x_{19}>0\right)+ 
2 I\left(x_{19}>2\right) x_{18}\\
-3 I\left(x_{19}>3\right)+
1.5 I\left(x_{16}>5\right) - 
5 I\left(x_{16}>1\right) x_{18}+
2 I\left(x_{16}>3\right)\\+
I\left(x_{16}>4\right) x_{19}- 
I\left(x_{16}>5\right) x_{20}-
I\left(x_{17}>1\right) I\left(x_{20}<2\right)+
2 x_{18}
\\
\\
\psi_{0}(x) = G_1(x) + G_2(x) + G_3(x) + G_4(x)
\\
\\
Y = \psi_{0}(X) + \epsilon,\ 
\epsilon \sim N(0,1). \numberthis \label{eq:complex}
\end{array}
\]
In the second scenario with more complicated distributions, M-HAL-SL performs slightly better than other super-learners (Table \ref{tab:simResult_d20_diversed_wo_g}).

Note that the asymptotic performance of M-HAL-SL relies on the convergence rate of the HAL algorithm, which allows M-HAL-SL to approximate more complex functions of the base learners. In comparison, the learner library of other super-learners are limited to simple linear combinations, and their asymptotic performances rely on oracle inequality and therefore the best estimator in the learner library. Indeed, at the largest sample size ($n = 2000$), both Convex-SL and NNLS-SL perform similarly as the best candidate algorithm Lrnr\_xgboost under either simple or more complicated distributions, but M-HAL-SL gains additional precision when the data generating distribution is more complex and a more flexible combination of base learners may jointly assist prediction. 
Our results show that M-HAL-SL is a valid alternative super-learner, and its finite-sample performance can be similar or slightly better than Convex-SL and NNLS-SL when base learners operate on all variables, depending on sample sizes and complexity of data generating distributions.

\begin{table}
    \centering

\begin{tabular}{lcccccccc}
\toprule
\multicolumn{1}{c}{ } & \multicolumn{2}{c}{n = 200} & \multicolumn{2}{c}{n = 500} & \multicolumn{2}{c}{n = 1000} & \multicolumn{2}{c}{n = 2000} \\
\cmidrule(l{3pt}r{3pt}){2-3} \cmidrule(l{3pt}r{3pt}){4-5} \cmidrule(l{3pt}r{3pt}){6-7} \cmidrule(l{3pt}r{3pt}){8-9}
metalearner & mse & relative\_mse & mse & relative\_mse & mse & relative\_mse & mse & relative\_mse\\
\midrule
Meta-HAL-d2 (Honest CV) & 2.25 & 0.96 & 1.60 & 0.96 & 1.34 & 0.98 & 1.21 & 0.99\\
Meta-HAL-d2 (Fast SL CV) & 2.23 & 0.95 & 1.59 & 0.96 & 1.34 & 0.98 & 1.21 & 0.99\\
Meta-HAL-d2 (Fast CV) & 2.25 & 0.96 & 1.61 & 0.97 & 1.34 & 0.98 & 1.21 & 1.00\\
Meta-HAL-d1 (Honest CV) & 2.22 & 0.94 & 1.58 & 0.95 & 1.33 & 0.98 & 1.20 & 0.99\\
Meta-HAL-d1 (Fast SL CV) & 2.23 & 0.95 & 1.59 & 0.96 & 1.33 & 0.98 & 1.20 & 0.99\\
Meta-HAL-d1 (Fast CV) & 2.23 & 0.95 & 1.58 & 0.95 & 1.33 & 0.98 & 1.20 & 0.99\\
Convex & 2.36 & 1.00 & 1.66 & 1.00 & 1.36 & 1.00 & 1.21 & 1.00\\
NNLS & 2.32 & 0.98 & 1.63 & 0.98 & 1.35 & 0.99 & 1.21 & 1.00\\
Discrete SL & 2.50 & 1.06 & 1.73 & 1.04 & 1.41 & 1.03 & 1.24 & 1.02\\
Average & 3.00 & 1.28 & 2.53 & 1.52 & 2.28 & 1.67 & 2.12 & 1.75\\
Lrnr\_mean & 5.05 & 2.16 & 5.03 & 3.04 & 5.03 & 3.69 & 5.02 & 4.14\\
Lrnr\_glm & 4.69 & 2.00 & 4.60 & 2.78 & 4.58 & 3.36 & 4.56 & 3.76\\
Lrnr\_xgboost & 2.49 & 1.06 & 1.74 & 1.04 & 1.41 & 1.03 & 1.24 & 1.02\\
Lrnr\_svm & 3.27 & 1.40 & 2.94 & 1.77 & 2.72 & 1.99 & 2.52 & 2.08\\
Lrnr\_rf & 2.51 & 1.07 & 1.78 & 1.07 & 1.44 & 1.06 & 1.27 & 1.04\\
\bottomrule
\end{tabular}

    \caption{Prediction performance of super-learners with simple distributions (\ref{eq:simple}) in Section \ref{sec:pred}. }
    \label{tab:simResult_d5_diversed}
\end{table}

\begin{table}
    \centering

\begin{tabular}{lcccccccc}
\toprule
\multicolumn{1}{c}{ } & \multicolumn{2}{c}{n = 200} & \multicolumn{2}{c}{n = 500} & \multicolumn{2}{c}{n = 1000} & \multicolumn{2}{c}{n = 2000} \\
\cmidrule(l{3pt}r{3pt}){2-3} \cmidrule(l{3pt}r{3pt}){4-5} \cmidrule(l{3pt}r{3pt}){6-7} \cmidrule(l{3pt}r{3pt}){8-9}
metalearner & mse & relative\_mse & mse & relative\_mse & mse & relative\_mse & mse & relative\_mse\\
\midrule
Meta-HAL-d2 (Honest CV) & 13.59 & 0.97 & 9.67 & 0.87 & 7.17 & 0.84 & 5.28 & 0.84\\
Meta-HAL-d2 (Fast SL CV) & 13.67 & 0.97 & 10.00 & 0.90 & 7.20 & 0.84 & 5.29 & 0.84\\
Meta-HAL-d2 (Fast CV) & 13.65 & 0.97 & 9.69 & 0.87 & 7.18 & 0.84 & 5.29 & 0.84\\
Meta-HAL-d1 (Honest CV) & 13.39 & 0.95 & 9.54 & 0.86 & 7.10 & 0.83 & 5.27 & 0.84\\
Meta-HAL-d1 (Fast SL CV) & 13.37 & 0.95 & 10.00 & 0.90 & 7.52 & 0.88 & 5.47 & 0.87\\
Meta-HAL-d1 (Fast CV) & 13.60 & 0.97 & 9.57 & 0.86 & 7.10 & 0.83 & 5.27 & 0.84\\
Convex & 14.06 & 1.00 & 11.15 & 1.00 & 8.55 & 1.00 & 6.29 & 1.00\\
NNLS & 13.33 & 0.95 & 9.90 & 0.89 & 7.46 & 0.87 & 5.51 & 0.88\\
Discrete SL & 14.80 & 1.05 & 11.05 & 0.99 & 8.24 & 0.96 & 6.10 & 0.97\\
Average & 14.62 & 1.04 & 12.69 & 1.14 & 11.29 & 1.32 & 10.14 & 1.61\\
Lrnr\_mean & 19.90 & 1.42 & 19.83 & 1.78 & 19.81 & 2.32 & 19.80 & 3.15\\
Lrnr\_glm & 16.98 & 1.21 & 15.69 & 1.41 & 15.35 & 1.80 & 15.16 & 2.42\\
Lrnr\_xgboost & 14.95 & 1.06 & 11.04 & 0.99 & 8.24 & 0.96 & 6.10 & 0.97\\
Lrnr\_svm & 15.21 & 1.08 & 12.97 & 1.17 & 11.50 & 1.35 & 10.32 & 1.64\\
Lrnr\_rf & 14.61 & 1.04 & 12.56 & 1.13 & 11.14 & 1.30 & 9.95 & 1.58\\
\bottomrule
\end{tabular}

    \caption{
    Prediction performance of super-learners with more complicated distributions (\ref{eq:complex}) in Section \ref{sec:pred}. 
    }
    \label{tab:simResult_d20_diversed_wo_g}
\end{table}

\subsection{Performance of M-HAL-SL as Ensemble Method}\label{sec:ensemble}

When each of the base learners utilizes only part of the variables and captures partial information, M-HAL-SL can illustrate advantages as a more flexible ensemble method. In those scenarios, the unknown and possibly non-linear relationship between the dependent variables and base learners can still be approximated by M-HAL-SL so long as the variation norm is bounded on the meta-level (which usually has much lower dimensions depending on the number of base learner algorithms), whereas Convex-SL and NNLS-SL may be capped by the performance of each individual learner. 

We consider the following two scenarios:
\begin{enumerate}
    \item data is generated with the simple distribution, and base learner algorithms only include univariate linear regression estimators on single covariate $X_i$, 
    \item data is generated with the more complicated distribution, and base learner algorithms are univariate linear regression estimators on single covariate $X_i$. 
\end{enumerate}

In both scenarios, the simulation results show expected performance. As base learners only learn from a subset of covariates, the prediction accuracy of M-HAL-SL significantly increases with larger sample sizes relative to other super-learners and base learners (Table \ref{tab:simResult_d5_ols} and Table \ref{tab:simResult_d20_ols_wo_g}). These results demonstrate the potential application of M-HAL-SL as a better model fusion method that can achieve higher precision than each individual model and regular super-learners, especially when each base learner contains only partial information. 
For example, this is applicable to the analysis of multi-modal data, where each modality requires separate training of complex models with tailored structures.

\begin{table}
    \centering

\begin{tabular}{lcccccccc}
\toprule
\multicolumn{1}{c}{ } & \multicolumn{2}{c}{n = 200} & \multicolumn{2}{c}{n = 500} & \multicolumn{2}{c}{n = 1000} & \multicolumn{2}{c}{n = 2000} \\
\cmidrule(l{3pt}r{3pt}){2-3} \cmidrule(l{3pt}r{3pt}){4-5} \cmidrule(l{3pt}r{3pt}){6-7} \cmidrule(l{3pt}r{3pt}){8-9}
metalearner & mse & relative\_mse & mse & relative\_mse & mse & relative\_mse & mse & relative\_mse\\
\midrule
Meta-HAL-d2 (Honest CV) & 3.72 & 0.79 & 3.03 & 0.65 & 2.62 & 0.56 & 2.39 & 0.51\\
Meta-HAL-d2 (Fast SL CV) & 3.74 & 0.80 & 3.12 & 0.67 & 2.74 & 0.59 & 2.46 & 0.53\\
Meta-HAL-d2 (Fast CV) & 3.76 & 0.80 & 3.15 & 0.68 & 2.70 & 0.58 & 2.47 & 0.53\\
Meta-HAL-d1 (Honest CV) & 4.45 & 0.95 & 4.03 & 0.86 & 3.79 & 0.81 & 3.61 & 0.78\\
Meta-HAL-d1 (Fast SL CV) & 4.58 & 0.97 & 4.23 & 0.90 & 3.92 & 0.84 & 3.65 & 0.78\\
Meta-HAL-d1 (Fast CV) & 4.43 & 0.94 & 4.08 & 0.87 & 3.84 & 0.82 & 3.66 & 0.79\\
Convex & 4.71 & 1.00 & 4.67 & 1.00 & 4.67 & 1.00 & 4.66 & 1.00\\
NNLS & 4.71 & 1.00 & 4.67 & 1.00 & 4.66 & 1.00 & 4.66 & 1.00\\
Discrete SL & 4.74 & 1.01 & 4.68 & 1.00 & 4.67 & 1.00 & 4.67 & 1.00\\
Average & 4.88 & 1.04 & 4.86 & 1.04 & 4.86 & 1.04 & 4.85 & 1.04\\
x1 & 5.04 & 1.07 & 5.00 & 1.07 & 4.99 & 1.07 & 4.99 & 1.07\\
x2 & 5.07 & 1.08 & 5.04 & 1.08 & 5.03 & 1.08 & 5.03 & 1.08\\
x3 & 5.01 & 1.06 & 4.98 & 1.06 & 4.97 & 1.06 & 4.96 & 1.06\\
x4 & 4.71 & 1.00 & 4.68 & 1.00 & 4.67 & 1.00 & 4.67 & 1.00\\
x5 & 5.06 & 1.08 & 5.03 & 1.08 & 5.02 & 1.08 & 5.01 & 1.08\\
\bottomrule
\end{tabular}

    \caption{
    Ensemble performance of super-learners when each base learner is a univariate linear regression on each single covariate $X_i$ (Section \ref{sec:ensemble}) with simple distributions (\ref{eq:simple}). 
    }
    \label{tab:simResult_d5_ols}
\end{table}

\begin{table}
    \centering

\begin{tabular}{lcccccccc}
\toprule
\multicolumn{1}{c}{ } & \multicolumn{2}{c}{n = 200} & \multicolumn{2}{c}{n = 500} & \multicolumn{2}{c}{n = 1000} & \multicolumn{2}{c}{n = 2000} \\
\cmidrule(l{3pt}r{3pt}){2-3} \cmidrule(l{3pt}r{3pt}){4-5} \cmidrule(l{3pt}r{3pt}){6-7} \cmidrule(l{3pt}r{3pt}){8-9}
metalearner & mse & relative\_mse & mse & relative\_mse & mse & relative\_mse & mse & relative\_mse\\
\midrule
Meta-HAL-d2 (Honest CV) & 14.64 & 0.80 & 10.14 & 0.56 & 7.35 & 0.40 & 4.88 & 0.26\\
Meta-HAL-d2 (Fast SL CV) & 14.64 & 0.80 & 10.14 & 0.56 & 7.35 & 0.40 & 4.88 & 0.26\\
Meta-HAL-d2 (Fast CV) & 14.90 & 0.82 & 10.35 & 0.57 & 7.60 & 0.41 & 5.01 & 0.27\\
Meta-HAL-d1 (Honest CV) & 15.44 & 0.85 & 12.80 & 0.70 & 11.36 & 0.62 & 10.31 & 0.56\\
Meta-HAL-d1 (Fast SL CV) & 15.52 & 0.85 & 12.99 & 0.71 & 11.53 & 0.63 & 10.41 & 0.56\\
Meta-HAL-d1 (Fast CV) & 15.58 & 0.85 & 12.93 & 0.71 & 11.54 & 0.63 & 10.46 & 0.56\\
Convex & 18.25 & 1.00 & 18.26 & 1.00 & 18.36 & 1.00 & 18.52 & 1.00\\
NNLS & 17.20 & 0.94 & 16.36 & 0.90 & 16.03 & 0.87 & 15.83 & 0.86\\
Discrete SL & 18.84 & 1.03 & 18.63 & 1.02 & 18.58 & 1.01 & 18.58 & 1.00\\
Average & 19.41 & 1.06 & 19.36 & 1.06 & 19.33 & 1.05 & 19.33 & 1.04\\
x1 & 19.91 & 1.09 & 19.76 & 1.08 & 19.73 & 1.07 & 19.71 & 1.06\\
x2 & 18.71 & 1.03 & 18.59 & 1.02 & 18.55 & 1.01 & 18.54 & 1.00\\
x3 & 19.92 & 1.09 & 19.79 & 1.08 & 19.76 & 1.08 & 19.74 & 1.07\\
x4 & 19.92 & 1.09 & 19.81 & 1.08 & 19.78 & 1.08 & 19.76 & 1.07\\
x5 & 19.99 & 1.10 & 19.86 & 1.09 & 19.82 & 1.08 & 19.80 & 1.07\\
x6 & 19.98 & 1.10 & 19.86 & 1.09 & 19.82 & 1.08 & 19.81 & 1.07\\
x7 & 19.95 & 1.09 & 19.83 & 1.09 & 19.80 & 1.08 & 19.78 & 1.07\\
x8 & 19.93 & 1.09 & 19.79 & 1.08 & 19.75 & 1.08 & 19.73 & 1.06\\
x9 & 19.96 & 1.09 & 19.81 & 1.08 & 19.77 & 1.08 & 19.76 & 1.07\\
x10 & 19.79 & 1.08 & 19.65 & 1.08 & 19.60 & 1.07 & 19.59 & 1.06\\
x11 & 19.88 & 1.09 & 19.75 & 1.08 & 19.73 & 1.07 & 19.71 & 1.06\\
x12 & 18.69 & 1.02 & 18.57 & 1.02 & 18.54 & 1.01 & 18.52 & 1.00\\
x13 & 19.93 & 1.09 & 19.80 & 1.08 & 19.76 & 1.08 & 19.74 & 1.07\\
x14 & 19.96 & 1.09 & 19.83 & 1.09 & 19.77 & 1.08 & 19.76 & 1.07\\
x15 & 19.98 & 1.10 & 19.86 & 1.09 & 19.82 & 1.08 & 19.80 & 1.07\\
x16 & 18.87 & 1.03 & 18.71 & 1.02 & 18.67 & 1.02 & 18.64 & 1.01\\
x17 & 19.76 & 1.08 & 19.66 & 1.08 & 19.62 & 1.07 & 19.61 & 1.06\\
x18 & 19.94 & 1.09 & 19.80 & 1.08 & 19.76 & 1.08 & 19.74 & 1.07\\
x19 & 19.90 & 1.09 & 19.75 & 1.08 & 19.72 & 1.07 & 19.70 & 1.06\\
x20 & 19.99 & 1.10 & 19.85 & 1.09 & 19.82 & 1.08 & 19.80 & 1.07\\
\bottomrule
\end{tabular}

    \caption{
    Ensemble performance of super-learners when each base learner is a univariate linear regression on each single covariate $X_i$ (Section \ref{sec:ensemble}) with more complicated distributions (\ref{eq:complex}). 
    }
    \label{tab:simResult_d20_ols_wo_g}
\end{table}

\subsection{Undersmoothed M-HAL-MLE for Plug-in Estimation of Target Feature}\label{sec:plugin}


This section evaluates the performance of (undersmoothed) M-HAL-SL plug-in estimation, following the treatment specific mean example (Section \ref{sec:TSM}). We first verify asymptotic linearity properties with low-dimensional ($p << n$) covariates. In addition, we simulate high-dimensional covariates with small $n$ large $p$ (common in electronic health records, genomics, and brain imaging data) and intentionally make the initial estimators overfitted with increased variation in order to test that the required bounded variation condition is indeed relaxed to the meta-level data (in much lower dimensions) rather than the original input data. 
Such higher dimensional and highly varying models emulate the specific challenges when machine learning and deep learning algorithms are applied. 

We simulate $4$ clinical covariates along with $4$ or $4n$ additional covariates for lower or higher dimensional settings at sample size $n = 200$. Only $4$ additional covariates remain active in the propensity score model and the outcome model; 
choosing a larger $\ell_1$ norm bound than the cross-validated choice for the Lasso algorithm emulates an overfitted estimator with inflated sectional variation. The clinical covariates are always part of the true data generating process not subject to regularization. 
Specifically, we have the following additive model, 
\begin{align*}
    W = & (W_c, W_h) \\
    W_h = & (W_a, W_n) \\
    G_0(W) = & \beta_{0A} + \beta_A \cdot \pmb 1^{\top}(W_c, W_a)\\
    Y = & \beta_Y \cdot \pmb 1^{\top}(W_c, W_a) + \psi_0 A +  \epsilon\\
    \epsilon \sim & N(0, 1). 
\end{align*}
$W_c$ is the vector of clinical covariates with $|W_c| = 4$. 
$W_h = (W_a, W_n)$ is the vector of additional covariates where only $W_a$ is active with $|W_a| = 4$; we set the noise vector $W_n = \varnothing$ for low-dimensional cases and $|W_h| = 4n$ for high-dimensional (sparse) cases. 
We set $\beta_{0A} = - |W_a| \beta_A / 2$ to avoid violating positivity assumptions. $\beta_A = 0.2$.  $\beta_Y = 0.6$. $\psi_0 = 1$. Note that $\psi_0 = E_{P_0}(Y | A = 1, W) - E_{P_0}(Y | A = 0, W)$ is the target parameter, which is the difference between the treatment specific means as defined in Section \ref{sec:TSM}. 

The following estimators are evaluated. 
\begin{itemize}
    \item $\hat \psi_{\text{noadj}}$: the difference between the observed group means without covariate adjustment. 
    \item $\hat \psi_{\text{tmle}}$: TMLE adjusting for all covariates, $W$. For low-dimensional settings, the initial estimators of $G_0$ and $Q_0$ for $\hat \psi_{\text{tmle}}$ are main-term logistic and linear regressions. For high-dimensional settings with additional covariates $W_h$, we use regularized (generalized) linear regression to emulate potentially over-fitted initial estimators in $\hat \psi_{\text{tmle}}$. Specifically, 
the additional loss term is $\lambda_{1} ||\hat\beta_h||_1$ for regularization, where $\hat\beta_h$ is the estimated coefficients for high-dimensional covariates $W_h$; we then reduce $\lambda_{1}$ to 
$0.1\%$ of the cross-validated choice $\lambda_{\text{cv}}$ to emulate over-fitted initial estimators with increased total variations. 
    \item $\hat \psi_{\text{meta}}^{\text{init}}$: M-HAL-SL plug-in, where the learner library, $\hat{{\bf Q}}$, consists of four univariate identity maps for the clinical covariate, the two group mean estimations using the same initial learner for $Q_0$ in $\hat \psi_{\text{tmle}}$, and the initial learner for $G_0$ in $\hat \psi_{\text{tmle}}$. Therefore, the meta-level data is $J=7$-dimensional, compared with the up to over $800$-dimensional original input. The fast M-HAL-MLE variation bound selector (verified in Section \ref{new_section6}) is applied. 
    \item $\hat \psi_{\text{meta}}^{\text{us}}$ ($\hat G$): Undersmoothed M-HAL-SL plug-in. Note that $\hat \psi_{\text{meta}}^{\text{us}}$ is a function of a separate $G$ estimator that converges to $\tilde G_0$ in (\ref{r1}), which may lead to super-efficiency when $\tilde G_0 \not = G_0$. 
    For a direct comparison with regular full-data $\hat \psi_{\text{tmle}}$, we use the same $\hat G$ as the initial learner for $G_0$ in $\hat \psi_{\text{tmle}}$. 
\end{itemize}

    In low-dimensional settings, undersmoothed M-HAL-SL plug-in achieves asymptotic linearity with a reasonable finite-sample bias/SE ratio at a small cost in MSE compared to regular TMLE (Figure \ref{fig:lowD}, Table \ref{tab:lowD}). Note that undersmoothing reduces bias and improves bias-variance trade-off of the M-HAL-SL plug-in. 
    
    In high-dimensional settings with LASSO initial learners, only undersmoothed M-HAL-SL plug-in maintains reasonable bias/SE ratios with bias reduction, illustrating stable performance even with overfitted full-data initial learners that have increased variation (Figure \ref{fig:lasso}, Table \ref{tab:lasso}).

\begin{figure}
    \centering
    \includegraphics[width=10cm]{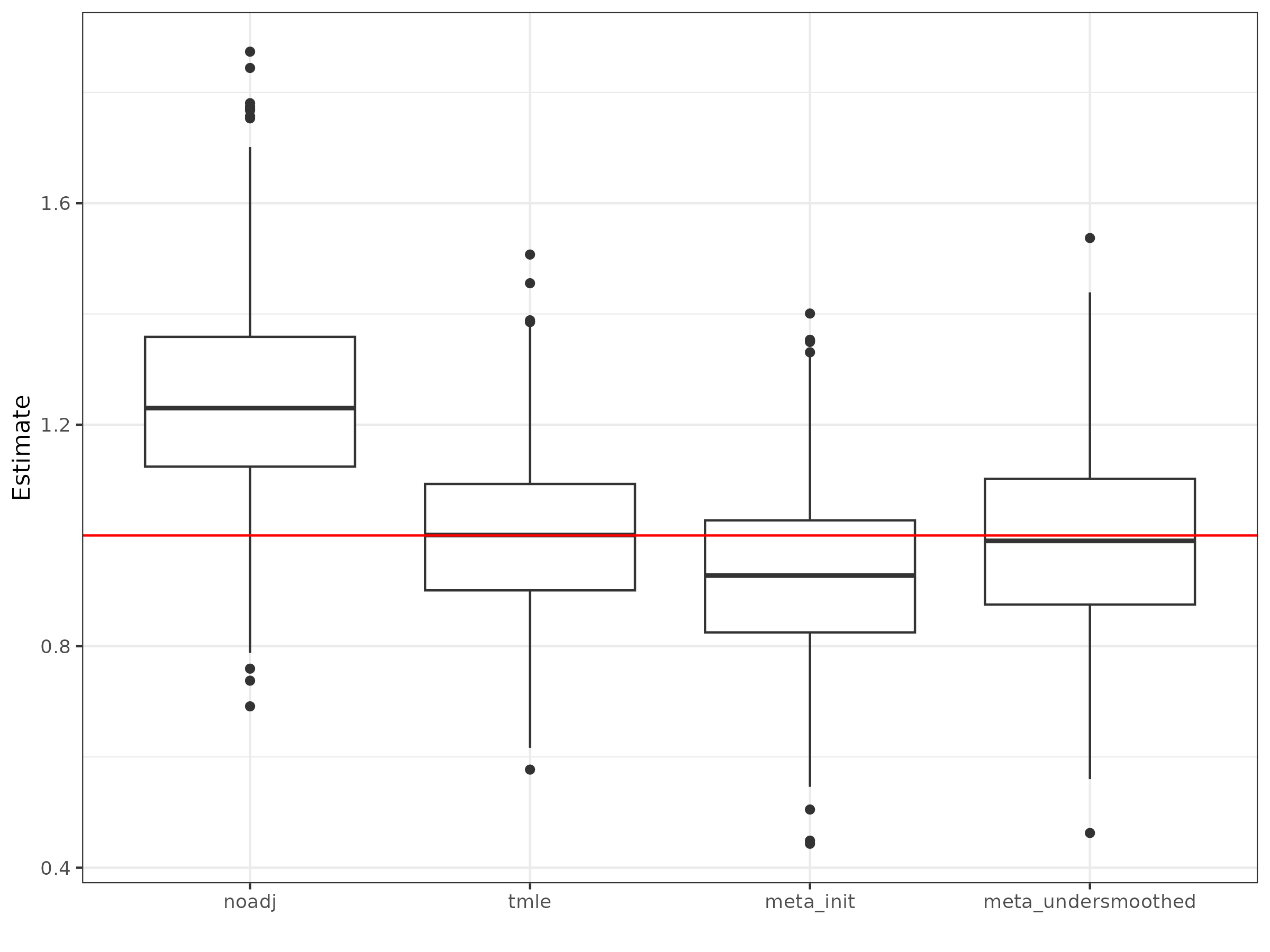}
    \caption{
        Performance of (undersmoothed) M-HAL-SL plug-in estimators compared with regular TMLE or no adjustment in Section \ref{sec:plugin}. Low-dimensional settings; $n = 200$ with $4$ clinical covariates and $4$ additional covariates. The horizontal red line indicates the true parameter value.
    }
    \label{fig:lowD}
\end{figure}

\begin{table}[ht]
\centering

\begin{tabular}{lrrrr}
\toprule
  & MSE & Bias & SD & Ratio\\
\midrule
noadj & 0.096 & 0.243 & 0.192 & 1.265\\
tmle & 0.022 & 0.003 & 0.149 & 0.017\\
meta\_init & 0.032 & -0.074 & 0.163 & -0.452\\
meta\_undersmoothed & 0.028 & -0.011 & 0.166 & -0.063\\
\bottomrule
\end{tabular}

 \caption{Performance of (undersmoothed) M-HAL-SL plug-in estimators compared with regular TMLE or no adjustment in Section \ref{sec:plugin}. Low-dimensional settings; $n = 200$ with $4$ clinical covariates and $4$ additional covariates. }
 \label{tab:lowD}
\end{table}

\begin{figure}
    \centering
    \includegraphics[width=10cm]{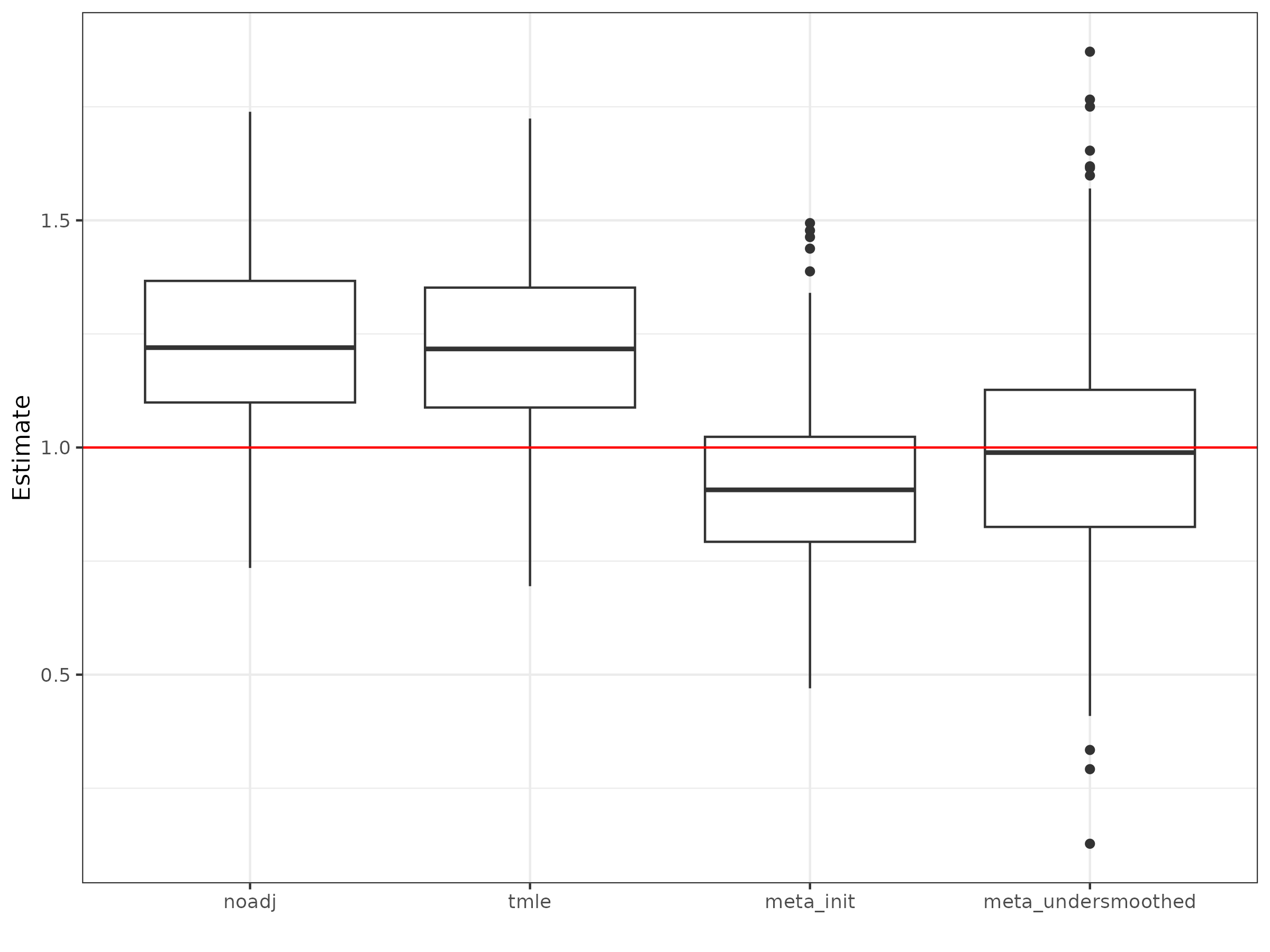}
    \caption{
                Performance of (undersmoothed) M-HAL-SL plug-in estimators compared with regular TMLE or no adjustment in Section \ref{sec:plugin}. High-dimensional settings where initial learners may be overfitted; $n = 200$ with $4$ clinical covariates and $800$ additional covariates; initial learners use the LASSO algorithm and $0.1\%$ of the cross-validated choice of $\lambda_1$.  The horizontal red line indicates the true parameter value.
    }
    \label{fig:lasso}
\end{figure}

\begin{table}[ht]
\centering
 \begin{tabular}{rrrrr}
  \hline
 & MSE & Bias & SD & Ratio \\
\hline
noadj & 0.091 & 0.233 & 0.192 & 1.217\\
  \hline
$\lambda_{\text{cv}}$ & MSE & Bias & SD & Ratio \\
\hline
tmle & 0.031 & 0.069 & 0.161 & 0.429\\
meta\_init & 0.036 & -0.082 & 0.170 & -0.483\\
meta\_undersmoothed & 0.064 & 0.015 & 0.253 & 0.058\\
  \hline
$0.1\% \lambda_{\text{cv}}$ & MSE & Bias & SD & Ratio \\
\hline
tmle & 0.088 & 0.227 & 0.192 & 1.185\\
meta\_init & 0.037 & -0.085 & 0.173 & -0.494\\
meta\_undersmoothed & 0.057 & -0.011 & 0.238 & -0.044\\
  \hline
\end{tabular}

 \caption{Performance of (undersmoothed) M-HAL-SL plug-in estimators compared with regular TMLE or no adjustment in Section \ref{sec:plugin}. High-dimensional settings where initial learners may be overfitted; $n = 200$ with $4$ clinical covariates and $800$ additional covariates; initial learners use the LASSO algorithm. $100\%$ or $0.1\%$ of the cross-validated choice of $\lambda_1$ is used. 
 }
 \label{tab:lasso}
\end{table}

\section{Data Application} \label{sec:data}

To demonstrate the utility of the proposed method, we applied it to a high-dimensional imaging-based mediator analysis in pain studies \citep{geuter2020multiple,nath2023machine}. The dataset consisted of 10,472 trials where thermal stimuli were applied and participants' subjective pain ratings were reported. Resting-state fMRI data was collected during the experiments, and preprocessed activation maps with $91\times 109\times 91$ voxels were analyzed as the high-dimensional mediator. The percentage of thermal stimulus's effect on pain rating, $Y$ mediated through activation maps, $Z$, 
\begin{align*}
    \frac{\mathbb{E}\{Y(1) - Y(1, Z(0))\}}{\mathbb{E}\{Y(1) - Y(0)\}}
\end{align*}
defined by the ratio of the natural indirect effect (NIE) and average treatment effect (ATE), was the target parameter. Under identification assumptions, the percentage mediated is a pathwise differentiable parameter. Due to the experiment design, a positive percentage mediated was expected. 

This estimation problem was challenging because of the estimation of conditional expectations given high-dimensional mediators. 
Therefore, we compared: 
\begin{itemize}
    \item Strategy 1: dimension reduction using pretrained ResNet3D models  (without the last classification layer, from $91\times 109\times 91$ to $512$), followed by targeted maximum likelihood estimation (TMLE) \citep{zheng2017longitudinal,wang2023targeted} using HAL-MLE as the initial estimators, and
    \item Strategy 2: plug-in estimation with M-HAL super-learners, where the additional meta-learning step (data-adaptive coordinate-transformation 
    on top of the $512$-dimensional summary input) further reduced the effective mediator dimension to $2$ (two estimated functions that identify NIE and predict the influence curves; see Appendix \ref{sec:appendix_med}).
\end{itemize} 
We generated a bootstrap sample of size $10,000$ to create bootstrap 95\% confidence intervals (CIs) and test the performance on a true effect sample where a positive percentage mediated is expected. A null effect sample was also generated by replacing the mediator, $Z$, with independent noise following standard normal distribution (so that $0\%$ mediated was expected), in order to test the performance of type-I error protection. 

Figure \ref{fig:tmle_metahal_imaging} shows that 
TMLE using the 512-dimensional transformed mediators is still unstable and potentially biased on the true effect samples 
(boostrap CI is above 0 but wider than the range $[0, 1]$ of the percentage mediated) and is subject to type-I error on the null effect sample  (bootstrap CI remains above 0 despite noninformative mediators). In comparison, M-HAL super-learner plug-in constructs a much narrower bootstrap CI above 0 from the true effect sample, conforming the positive effect; meanwhile, the bootstrap CI from the null effect sample is around the truth 0, protecting against type-I error. This verifies the reliable asymptotic properties of M-HAL super-learner plug-in, which is based on more realistic conditions (defined with respect to a much lower dimensional function space) that are more easily satisfied on the finite sample. It illustrates that the proposed method, with the additional meta-learning step and data-adaptive coordinate-transformation, is promising for effectively transferring existing model knowledge while providing reliable inference under traditionally challenging curse-of-dimensionality scenarios. 

The use case of meta-HAL in the context of multiple pretrained models is further investigated. Seven pretrained models, defined with different network depths, are available and may construct initial dimension reduction from $91\times 109\times 91$ to either $512$ or $2,048$, depending on the dimension of the layer prior to the last classification layer. Without further prior knowledge, the optimal model choice is unknown, and therefore it is preferable to use an ensemble that integrates all available information. However, combining all these available dimension reduction models would be challenging with Strategy 1, which already suffers from the curse of dimensionality using one of the models (ResNet3D\_10). Because Strategy 2 implements a 2-dimensional coordinate-transformation for each model, it is possible to construct an effective ensemble. Figure \ref{fig:ensemble} shows that meta-HAL plug-in utilizing a 14-dimensional coordinate-transformation out of all the seven networks successfully confirms the largest percentage mediated, compared with the meta-HAL plug-in estimation using each single network. 

\begin{figure}
    \centering
    \includegraphics[width=0.9\linewidth]{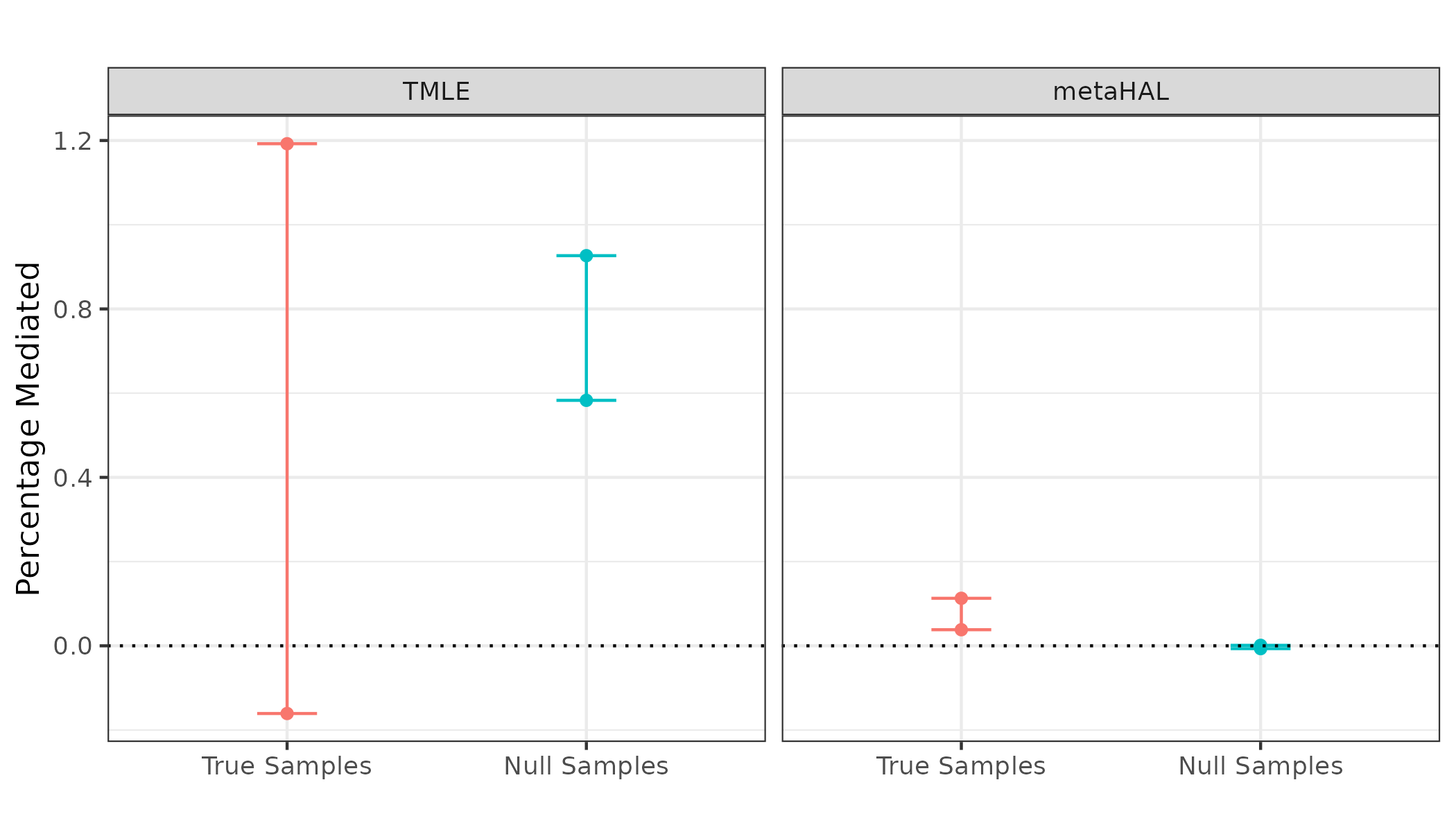}
    \caption{Bootstrap CIs using TMLE with standard mediator dimension reduction versus undersmoothed meta-HAL plug-in. Bootstrap sample size: 10,000. True effect samples: preprocessed fMRI mediators ($91\times 109\times 91$-dimensional) transformed through a \texttt{ResNet3D\_10} model ($512$-dimensional, stopped before the fully-connected classification layer)  with pretrained weights \citep{chen2019med3d}. Null effect samples: each dimension of the original mediators replaced by independent standard normal noise.}
    \label{fig:tmle_metahal_imaging}
\end{figure}

\begin{figure}
    \centering
    \includegraphics[width=0.9\linewidth]{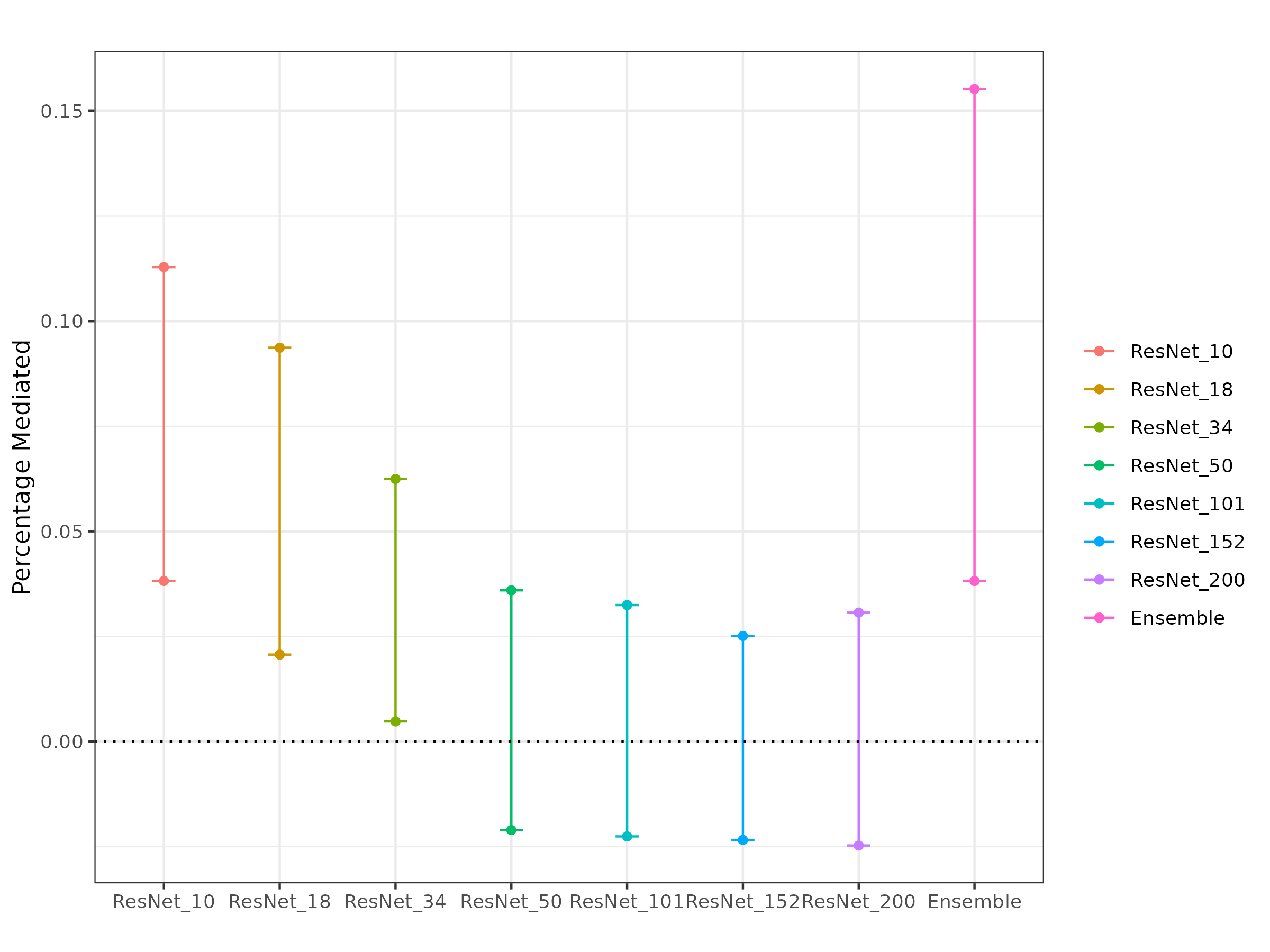}
    \caption{
Bootstrap CIs with undersmoothed meta-HAL plug-in using coordinate-transformation created by  each individual pretrained ResNet models (labeled by the depths of the networks) or an ensemble of all the seven models. Bootstrap sample size: 10,000. 
    }
    \label{fig:ensemble}
\end{figure}

 \section{Discussion}\label{section8} 
 We proposed a super-learner that uses HAL-MLE to select a best ensemble among all cadlag functions of the $J$ estimators in its library with a  bound on its sectional variation norm. 
 The oracle ensemble estimated by this HAL-MLE results in an oracle estimator that truly represents a powerful estimator. This is also reflected by the fact that if one selects $J$ equal to the dimension $k$ of the input of the true function, and the realized $J$ estimates   represent an invertible transformation of the input, then the realized oracle estimate actually equals the true function. Even when the library of estimators is very small, but includes one good estimator, then the oracle ensemble will improve upon this estimator by having enormous flexibility to correct its errors. 
 
 Therefore, in great generality, our results demonstrate that the performance of the M-HAL super-learner w.r.t. true function  is all about how well the M-HAL-MLE estimates the oracle ensemble, where the $J$ cross-fitted estimates are treated as fixed and represent the coordinates of the oracle ensemble. In particular, we show that the same applies for any pathwise differentiable target feature: the performance of the target feature of the  M-HAL super-learner w.r.t. target estimand is all about how well the corresponding target feature (i.e., w.r.t. $J$ coordinates) of the meta HAL-MLE estimates the target feature of the oracle ensemble, where the $J$ cross-fitted estimates are treated as fixed coordinates. 
 In this manner, we were able to show that 1) the M-HAL super-learner converges at a rate $n^{-2/3}(\log n)^{d^r}$ to the true function w.r.t. excess risk; and 2) when the sectional variation norm is chosen to satisfy a global undersmoothing criterion, then  target features of the  M-HAL super-learner will be asymptotically linear (efficient or super-efficient) estimators of the target features of the true function.  These two properties equally apply to a targeted M-HAL super-learner in which we use targeted HAL-MLE to estimate the oracle ensemble, thereby enforcing it to solve the efficient influence curve equation for a user supplied set of target features. 

In an upcoming tech report \citep{vanderLaan&Cai20} we also proposed a targeted-HAL-MLE (T-HAL-MLE) defined by enforcing in the definition of the HAL-MLE  an additional constraint, beyond the sectional variation norm bound, defined as the empirical mean of the efficient influence curve equation for the target parameter being equal to zero (or its Euclidean norm small enough). 
This targeting preserves the rate of convergence of the HAL-MLE
and
is now asymptotically  efficient for the target parameter without need for undersmoothing.
It
also remains
asymptotically efficient for other target parameters under the global undersmoothing condition. 
By using the analogue T-HAL-MLE of the true oracle cadlag function at the meta-level, the corresponding T-M-HAL super-learner of our true  function has the same above mentioned asymptotic properties as the T-HAL-MLE. Such meta-level targeted MLE will be important extension of this work, achieving similar asymptotic performance without undersmoothing. 

Of course, one could also decide to use the M-HAL super-learner as a highly powerful initial estimator of the functional parameter in the TMLE of a target parameter, in which case no undersmoothing will be needed
\citep{vanderLaan&Rubin06,vanderLaan08, vanderLaan&Rose11,vanderLaan&Rose18,van2016one}. 
When using as the initial estimator an undersmoothed M-HAL super-learner, we can also design this TMLE to preserve the score equations solved by the initial estimator. 
For example, one can choose to orthogonalize the efficient score with respect to the solved scores, or to jointly solve all the desired score equations with a one-step update along multivariate submodels \citep{van2016one}. 
Similar to an undersmoothed T-M-HAL super-learner, such score-preserving TMLE can also efficiently estimate other smooth features that it did not target. 
In other words, the targeting step of a score-preserving TMLE does not destroy the properties of the undersmoothed M-HAL super-learner for plug-in estimation of features that it did not target. 
Although not the topic of this article, we suggest that undersmoothing the M-HAL super-learner when using it as initial estimator in the TMLE (or undersmoothing the T-M-HAL super-learner)  might also benefit the target feature (even though the estimator is already targeted towards this feature). That is, this undersmoothing  will generally improve the second order behavior of the resulting plug-in estimator  of the target parameter, by solving additional score equations that shrink the size of its exact second order remainder \citep{van2021higher}.

A remaining question is then what we gained relative to using a regular HAL-MLE (or T-HAL-MLE).  However, the transformed coordinates implied by the $J$ cross-fitted functions  will typically allow the oracle ensemble to be a cadlag function with significantly smaller sectional variation norm than the true function as a function of the original coordinates. This is obvious if we select $J=1$ and choose a single super-learner as the ''library'' of estimators, in which case the oracle ensemble is a 1-dimensional function. However, even when we select $J$ larger, if there are good estimators in the library, then, even using the best estimator among the $J$ estimators  as intercept would already mean that  the sectional variation norm only needs to suffice to fit the residual bias. 
 In addition, even when none of the $J$ estimators are good, but they represent an effective coordinate-transformation, gains could be expected. 
 
 Another interesting feature of the M-HAL super-learner is that we did not have to make assumptions on the parameter space of the true function, beyond that the oracle estimator needs to do a good job in approximating the true function. Assuming the latter, this means that we allow the true function to be non-cadlag and/or have infinite sectional variation norm. So, in essence, the only smoothness condition we enforced on the parameter space of the true function is that an oracle ensemble of $J$ candidate estimators has an excess risk that converges at a good rate (e.g, $n^{-1/4}$) to the true function.

The M-HAL super-learner is highly user friendly by not requiring the user to choose a meta-learning step. The only tuning parameter for the meta-learning is the sectional variation norm (i.e, $L_1$-norm in HAL), and that one can be optimally selected with the cross-validation selector for the sake of the function as a whole, and we could target this selector to a collection of target features as well with relatively straightforward undersmoothing criteria.  

An interesting question is if an HAL-MLE using as coordinates the $J$ fitted functions is relatively interpretable if the $J$ estimators are themselves interpretable estimators \citep{valdes2021conditional,barragan2022towards}. The basis functions in HAL are just indicators, suggesting that these indicators evaluated at the cross-fitted functions represent an interpretable basis function. Since the HAL-MLE is just a linear regression model in these basis functions, this suggests that the M-HAL super-learner fit is represented by a linear combination of interpretable basis functions, making it interpretable itself. In this manner, M-HAL super-learner is able to map a set of interpretable algorithms into a very powerful algorithm that is still interpretable. Another feature of HAL-MLE and thereby the M-HAL super-learner is that the HAL-fits have at most $n-1$ non-zero coefficients, again simplifying its interpretation (and fast evaluation of the fitted function).

The application of this M-HAL super-learner goes beyond the treatment specific mean examples, suitable for causal estimation problems with longitudinal interventions \citep{vanderLaan10,vanderLaan10b,petersen2014targeted} and mediation problems with (static or stochastic) interventions across multiple variables \citep{zheng2017longitudinal,wang2023targeted}. We can design coordinate transformations of baseline and time-varying covariates such that the loss function of M-HAL-SL only depends on the reduced data as (\ref{eq:dimr_loss}) and the EIC of the reduced data problem links to the EIC of the original data problem as (\ref{linkDrDstar}). The transformed coordinates lead to ensembles with generally smaller sectional variation norms and thereby potentially more reliable asymptotic linearity, as well as reduced computational costs for other applicable meta-level estimators such as collaborative TMLE \citep{Gruber2010ctmle}, targeted HAL-MLE \citep{vanderLaan&Cai20}, and higher-order spline-HAL \citep{vanderLaan&Benkeser&Cai19}.

An alternative and potentially more flexible approach is to define submodels where the conditional data likelihoods depend on full data only through summary covariates. 
For example, one can data-adaptively define summary covariates with predictors of conditional densities based on conditional hazards with exponential link functions. 
The pre-determined submodels can be viewed as a particular class of model constraints through dimension reduction. 
The plug-in at the projection of the true data distribution onto this submodel can be estimated as a projection target parameter with adaptive TMLE \citep{van2023adaptive}, where the oracle bias of the projection parameter, similar to that of the plug-in at the oracle ensemble of M-HAL-MLE, can be reasonably controlled.
This approach relaxes the conditions that the EICs need to be strictly linked before and after the coordinate transformation, and constructs a richer class of asymptotically linear and possibly super-efficient estimators. 
Future work following this direction is applicable to the analysis of network and single time-series data \citep{ogburn2022causal,malenica2021adaptive} or dimension reduction of other high-dimensional data such as electronic health records (EHR),  imaging, and genomics.

\begin{acks}[Acknowledgments]
This research is funded by NIH-grant R01AI074345-10A1.
\end{acks}

\begin{acks}[Data Availability Statement]
    The simulation data and analysis code that support the findings of this study are openly available in the GitHub repository at \href{https://github.com/zy-wang1/metaHAL}{https://github.com/zy-wang1/metaHAL}, which includes code to generate synthetic data with a structure similar to the real neuroimaging data. The raw task-based fMRI data on thermal pain can be accessed through the following studies \cite{koban2019different,krishnan2016somatic,roy2014representation,wager2013fmri,woo2015distinct};
    derived data supporting the findings of this study are available from the corresponding author Z.W. on request. 
\end{acks}

\clearpage
   
 \begin{appendix}

\subsection*{Supplementary Materials}\label{Appendix} 

\section{Notation for Meta-HAL super-learner  of functional parameter and its target features}
\begin{description}\label{AppendixA1}
\item[$O$:]{Unit data structure/random variable}
\item[$P$:]{Possible probability distribution of $O$}
\item[$Pf=\int f(o)dP(o)$:]{Expectation operator w.r.t. $P$}
\item[$P_0$:]{True probability distribution of $O$}
\item[$O_1,\ldots,O_n$:]{$n$ i.i.d. copies of $O\sim P_0$}
\item[$P_n$:]{Empirical measure of $O_1,\ldots,O_n$}
\item[$P_nf=\int f(o)dP_n(o)=1/n\sum_{i=1}^nf(O_i)$:]{Empirical mean operator}
\item[${\cal M}$:]{Statistical model for $P_0$, set of possible probability distributions including $P_0$}
\item[$Q:{\cal M}\rightarrow {\cal Q}$:]{Functional parameter of interest, where $Q(P):\openr^k\rightarrow [0,1]$ is a $k$-variate $[0,1]$-valued function}
\item[${\cal Q}=\{Q(P):P\in {\cal M}\}$:]{Parameter space of $Q$  consisting of $k$-variate real valued functions}
\item[$\Psi:{\cal M}\rightarrow\openr^k$:]{Euclidean valued target parameter mapping $P$ into $\Psi(P)$, chosen so that $\Psi(P)$ only depends on $P$ through $Q(P)$}
\item[$\Psi(Q)$:]{Alternative notation for $\Psi(P)$, represents a target feature of $Q$}
\item[$D^*(P)$:]{Canonical gradient of pathwise derivative $\left . \frac{d}{d\epsilon}\Psi(P_{\epsilon})\right |_{\epsilon =0}=P D^*(P)S$ of $\Psi$ at $P$ w.r.t. class of paths  $\{P_{\epsilon}:\epsilon\in (-\delta,\delta)\}$ through $P$ with score $S$}
\item[$G:{\cal M}\rightarrow{\cal G}$:]{Functional nuisance parameter $G(P)$ so that $D^*(P)$ only depends on $P$ through $Q(P)$ and $G(P)$. If $D^*(P)$ only depends on $Q(P)$, then $G(P)$ is empty and can be ignored}
\item[$D^*(Q,G)$:]{Alternative notation for $D^*(P)$}
\item[${\cal G}$:]{Parameter space of $G$ defined as ${\cal G}=\{G(P):P\in {\cal M}\}$}
\item[$R_2(P,P_0)$:]{Notation for exact second order remainder $R_2(P,P_0)\equiv \Psi(P)-\Psi(P_0)+P_0D^*(P)$ for $\Psi(P)-\Psi(P_0)$}
\item[$R_{20}(Q,G,Q_0,G_0)$:]{Alternative notation for $R_2(P,P_0)$}
\item[$L(Q)(o)$:]{ Loss function for $Q$ so that $Q_0=Q(P_0)=\arg\min_{Q\in {\cal Q}}P_0L(Q)$}
\item[$d_0(Q,Q_0)=P_0L(Q)-P_0L(Q_0)$:]{Excess risk of $Q$, loss-based dissimilarity}
\item[$L_1(G)(o)$:]{ Loss function for $G$ so that $G_0=G(P_0)=\arg\min_{G\in {\cal G}}P_0L_1(G)$}
\item[$d_{01}(G,G_0)$:]{ Excess risk of $G$, loss-based dissimilarity} 
\item[${\cal M}_{np}$:]{Set of all possible empirical probability measures that can occur as a realization of $P_n$ for any sample size $n$}
\item[$\hat{Q}_j:{\cal M}_{np}\rightarrow{\cal Q}$:]{An estimator that maps an empirical measure $P_n$ (e.g., of training sample)  into an element of ${\cal Q}$, $j=1,\ldots,J$}
\item[$\hat{\bf Q}=(\hat{Q}_j:j=1,\ldots,J)$:]{Collection of $J$ estimators, in context of super-learner it is called the library of $J$ estimators of $Q_0$}
\item[$(v_i,O_i)$:]{Using $V$-fold sample splitting, each observation $O_i$ gets assigned a value $v_i\in \{1,\ldots,V\}$. For each $v$, it defines a $v$-th sample split in validation sample $\{O_i: v_i=v\}$ and training sample $\{O_i:v_i\not =v\}$}
\item[$P_{n,v}^1$:]{Empirical measure of validation sample $\{O_i: v_i=v\}$ (approximately $n/V$ observations) for the $v$-th sample split}
\item[$P_{n,v}$:]{Empirical measure of training sample $\{O_i:v_i\not =v\}$ (approximately $n-n/V$ observations)}
\item[${\bf Q}_{n,v}=\hat{\bf Q}(P_{n,v})$:]{Vector of $J$ estimates based on training sample $P_{n,v}$, $v=1,\ldots,V$}
\item[${\bf Q}_n(v,x)={\bf Q}_{n,v}(x)$:]{ representing the collection of $V$ $J$-dimensional vector of estimates based on training sample $P_{n,v}$, $v=1,\ldots,V$}
\item[$P^{\bar{V}}$:]{Probability measure of $(\bar{V},O)$ implied by $P$ defined by $\bar{V}\sim U\{1,\ldots,V\}$ and conditional probability measure of $O$, given $\bar{V}=v$, equals $P$}
\item[$P^{\bar{V}}f=\frac{1}{V}\sum_{v=1}^V \int f(v,o)dP(o)$]
\item[$P_n^{\bar{V}}$:]{Empirical measure of $(v_i,O_i)$, $i=1,\ldots,n$}
\item[$P_n^{\bar{V}}f=\frac{1}{V}\sum_{v=1}^V P_{n,v}^1 f(v,\cdot)$]
\item[${\cal Q}^r$:]{Collection of real valued cadlag functions $Q^r:[0,1]^J\rightarrow\openr$ with sectional variation norm $\pl Q^r\pl_v^*$ bounded by some $C^u<\infty$. We also refer to this as collection of candidate ensembles of $J$  estimators}
\item[$x_s=(x(j):j\in s)$:]{ subvector of $x\in [0,1]^J$ defined by subset $s\subset\{1,\ldots,J\}$}
\item[$x_{-s}=(x(j):j\not \in s)$]
\item[$y=(x_s,0_{-s})$:]{ vector defined by $y(j)=x(j)$ if $j\in s$ and $y(j)=0$ if $j\not \in s$}
\item[$Q^r_s(x_s)=Q^r(x_s,0_{-s})$:]{ $s$-specific section of $Q^r$ that sets coordinates in complement of $s$ equal to $0$ and is viewed as function on $\mid s\mid$ dimensional $s$-specific edge $E_s\equiv \{x\in [0,1]^J:x_{-s}=0\}$ of $[0,1]^J$. Note $[0,1]^J=\cup_{s\subset\{1,\ldots,J\}}E_s$}
\item[$\pl Q^r\pl_v^*=\mid Q^r(0)\mid+\sum_{s\subset\{1,\ldots,J\}}\int_{(0_s,1_s]}\mid dQ^r_s(u)\mid$:]{ sectional variation norm of function $Q^r:[0,1]^J\rightarrow\openr$}
\item[$Q^r(x)=\sum_{s\subset\{1,\ldots,J\}}\int \phi_{s,x_s}(u)dQ^r_s(u)$:]{ representation of cadlag function $Q^r$ as infinitesimal linear combination of  tensor products of  zero-spline basis functions, $x\rightarrow \phi_{s,x_s}(u)=I(u\leq x_s)$, with knot point $u$. By convention, this sum includes intercept $Q^r(0)$ (corresponding with empty set $s$)}
\item[$Q^r$ {\rm for discrete measures}:] {Note that if $Q^r_s$ generates discrete measure $dQ^r_s$, then $Q^r(x)=\sum_{(s,j)}\beta^r(s,j)\phi_{s,j}(x)$ where $\beta^r(s,j)=dQ_s^r(u_{s,j})$ at support point $u_{s,j}$ of $dQ^r_s$}
\item[$Q^r\circ {\bf Q}_n(v,x)=Q^r({\bf Q}_n(v,x))$:]{ composition of ensemble $Q^r$ with $J$-dimensional vector ${\bf Q}_{n,v}$ of estimated functions based on $P_{n,v}$}
\item[$Q_{0,n}^r=\arg\min_{Q^r\in {\cal Q}^r}P_0^{\bar{V}}L(Q^r\circ{\bf Q}_n)$:]{ oracle ensemble that minimizes the conditional risk $1/V\sum_v P_0 L(Q^r\circ{\bf Q}_{n,v})$}
\item[$Q_{0,n}(v,x)=Q_0^r\circ{\bf Q}_n(v,x)$:]{ $v$-specific oracle estimator defined by applying the oracle ensemble $Q_{0,n}^r$ to the $J$ estimates ${\bf Q}_{n,v}$,   $v=1,\ldots,V$}
\item[$\bar{Q}_{0,n}(x)=\frac{1}{V}\sum_{v=1}^V Q_{0,n}(v,x)$:]{ Single oracle estimator obtained from $Q_{0,n}$ by averaging the $v$-specific oracle estimates $Q_{0,n,v}$ across the $V$ sample splits}
\item[$Q_n^r=\arg\min_{Q^r\in {\cal Q}^r,\pl Q^r\pl_v^*<C_n}P_n^{\bar{V}}L(Q^r\circ{\bf Q}_n)$:]{ M-HAL-MLE of oracle ensemble $Q_{0,n}^r$ using bound $C_n$ for sectional variation norm, minimizing the cross-validated risk $1/V\sum_v P_{n,v}^1 L(Q^r\circ{\bf Q}_{n,v})$ over all  ensemble specific estimators $Q^r\circ \hat{\bf Q}:{\cal M}^{np}\rightarrow{\cal Q}$. It is thus the cross-validation selector for this class of ensemble specific estimators}
\item[$Q_n^{rC}$]{ $Q_n^r$ using $C$ as bound on sectional variation norm}
\item[$Q_n^r=\sum_{(s,j)}\beta_n^r(s,j)\phi_{s,j}$:]{ finite dimensional representation of $Q_n^r$ due to the unrestricted M-HAL-MLE $Q_n^r$ being discrete, or due to choosing it to be discrete on a user  rich set of knot-points}
\item[$C_n$:]{ bound on sectional variation norm enforced in $Q_n^r$. $C_n\geq C_{n,cv}$, where $C_{n,cv}=\arg\min_C 1/V\sum_{v=1}^V P_{n,v}^1 L(Q_n^{r,C})$ is the cross-validation selector}
\item[${\cal J}_n(C_n)$:]{ Set of coefficient-indices $(s,j)$ with $\beta_n^r(s,j)\not =0$}
\item[$C_0^v\equiv\pl Q_{0,n}^r\pl_v^*$:]{ sectional variation norm of oracle selector $Q_{0,n}^r$}
\item[{\rm M-HAL-MLE}:]{ Meta Highly Adaptive Lasso Minimum Loss Estimator $Q_n^r$}
\item[$Q_{n,v}(x)=Q_n(v,x)=Q_n^r\circ {\bf Q}_n(v,x)$:]{ $v$-specific M-HAL super-learner  (of oracle estimator $Q_{0,n,v}=Q_{0,n}(v,\cdot)$) defined by applying the M-HAL-MLE $Q_n^r$ to the $J$-dimensional vector ${\bf Q}_{n,v}$ of estimates, $v=1,\ldots,V$}
\item[$d_0^{\bar{V}}(Q_n,Q_{0,n})=P_0^{\bar{V}}L(Q_n)-P_0^{\bar{V}}L(Q_{0,n})$:]{ excess risk or M-HAL SL relative to oracle estimator $Q_{0,n}$, which is equal to $1/V\sum_v P_0\{L(Q_{n,v})-L(Q_{0,n,v})\}$}
\item[$\bar{Q}_n(x)=\frac{1}{V}\sum_{v=1}^V Q_n(v,x)$:]{ single M-HAL super-learner defined by averaging the $v$-specific M-HAL super-learners $Q_{n,v}=Q_n(v,\cdot)$ across the $V$ sample splits}
\item[ $d_0(\bar{Q}_n,Q_0)$:]{ excess risk of M-HAL SL}
\item[{\rm M-HAL SL:}]{ Meta Highly Adaptive Lasso Super Learner defined by $Q_n$ or by $\bar{Q}_n$}
\item[$P_n^{\bar{V}}D^*(Q_n,G_{0,n}^r)=o_P(n^{-1/2})$:] { the cross-validated efficient influence curve equation  $1/V\sum_v P_{n,v}^1 D^*(Q_{n,v},G_{0,n,v}^r)=o_P(n^{-1/2})$.
 If this equation holds for the M-HAL SL $Q_n$ by selecting $C_n$ large enough in definition of HAL-MLE $Q_n^r$, then a standard analysis shows that $\Psi^{\bar{V}}(Q_n)$ is asymptotically linear}
\end{description}

\section{Notation for equivalent estimation problem implied by treating ${\bf Q}_n$ as a fixed coordinate transformation}\label{AppendixA2}
\subsection{Explanation}
The most important task is the analysis of $Q_n=Q_n^r\circ{\bf Q}_n$ as an estimator of $Q_{0,n}=Q_{0,n}^r\circ{\bf Q}_n$ and its target features 
$\Psi^{\bar{V}}(Q_{0,n})=1/V\sum_{v=1}^V \Psi(Q_{0,n}^r\circ{\bf Q}_{n,v})$. The remaining bias term $d_0(Q_{0,n},Q_0)$ is typically of significantly smaller order (or even 0). 
If we treat ${\bf Q}_n$ as fixed, and view $Q_{n}(v,x)=Q_n^r\circ{\bf Q}_{n,v}(x)$ and $Q_{0,n}(v,x)=Q_{0,n}^r\circ{\bf Q}_{n,v}(x)$ as a function in the new $J$ coordinates ${\bf Q}_n(v,x)$ instead of $(v,x)$, then $Q_n$ becomes $Q_n^r(\cdot)$ and $Q_{0,n}$ becomes $Q_{0,n}^r(\cdot)$.
In particular, $\Psi^{\bar{V}}(Q_{0,n})=\frac{1}{V}\sum_{v=1}^V \Psi(Q_{0,n}^r\circ{\bf Q}_{n,v})$  is now only a target feature  of $Q_{0,n}^r$, and $\Psi(Q_{n})$ is now a target feature of $Q_n^r$.
As a consequence, $Q_n^r$  is  just a regular HAL-MLE of $Q_{0,n}^r$ with this new coordinate transformation $(v,x)\rightarrow {\bf Q}_n(v,x)$ based on data set $(v_i,O_i)$, $i=1,\ldots,n$. 
In addition, the data $(v_i,O_i)$ can be recoded in terms of the new coordinates ${\bf Q}_n$ resulting in a reduction $(v_i,O_i^r)$.   This then teaches us that if ${\bf Q}_n$ would truly be fixed,  all our previous results for standard HAL-MLE based on i.i.d, data, including efficient plug-in estimation of target features, can be applied to this $Q_n^r$ as estimator of $Q_{0,n}^r$.  The reason that the data dependence of ${\bf Q}_n$ does not cause issues is due to the loss at $(v_i,O_i)$ using the transformation ${\bf Q}_{n,v_i}$ based on the training sample excluding $O_i$, allowing conditioning on ${\bf Q}_{n,v}$ whenever dealing with an empirical process w.r.t. $P_{n,v}^1$. 
\subsection{Notation for equivalent estimation problem treating ${\bf Q}_n$ as fixed coordinate transformation}
\begin{description}
\item[$O^r_v=O^r(v,O)$:]{ reduction of $O$ chosen so that the loss $L(Q^r\circ{\bf Q}_{n,v})(o)$ of  candidate $Q^r\circ{\bf Q}_{n,v}\in {\cal Q}$ only  depends on $o$ through $o^r(v,o)$}
\item[$L^r(Q^r)(v,o^r)$:]{ reduced data loss defined by $L^r(Q^r)(v,O^r(v,o))=L(Q^r\circ{\bf Q}_{n,v})(o)$, $v=1,\ldots,V$}
\item[$O^r=O^r(\bar{V},O)$:]{ viewed as a random variable implied by distribution of $(\bar{V},O)\sim P_0^{\bar{V}}$ treating ${\bf Q}_n$ as fixed (i.e., non random fixed functions)}
\item[$P^r$:]{ probability distribution of $(\bar{V},O^r)$ implied by $P^{\bar{V}}$}
\item[$P_0^r$:]{ true probability distribution of $(\bar{V},O^r)$ implied by $P_0^{\bar{V}}$}
\item[$d_0^r(Q^r,Q_{0,n}^r)=P_0^rL^r(Q^r)-P_0^rL^r(Q_{0,n}^r)$]
\item[${\cal M}^r=\{P^r:P\in {\cal M}\}$:]{ statistical model for $(\bar{V},O^r)\sim  P_0^r$, again, treating ${\bf Q}_n$ as fixed}
\item[$P_n^r$:]{ empirical measure of $(v_i,O^r_i=O^r(v_i,O_i))$, $i=1,\ldots,n$}
\item[$Q^r:{\cal M}^r\rightarrow {\cal Q}^r$:]{ $Q^r(P^r)=\arg\min_{Q^r\in {\cal Q}^r}P^r L^r(Q^r)$}
\item[$Q^r_0=Q^r(P_0^r)$]
\item[$\Psi^r:{\cal M}^r\rightarrow\openr$:]{ defined by $\Psi^r(P^r)=\frac{1}{V}\sum_{v=1}^V \Psi(Q^r\circ{\bf Q}_{n,v})$}
\item[$\Psi^r(Q^r)$:]{ alternative notation for $\Psi^r(P^r)$ to emphasize it only depends on $P^r$ through $Q^r$}
\item[$D^r(P^r)(\bar{V},O^r)$:]{ canonical gradient of $\Psi^r$ at $P^r$}
\item[$G^r:{\cal M}^r\rightarrow {\cal G}^r$:]{ nuisance parameter so that $D^r(P^r)$ only depends on $P^r$ through $Q^r$ and $G^r$. $G_{0,n}^r=G^r(P_0^r)$ true nuisance parameter. It is chosen so that $G^r_v:{\cal M}^r_v\rightarrow {\cal G}$ maps distribution $P^r_v$ of $O^r_v$ into parameter space ${\cal G}$ of $G$}
\item[$D^r(Q^r,G^r)$:]{ alternative notation for canonical gradient $D^r(P^r)$}
\item[$R_{20}(Q^r,G^r,Q_{0,n}^r,G_{0,n}^r)\equiv \Psi^r(Q^r)-\Psi^r(Q_{0,n}^r)+P_0^rD^r(Q^r,G^r)$:]{ exact second order remainder for $\Psi^r$}
\item[$P_n^r D^r(Q_n^r,G^r_{0,n})=o_P(n^{-1/2})$:]{ efficient influence curve equation that would be solved by HAL-MLE $Q_n^r$ when selecting $C_n$ large enough}
\end{description}

\subsection{Equivalences}
We now have the following equivalences between our estimation problem and corresponding  fixed ${\bf Q}_n$-formulation of the estimation problem defined above in model ${\cal M}^r$:
\begin{eqnarray*}
Q^r_{0,n}&=&\arg\min_{Q^r\in {\cal Q}^r}P_0^r L^r(Q^r)=\arg\min_{Q^r\in {\cal Q}^r}P_0^{\bar{V}}L(Q^r\circ{\bf Q}_n) \\
Q^r_n&=&\arg\min_{Q^r\in {\cal Q}^r,\pl Q^r\pl_v^*<C_n}P_n^r L^r(Q^r)\\
&=&\arg\min_{Q^r\in {\cal Q}^r,\pl Q^r\pl_v^*<C_n}P_n^{\bar{V}}L(Q^r\circ{\bf Q}_{n,\bar{V}})\\
\Psi^{\bar{V}}(Q_{0,n})&=& \Psi^r(Q_{0,n}^r)\\
\Psi^{\bar{V}}(Q_n)&=&\Psi^r(Q_n^r)\\
d_0^{\bar{V}}(Q_n,Q_{0,n})&=&d_0^r(Q_n^r,Q_{0,n}^r)\\
D^r(Q^r,G^r)(\bar{V},O^r)&=& D^*(Q^r\circ{\bf Q}_{n,\bar{V}},G^r_{\bar{V}})(O)\mbox{ by assumption (\ref{r1})}\\
R_{20}^r(Q^r,G^r,Q_{0,n}^r,G_{0,n}^r)&=&\frac{1}{V}\sum_{v=1}^V R_{20}(Q^r\circ{\bf Q}_{n,v},G^r_v,Q^r_{0,n}\circ{\bf Q}_{n,v},G^r_{0,n,v})\\
P_n^r D^r(Q^r_n,G^r_{0,n})&=& P_n^{\bar{V}}D^*(Q_n,G^r_{0,n})
\end{eqnarray*}

\section{Proofs for Convergence Rate of M-HAL-SL}\label{sec:proofs_rate}

{\bf Proof of Lemma \ref{lemmarateQnr}:}
We have
\begin{eqnarray*}
0&\leq & d_0^{\bar{V}}(Q_n,Q_{0,n})\\
&=& \frac{1}{V}\sum_{v=1}^V P_0 \{L(Q_{n,v})-L(Q_{0, n, v}) \} \\
&=&-\frac{1}{V}\sum_{v=1}^V (P_{n,v}^1-P_0)  \{L(Q_{n,v})-L(Q_{0, n, v})\}\\
&&+\frac{1}{V}\sum_{v=1}^V P_{n,v}^1 \{L(Q_{n,v})-L(Q_{0, n, v}\}\\
&\leq&-\frac{1}{V}\sum_{v=1}^V (P_{n,v}^1-P_0) \{L(Q_{n,v})-L(Q_{0, n, v})\},
\end{eqnarray*}
where the last inequality is by definition of $Q_n^r$ and thus $Q_n$.
For a given $v$, conditional on the training sample so that ${\bf Q}_{n,v}$ is fixed, we have \[
\begin{array}{l}
\mid(P_{n,v}^1-P_0)\{L(Q_{n,v})-L(Q_{0, n, v})\}\mid \\
 \leq 
\sup_{Q^r, Q^r_1\in {\cal Q}^r}\mid (P_{n,v}^1-P_0)\{L(Q^r\circ{\bf Q}_{n,v})-L(Q^r_1\circ{\bf Q}_{n,v})\mid \\
= \sup_{Q^r, Q^r_1\in {\cal Q}^r} \mid (P_{n,v}^{r,1}-P_{0,v}^r)\{L^r(Q^r)-L^r(Q^r_1)\}\mid
\end{array}
\]
For a given $v$,  
${\cal F}^r_1\equiv \{L^r(Q^r) - L^r(Q^r_1): Q^r, Q^r_1\in {\cal Q}^r\}$ are cadlag functions of $O^r(v,O)$ with a universal bound on its sectional variation norm. Let $d^r$ be the dimension of $O^r$.
The typical HAL-MLE proof now proceeds with  the following ingredients: 1) 
${\cal F}^r_1$ is a Donsker class with bracketing entropy number 
$ \log N_{[]}(\epsilon,{\cal F}^r_1,L^2(P^r))\lesssim  \epsilon^{-1}(\log \epsilon)^{-d}$ 
(Proposition 2 in \citep{Bibaut&vanderLaan19});
2) for each $v$, $P_{0,v}^r\{L^r(Q^r)-L(Q_{1}^r)\}^2\leq M_2^r d_0^r(Q^r,Q_{1}^r)$;  3) the bracketing entropy integral is bounded as follows $J_{[]}(\delta,{\cal F}^r_1,L^2(P^r))\lesssim \delta^{1/2}(\log \delta)^{-d/2}$ \citep{Bibaut&vanderLaan19}; the modulus of discontinuity for the empirical process can be bounded accordingly as \[
\sup_{f,\pl f\pl<\delta}\mid G_n(f)\mid\lesssim
J_{[]}(\delta,{\cal F}^r_1,L^2(P^r))\left (1+\frac{J_{[]}(\delta,{\cal F}^r_1,L_2(P^r))}{\delta^2n^{1/2}} M\right )\]
 (\citet{vanderVaart&Wellner11} and Lemma  3.4.2  in  \citet{vanderVaart&Wellner96}). This allows us to apply the iterative HAL-MLE proof in Appendix of \citep{vanderLaan15}   or direct proof \citep{Bibaut&vanderLaan19} to establish that $d_0(Q_n,Q_{0,n}) =O_P(n^{-2/3}(\log n)^{d^r})$. For example, \citep{Bibaut&vanderLaan19} gives this result as a corollary. 

$\Box$

\section{Proofs for Asymptotic Linearity Theorem}\label{sec:proofs_AL}

Let $J(\delta, \mathcal{F}, L_2)$ denote the uniform entropy integral for a class $\mathcal{F}$.
Define $\mathbb{G}_{n}f = \sqrt{n} (P_{n} - P_0)f$ so that $\mathbb{G}_{\frac{n}{V}, v}f = \sqrt{n/V} (P_{n, v}^1 - P_0)f$. 
Assume that the base learner algorithms converge to some $\pmb Q^*$ such that $\norm{\pmb Q_{n, v} - \pmb Q^*}_{P_0}^2 \overset{P}{\to} 0$ for all $v$. 
Define
\begin{align*}
    Q_0^r = \arg\min_{Q^r} \frac{1}{V} \sum_{v = 1}^V P_0 L(Q^r \circ \pmb Q^*).
\end{align*}
Note that $G^r$ depends on $O$ and $\pmb Q_{n}(\bar V, X)$ through a transformation of $O^r(\bar V, O)$ so that $G^r \in \mathcal{G}^r$ and $\pmb Q_{n, v}(X)$ define $G^r_v \in \mathcal{G}$. Use notation $G^r(\pmb Q_{n, v}) \equiv G_v^r$ to highlight the definition of $G^r_v$ using $G^r$ and $\pmb Q_{n, v}(X)$. 
For each $G^r \in \mathcal{G}^r$, 
replace $\pmb Q_{n, v}(X)$ with $\pmb Q^*(X)$ in the definition of $G^r_v$ to define $G^r_*(\pmb Q^*) \in \mathcal{G}$. 
Also assume there exists $G_0^r \in \mathcal{G}^r$ such that 
$\norm{G_{0, n, v}^r - G_{0, *}^r(\pmb Q^*)}_{P_0}^2 \overset{P}{\to} 0$ for all $v$. 

\begin{lemma}\label{lemma:AL1}
Define
\begin{align*}
    f(Q^r, G^r, \pmb Q) = & D^*(Q^r \circ \pmb Q, G^r(\pmb Q)) - D^*(Q_0^r \circ \pmb Q, G_{0, *}^r(\pmb Q)) \\
    f_{n, v} = & f(Q_n^r, G_{0, n}^r, \pmb Q_{n, v}) = D^*(Q_n^r \circ \pmb Q_{n, v}, G^r_{0, n, v}) - D^*(Q_{0}^r \circ \pmb Q_{n, v}, G_{0, v}^r) \\
    \mathcal{F}_{n, v} = & \{f(Q^r, G^r, \pmb Q_{n, v}): Q^r \in \mathcal{Q}^r, G^r \in \mathcal{G}^r\} \\
    \mathcal{F} = & \cup \{\mathcal{F}_{n, v}: n = 1, 2, \dots; v = 1, \dots, V\}. 
\end{align*}
Assume the following regularity conditions.
\begin{enumerate}
    \item $Q_0^r = \arg\min_{Q^r} \frac{1}{V} \sum_{v = 1}^V P_0 L(Q^r \circ \pmb Q^*)$ for some limit $\pmb Q^*$ such that $\norm{\pmb Q_{n, v} - \pmb Q^*}_{P_0}^2 \overset{P}{\to} 0$ for all $v$. In addition, there exists $G_0^r \in \mathcal{G}^r$ such that $\norm{G_{0, n, v}^r - G_{0, *}^r(\pmb Q^*)}_{P_0}^2 \overset{P}{\to} 0$ for all $v$. Lastly, $P_0 f_{n, v}^2 \overset{P}{\to} 0$ uniformly for all $v$, \label{cond:conv}
    \item $\mathcal{F}$ is a $P_0$-measurable class with envelope function $F$ such that $|f| < F < c < \infty$ for some constant $c$ and for all $f \in \mathcal{F}$, \label{cond:env}
    \item there exists a positive sequence $1 > \delta_n \to 0$ at a slow rate such that $\frac{1}{\log(n) \delta_n^2 } \to 0$ and $\frac{P_0 \sup_{f\in \mathcal{F}_{n, v}} f^2}{\delta_n^2} \overset{P}{\to} 0$ (uniformly over $v$), and $\sum_{v = 1}^V J(\delta_n, \mathcal{F}_{n, v}, L_2) \overset{P}{\to} 0$. \label{cond:seq}
\end{enumerate}
Then,  we have that for fixed $V$, 
\begin{align*}
    &\frac{1}{V}\sum_{v=1}^V (P_{n,v}^1-P_0) \Big\{
    D^*(Q_n^r\circ{\bf Q}_{n,v},G_{0,n, v}^r)
    - D^*(Q_0^r\circ{\bf Q}_{n,v},G_{0, v}^r)
    \Big\} \\
    \equiv & \frac{1}{\sqrt{nV}}\sum_{v=1}^V \mathbb{G}_{\frac{n}{V}, v} f_{n, v}
    = o_P(n^{-1/2}). 
\end{align*}
\end{lemma}

\begin{proof}
Define $f_{n, v}^* = f_{n, v} \indicator{f_{n, v}^2 \leq \delta_n^2 P_0F^2}$. Then $F$ is also an envelope for $\mathcal{F}_{n, v}^* \equiv \{f\indicator{f^2 \leq \delta_n^2 P_0F^2}: f \in \mathcal{F}_{n, v}\}$, and $J(\delta, \mathcal{F}_{n, v}^*, L_2) \leq J(\delta, \mathcal{F}_{n, v}, L_2)$. 
Using conditions \ref{cond:env} and \ref{cond:seq}, by Theorem 2.1 of \cite{van2011local}, we have
\begin{align*}
    E^*_{P_0} \norm{\mathbb{G}_{\frac{n}{V}, v}}_{\mathcal{F}_{n, v}^*}
    & \lesssim J(\delta_n, \mathcal{F}_{n, v}, L_2) (1 + \frac{J(\delta_n(\log(1/\delta_n))^{1/p}, \mathcal{F}_{n, v}, L_2)}{\delta_n^2\sqrt{n/V} \norm{F}_{P_0} }) \norm{F}_{P_0} \\
    & = O_P(J(\delta_n, \mathcal{F}_{n, v}, L_2)) = o_P(1). 
\end{align*}
By symmetrization and Khintchine's inequality, we have $E^*_{P_0} \norm{\mathbb{G}_{\frac{n}{V}, v} }_{\mathcal{F}_{n, v}} \lesssim \norm{F}_{P_0}$, and then
\begin{align*}
    E^*_{P_0} \sup_{f_{n, v}\in\mathcal{F}_{n, v}} |\mathbb{G}_{\frac{n}{V}, v}f_{n, v}\indicator{f_{n, v}^2 > \delta_n^2P_0F^2}|
    & \lesssim P_0^*(\sup_{f_{n, v}\in\mathcal{F}_{n, v}} f_{n, v}^2 > \delta_n^2P_0F^2) 
    \lesssim \frac{P_0 \sup_{f\in \mathcal{F}_{n, v}} f^2}{\delta_n^2} = o_P(1).
\end{align*}
Therefore, $E^*_{P_0}\norm{\mathbb{G}_{\frac{n}{V}, v}}_{\mathcal F_{n, v}} = o_P(1)$.
Combining this with the conditions 1-3, we have 
$\mathbb{G}_{\frac{n}{V}, v} f_{n, v}$ weakly converges to $0$, which implies 
$\frac{1}{\sqrt{nV}} \sum_{v = 1}^V \mathbb{G}_{\frac{n}{V}, v} f_{n, v} = o_P(n^{-1/2})$.
\end{proof}

\begin{remark}
    Condition \ref{cond:seq} is a much weaker assumption than a Donsker class condition over $\{D^*(Q, G): Q\in\mathcal{Q}, \mathcal{G}\}$ for the original data problem. 
    For example, if for fixed $\pmb Q_{n, v}$, $\{f(Q^r, G^r, \pmb Q_{n, v}): Q^r \in \mathcal{Q}^r, G^r \in \mathcal{G}^r\}$ is a class of cadlag functions with a sectional variation norm bound not depending on $\pmb Q_{n, v}$ --- which is a reasonable assumption due to the bounded sectional variation norms of $\mathcal{Q}^r, \mathcal{G}^r$ --- then $J(\delta_n, \mathcal{F}, L_2) = O_P(\sqrt{\delta_n}) = o_P(1)$ uniformly across all possible $\pmb Q_{n, v}$. Such (uniform) Donsker class conditions defined on the meta level may be easier to achieve, and have an advantage when the original data problem is complex and all reasonable initial estimators are highly variable. 
\end{remark}

\begin{lemma}\label{lemma:AL2}
Recall the conditions of Lemma \ref{lemma:AL1}. Additionally, assume the following conditions. 
\begin{enumerate}[start=4]
    \item there exist uniform convergence limits $\tilde Q_0 \in \mathcal{Q}$ and $\tilde G_0 \in \mathcal{G}$: for all $\epsilon > 0$, there exists $N$ such that for all $n \geq N$, $v = 1, \dots, V$, and all realizations of $\pmb Q_{n, v}$, 
    \begin{align*}
   P_0 \{D^*(Q_{0, n}^r\circ{\bf Q}_{n,v},G_{0, n, v}^r ) 
    - D^*(\tilde Q_0, \tilde G_0)\}^2 < \epsilon; 
    \end{align*}
    denote $h(\pmb Q) = D^*(Q_0^r\circ{\bf Q},G_{0, *}^r(\pmb Q))
    - D^*(\tilde Q_0, \tilde G_0)$, 
    \item $\sup_{\pmb Q_{n, v}} \norm{h(\pmb Q_{n, v})}_{\infty} \leq M_n$ for some $M_n < \infty$, and $\lim_{n \to \infty} M_n < \infty$ or $M_n \to \infty$ at a rate slower than $\sqrt{n}$ so that $M_n/\sqrt{n} \to 0$.
\end{enumerate}
Then we have 
\begin{align*}
    \frac{1}{V}\sum_{v=1}^V (P_{n,v}^1-P_0) \Big\{
    D^*(Q_0^r\circ{\bf Q}_{n,v},G_{0, v}^r) 
    - D^*(\tilde Q_0, \tilde G_0)
    \Big\}
    = \frac{1}{\sqrt{nV}} \sum_{v = 1}^V \mathbb{G}_{\frac{n}{V}, v} h(\pmb Q_{n, v})
    = o_P(n^{-1/2}). 
\end{align*}
\end{lemma}

\begin{proof}
    By Bernstein's Inequality, for all $a > 0$, 
    \begin{align*}
        E_{P_0}(\indicator{|\mathbb{G}_{n/V, v}h(\pmb Q_{n, v})| > a} | \pmb Q_{n, v}) 
        & \leq 2 \exp \{-\frac{1}{4} \frac{a^2}{E_{P_0}[ h(\pmb Q_{n, v})^2 | \pmb Q_{n, v}] + a M/\sqrt{n/V}}\}
    \end{align*}
    Due to the uniform bound on $h(\pmb Q_{n, v})$ across all realizations of $\pmb Q_{n, v}$, by double expectations, for all small enough $\epsilon > 0$ and $\lambda = -\frac{a^2}{4\log(\epsilon/2)} > 0$, there exists $N$ such that for all $n \geq N$, $aM_n / \sqrt{n/V} \leq \lambda / 2$, and
    \begin{align*}
        P_0(|\mathbb{G}_{n/V, v}h(\pmb Q_{n, v})| > a) \leq E_{P_0} 2\exp\{-\frac{1}{4} \frac{a^2}{\frac{\lambda}{2} + \frac{\lambda}{2}}\} \leq \epsilon. 
    \end{align*}
\end{proof}

We note that, so far, we have not assumed any Donsker class assumption on ${\bf Q}_n$ (we only relied on the oracle selected ensemble of ${\bf Q}_n$ to be a good estimator of $Q_{0,n}$). This again demonstrates that the asymptotic properties of M-HAL-MLE features do not rely on a Donsker class assumption on $\mathcal{D}^* \equiv \{D^*(Q, G): Q \in \mathcal{Q}, G \in \mathcal{G}\}$, but on a Donsker class condition on the meta level with fixed $\pmb Q_{n, v}$ and only a uniform boundedness condition across realizations of $\pmb Q_{n, v}$. 
When assumption (\ref{r1}) holds and $\tilde Q_0 = Q_0$, the two lemmas above lead to the desired asymptotic linearity result in Theorem \ref{theoremaslinmhalmle}: 
\[
\Psi^r(Q_n^r)-\Psi^r(Q^r_{0,n})=P_n D^*(Q_0,\tilde{G}_0)+o_P(n^{-1/2}).\]

\section{Proofs for Difference of Target Feature of Oracle Estimator and Target Estimand}\label{sec:proofs_diff}

{\bf Proof of Theorem \ref{theoremoracleensemble1}:}
Recall that $D^r(P^r)=D^r(Q^r)$ is the canonical gradient of $\Psi^r:{\cal M}^r\rightarrow\openr$. Thus, $P_0^r D^r(Q_{0,n}^r)=0$. We also have that $D^r(Q_{0,n}^r)(v,O^r(v,o))=D^*(Q_{0,n}^r\circ{\bf Q}_{n,v})(o)$, so that $P_0^rD^r(Q_{0,n}^r)=0$ implies
$\frac{1}{V}\sum_{v=1}^V P_0 D^*(Q_{0,n}^r\circ{\bf Q}_{n,v})=0$.
Now, use that, by definition of $R_{20}$,  \[
\Psi(Q^r_{0,n}\circ{\bf Q}_{n,v})-\Psi(Q_0)=-P_0D^*(Q^r_{0,n}\circ{\bf Q}_{n,v})+R_{20}(Q^r_{0,n}\circ{\bf Q}_{n,v},Q_0).\] Recall notation $Q_{0,n,v}=Q_{0,n}^r\circ{\bf Q}_{n,v}$. Thus, this proves
\begin{eqnarray*}
\left\{ \frac{1}{V}\sum_v \Psi(Q^r_{0,n}\circ{\bf Q}_{n,v})-\Psi(Q_0)\right\}&=&
-\frac{1}{V}\sum_{v=1}^V P_0 D^*(Q_{0,n,v})\\
&&\hspace*{-3cm} +\frac{1}{V}\sum_{v=1}^V R_{20}(Q_{0,n,v},Q_0)= 
\frac{1}{V}\sum_{v=1}^V R_{20}(Q_{0,n,v},Q_0).
\end{eqnarray*}
By assumption (\ref{r1}), we have $\frac{1}{V}\sum_v R_{20}(Q_{0,n,v},Q_0)=O(d_0(Q_{0,n},Q_0))$.
$\Box$

In our treatment specific mean example we could have defined $G_{0,n,v}^{*}$ as the conditional mean of $A$, given ${\bf Q}_{n,v}(W),{\bf Q}_0(W)$, where ${\bf Q}_0$ represents the limit of ${\bf Q}_n$, while  $G_{0,n,v}^r=E_0(A\mid {\bf Q}_{n,v}(W))$. As in our example, we then note that $G_{0,n,v}^r$ and $G_{0,n,v}^{*}$ are conditional expectations in which $G_{0,n,v}^{*}$ conditions on an extra ${\bf Q}_0(W)$ coming from the limit of ${\bf Q}_{n,v}$. As in our example, we can then utilize that conditional expectations are projection operators to bound the  $L^2$-norm of the difference in terms  of  ${\bf Q}_{n,v}$ and ${\bf Q}_0$. We improved on this approach by selecting $G_{0,n,v}^*$ as the conditional mean of $A$, given $Q_{0,n,v}(W)$, $Q_0(W)$, while still following the same subsequent steps. This allowed us to bound the $L^2$-norm of the difference of $G_{0,n,v}^*-G_{0,n,v}^r$ in terms of the difference of the oracle ensemble $Q_{0,n}^r\circ{\bf Q}_{n,v}$ of ${\bf Q}_{n,v}$ and $Q_0$,  instead of a difference of the $J$-dimensional ${\bf Q}_{n,v}$ versus ${\bf Q}_0$. In this manner we obtained a natural bound  $d_0^{\bar{V},1/2}(Q_{0,n},Q_0)$ for the $L^2$-norm of $G_{0,n,v}^*-G_{0,n,v}^r$. This insight clearly suggests that one  should aim to define $G_{0,n,v}^*$ with minimal conditioning, even though conditioning on ${\bf Q}_{n,v}$ and ${\bf Q}_0$ would suffice.

{\bf Proof of Theorem \ref{theoremoracleensemble}:}
We have
\begin{eqnarray*}
\left\{ \frac{1}{V}\sum_v \Psi(Q_{0,n}^r\circ{\bf Q}_{n,v})-\Psi(Q_0)\right\}
&=&
-\frac{1}{V}\sum_v P_0D^*(Q_{0,n}^r\circ{\bf Q}_{n,v},G_0)\\
&&+\frac{1}{V}\sum_v R_{20}(Q_{0,n}^r\circ{\bf Q}_{n,v},G_0,Q_0,G_0).
\end{eqnarray*}
By assumption (\ref{r1}),  $\frac{1}{V}\sum_v R_{20}(Q_{0,n}^r\circ{\bf Q}_{n,v},G_0,Q_0,G_0)=O_P(d_0^{\bar{V}}(Q_{0,n},Q_0))$.
It remains to analyze $\frac{1}{V}\sum_v P_0D^*(Q_{0,n}^r\circ{\bf Q}_{n,v},G_0)$. 
By definition of ${G}_{0,n,v}^{*}$ we have
\[
\frac{1}{V}\sum_v P_0D^*(Q_{0,n}^r\circ{\bf Q}_{n,v},G_0)=\frac{1}{V}\sum_v P_0D^*(Q_{0,n}^r\circ{\bf Q}_{n,v},G_{0,n,v}^{*}).\]
So now it remains to analyze $1/V\sum_v P_0 D^*(Q_{0,n}^r\circ{\bf Q}_{n,v},G_{0,n,v}^{*})$. 

We have $0=P_0^rD^r(Q_{0,n}^r,G_{0,n}^r)$, which equals $1/V\sum_{v=1}^V P_0D^*(Q_{0,n}^r\circ{\bf Q}_{n,v},G^r_{0,n,v})$. Thus, it follows that $1/V\sum_v P_0 D^*(Q_{0,n}^r\circ{\bf Q}_{n,v},G^r_{0,n,v})=0$.
Subtracting this from our expression yields the following expression for $\Psi^r(Q_{0,n}^r)-\Psi(Q_0)$
\[
\begin{array}{l}
 \frac{1}{V}\sum_v \Psi(Q_{0,n}^r\circ {\bf Q}_{nv})-\Psi(Q_0)\\
=
\frac{1}{V}\sum_v P_0 \left\{ D^*(Q_{0,n}^r\circ {\bf Q}_{n,v},G_{0,n,v})-D^*(Q_{0,n}^r\circ{\bf Q}_{n,v},G_{0,n,v}^{*})\right\}+O_P(d_0^{\bar{V}}(Q_{0,n},Q_0))\end{array}
\]
Recall $Q_{0,n,v}=Q_{0,n}^r\circ{\bf Q}_{n,v}$. Using  $\Psi(Q_{0,n,v})-\Psi(Q_0)=-P_0D^*(Q_{0,n,v},G_{0,n,v})+R_2(Q_{0,n,v},G_{0,n,v},Q_0,G_0)$, and $\Psi(Q_{0,n,v})-\Psi(Q_0)=-P_0D^*(Q_{0,n,v},G_{0,n,v}^*)+R_2(Q_{0,n,v},G_{0,n,v}^*,Q_0,G_0)$,  it follows that the leading term on right-hand side above equals $r_n$.
$\Box$

\section{Undersmoothing Conditions} \label{sec:undersmoothing}

 Here we will prove a formal Theorem \ref{thscoreequationtarget} establishing sufficient conditions for $P_n^rD^r(Q_n^r,G_{0,n}^r)=o_P(n^{-1/2})$, involving an undersmoothing condition (\ref{assumptionatarget}) that will hold by selecting $C_n$ large enough, and a condition (\ref{a3}). The latter is shown to be about the linear span of the basis functions in $Q_n^r$ with non-zero coefficients approximating a function having to do with $G_{0,n}^r$. 
 
  Recall $Q_n^r=\arg\min_{Q^r\in {\cal Q}^r,\pl Q^r\pl_v^*<C_n}P_n^r L^r(Q^r)$.
  Consider a path $\{Q^{r,h}_{n,\epsilon}:\epsilon\}$,  indexed by a function $h$, defined by \begin{equation}\label{halmlepaths}
 Q^{r,h}_{n,\epsilon}=(1+\epsilon h(0))Q_n^r(0)+\sum_{s\subset\{1,\ldots,d\}}
 \int_{(0_s,x_s]}(1+\epsilon h(s,u_s)) dQ^r_{n,s}(u_s),\end{equation}
 where $h$ has to satisfy the  restriction $r(h,Q_n^r)=0$ defined by \[
 r(h,Q_n^r)\equiv  h(0)\mid Q_n^r(0)\mid +\sum_{s\subset\{1,\ldots,d\}}
 \int_{(0_s,\tau_s]}(1+\epsilon h(s,u_s)) \mid dQ^r_{n,s}(u_s)\mid .\] 
 Due to the constraint $r(h,Q_n^r)$,  it follows that for any uniformly bounded function $h$ with $r(h,Q_n^r)$, $\{Q^{r,h}_{n,\epsilon}:\epsilon\} \subset {\cal Q}^r(C_n)$. Specifically, the sectional variation norm  of $Q^{r,h}_{n,\epsilon}$ does not change as $\epsilon$ moves away from zero locally.
 Consider the score equation for $Q_n^r$ of the empirical risk it minimizes: \begin{equation}\label{scorefn}
 S_h(Q_n^r)\equiv  \frac{d}{d\epsilon}\left .  L^r(Q^{r,h}_{n,\epsilon}) \right |_{\epsilon =0}.\end{equation}
 By (\ref{r1}) (on loss $L^r$), it also follows that $\{S_h(Q^r): Q^r\in {\cal Q}^r\}\subset {\cal D}_{d^r,v}[0,\tau^r]$.

Since $Q_n^r$ minimizes this empirical risk over all $Q^r\in {\cal Q}^r(C_n)$,  we know that  $P_n^r S_h(Q_n^r)=0$ for all uniformly bounded $h$ with $r(h,Q_n^r)=0$.
Thus, we have $P_n^r S_h(Q_n^r)=0$ for all bounded $h$ with $r(h,Q_n^r)=0$.
Let ${\cal S}(Q_n^r)=\{S_h(Q_n^r): h\}$ be the linear span of all these score functions $S_h(Q_n^r)$ indexed by any bounded function $h$.

Analogue to the proof of the theorems in \citep{vanderLaan&Benkeser&Cai19}  we can now establish that for large enough $C_n$, the linear span of the score equations $P_n^r S_h(Q_n^r)$ with $r(h,Q_n^r)=0$ approximate  the efficient influence curve equation $P_n^r D^r(Q_n^r,G_{0,n}^r)$.
This works as follows. Let $D^r_n(Q_n^r,G_{0,n}^r)(v,o^r)=D^*_n(Q^r\circ{\bf Q}_{n,v},G^r_{0,n,v})(o)$ be an approximation of $D^r(Q_n^r,G_{0,n}^r)$ that is contained in ${\cal S}(Q_n^r)=\{S_h(Q_n^r):h\}$, without the restriction $r(h,Q_n^r)=0$.
Let $h^*(Q_n^r,G_{0,n}^r)$ be the corresponding index so that $D^r_n(Q_n^r,G_{0,n}^r)=S_{h^*(Q_n^r,G_{0,n}^r)}(Q_n^r)$. For notational convenience, in this proof let's denote it with $h^*$. Recall $Q_n^r=\sum_{(s,j)\in {\cal J}_n(C_n)}\beta_n^r(s,j)\phi_{s,j}$.
Let $\tilde{h}(s,j)=h^*(s,j)$ except at one $(s^*,j^*)\in {\cal J}_n(C_n)$ and defined such that
$r(\tilde{h},Q_n^r)=\sum_{(s,j)\in {\cal J}_n(C_n)}\tilde{h}(s,j)\mid \beta_n^r(s,j)\mid =0$. Thus, $\tilde{h}(s^*,j^*)=-\frac{\sum_{(s,j)\not =(s^*,j^*)} h^*(s,j)\mid \beta_n^r(s,j)\mid}{\mid \beta_n^r(s^*,j^*)\mid}$.
Then $P_n^rS_{\tilde{h}}(Q_n^r)=0$. We now want to choose $(s^*,j^*)$ such that $P_n^r(S_{h^*}(Q_n^r)-S_{\tilde{h}}(Q_n^r))$ minimal, so that subsequently setting it smaller than $o(n^{-1/2})$ yields the global undersmoothing criterion. 

Note that
\[
P_n^r(S_{\tilde{h}}(Q_n^r)-S_{h^*}(Q_n^r))=P_n^r\frac{d}{dQ_n^r}L^r(Q_n^r)\left(\sum_{(s,j)}(\tilde{h}-h^*)(s,j)\beta_n(s,j)\phi_{s,j}\right).
\]
We have $\sum_{(s,j)}(\tilde{h}-h^*)(s,j)\beta_n^r(s,j)\phi_{s,j}=
c_n(s^*,j^*)\phi_{s^*,j^*}$ with
\[
c_n(s^*,j^*)=-\beta_n^r(s^*,,j^*)\left\{\frac{\sum_{(s,j)\not =(s^*,j^*)}h^*(s,j)\mid \beta_n^r(s,j)\mid}{\mid \beta_n^r(s^*,j^*)\mid} + h^*(s^*,j^*)\right\}.\]
Note that $c_n(s^*,j^*)$ is bounded by $\sum_{(s,j)}\mid  h^*(s,j)\mid \mid \beta_n^r(s,j)\mid$ which is thus bounded by  $\pl h^*\pl_{\infty}C_n$ (using $\sum_{(s,j)}\mid \beta_n^r(s,j)\mid =C_n$). Thus, under this trivial assumption we have $c_n(s^*,j^*)=O_P(1)$. So we have obtained:
\begin{eqnarray*}
P_n^r(S_{\tilde{h}}(Q_n^r)-S_{h^*}(Q_n^r))&=&c_n(s^*,j^*)P_n^r \frac{d}{dQ_n^r}L^r(Q_n^r)(\phi_{s^*,j^*})\\
&=&O_P\left (C_n P_n^r\frac{d}{dQ_n^r}L^r(Q_n^r)(\phi_{s^*,j^*})\right ).
\end{eqnarray*}
Therefore, our undersmoothing condition can be chosen to be $C_n \min_{(s,j)\in {\cal J}_n(C_n)}\mid P_n^r \frac{d}{dQ_n^r}L^r(Q_n^r)(\phi_{s^*,j^*})=o(n^{-1/2})$, which then implies $P_n^rD^r_n(Q_n^r,G_{0,n}^r)=o(n^{-1/2})$. 
This completes the proof of the first part of the next theorem, while the remaining part provides the extra condition (\ref{a3}) that makes $P_n^r (D^r_n-D^r)(Q_n^r,G_{0,n}^r)=o_P(n^{-1/2})$, so that we also obtain $P_n^r D^r(Q_n^r,G_{0,n}^r)=o_P(n^{-1/2})$ as well.

\begin{theorem}\label{thscoreequationtarget}
\ \nl
{\bf Definitions:} \nl
 Consider HAL-MLE ensemble   $Q_n^r=\arg\min_{Q^r\in {\cal Q}^r(C_n), f\ll^*\mu_n}P_n^r L^r(Q^r)$  for some  selector $C_n$ with $C_{n,cv}\leq C_n\leq C^u<\infty$ with probability tending to 1, where $C_{n,cv}$ is the cross-validation selector of $C$.
  Recall that $L^r(Q^r)(v,o^r)=L(Q^r\circ{\bf Q}_{n,v})(o)$; $P_n^rL^r(Q^r)=\frac{1}{V}\sum_{v=1}^V P_{n,v}^1 L(Q^r\circ{\bf Q}_{n,v})$; $Q_n^r=\sum_{(s,j)\in {\cal J}_n(C_n)}\beta_n(s,j)\phi_{s,j}$, where ${\cal J}_n(C_n)$ provide the indices of the basis functions with non-zero coefficients and $\beta_n$ denotes the corresponding coefficients. Here we emphasize that ${\cal J}_n(C_n)$ is implied by the $L_1$-norm bound $C_n$ in definition of $Q_n^r$. We have that $Q_n^r$ solves the score equations $P_n^r S_h(Q_n^r)=0$ for all bounded $h$ with $r(h,Q_n^r)=0$, and $S_h(Q_n^r)$ defined by (\ref{scorefn}).
 
Given a realization of ${\bf Q}_n$ (i.e., treating it as fixed),  consider the target parameter  ${\Psi}^r:{\cal Q}^r\rightarrow\openr$ defined by $\Psi^r(Q^r)=\frac{1}{V}\sum_{v=1}^V \Psi(Q^r\circ{\bf Q}_{n,v})$. Assume its  canonical gradient $D^r(Q^r,G^r)(v,o^r)=D^*(Q^r\circ{\bf Q}_{n,v},G^r_v)(o)$ at $P^r\in {\cal M}^r$, and exact second order remainder ${R}_2^r(Q^r,G^r,Q^r_{0,n},G_{0,n}^r)=\Psi^r(Q^r)-\Psi^r(Q^r_{0,n})+P_0 D^r(Q^r,G^r)$ given by $\frac{1}{V}\sum_v R_{20}(Q^r\circ{\bf Q}_{n,v},G^r_v,Q^r_{0,n}\circ{\bf Q}_{n,v},G^r_{0,n,v}$.
Let $D^r_n(Q_n^r,G_{0,n}^r)(v,o^r)=D^*_n(Q^r\circ{\bf Q}_{n,v},G^r_{0,n,v})(o)$ be an approximation of $D^r(Q_n^r,G_{0,n}^r)$ that is contained in ${\cal S}(Q_n^r)=\{S_h(Q_n^r):h\}$, without the restriction $r(h,Q_n^r)=0$.
Let $h^*(Q_n^r,G_{0,n}^r)$ be the corresponding index so that $D^r_n(Q_n^r,G_{0,n}^r)=S_{h^*(Q_n^r,G_{0,n}^r)}(Q_n^r)$.

{\bf Assumptions:}
Assume $\pl \tilde{h}^*(Q_n^r,G_{0,n}^r)\pl_{\infty}=O_P(1)$,
and the global undersmoothing criterion
\begin{equation}
C_n \min_{(s,j)\in {\cal J}_n(C_n)}\pl P_n^r \frac{d}{dQ_n^r}L^r(Q_n^r)(\phi_{s,j})\pl =o(n^{-1/2}).
\label{assumptionatarget}
\end{equation} 

{\bf Conclusion:}
 Then, 
 \[
 P_n^r D^r_n(Q_n^r,G_{0,n}^r)=o_P(n^{-1.2}).\]
 
 We can also replace (\ref{assumptionatarget}) by \begin{equation}\label{suffassumptionatarget}
\min_{(s,j)\in {\cal J}_n(C_n)}\pl P_0^r\left\{\frac{d}{dQ_n^r}L^r(Q_n^r)(\phi_{s,j})-\frac{d}{dQ^r_{0,n}}L^r(Q^r_{0,n})(\phi_{s,j})\right\}\pl =o_P(n^{-1/2}),\end{equation} and, for the choice $(s^*,j^*)$ that minimizes the latter, we have $P_0^r \{\frac{d}{dQ_n^r}L^r(Q_n^r)(\phi_{s^*,j^*}) \}^2\rightarrow_p 0$.

If also
\begin{equation}\label{a3} P_0^r\{D^r(Q_n^r,G_{0,n}^r)-D^r_n(Q_n^r,G_{0,n}^r)\}=o_P(n^{-1/2}),
\end{equation} then we have
\[
P_n^r D^r(Q_n^r,G_{0,n}^r)=o_P(n^{-1/2}).\]
 \end{theorem}
 \textbf{Remark regarding setting cut-off for undersmoothing condition}
 Our proof shows that $P_n^r D^r_n(Q_n^r,G_{0,n}^r)\approx \max_{(s,j)\in {\cal J}_n(C_n)}\mid h^*(s,j)\mid C_n\min_{(s,j)\in {\cal J}_n(C_n)}\pl P_n^r \frac{d}{dQ_n^r}L^r(Q_n^r)(\phi_{s,j})\pl $, where $h^*$ is so that $D_n^r(Q_n^r,G_{0,n}^r)=S_{h^*}(Q_n^r)$.
 A sensible bound for $P_n^rD^r_n(Q_n^r,G_{0,n}^r)$ is $\sigma_n/(n^{1/2}\log n)$.
 Thus, we would want to select $C_n>C_{n,cv}$ so that
 \[
 \max_{(s,j)\in {\cal J}_n(C_n)}\mid h^*(s,j)\mid C_n\min_{(s,j)\in {\cal J}_n(C_n)}\pl P_n^r \frac{d}{dQ_n^r}L^r(Q_n^r)(\phi_{s,j})\pl  \approx \sigma_n/(n^{1/2}\log n).\]
 If one knows the canonical gradient $D^r(Q_n^r,G_n^r)$ for a given estimator $G_n^r$, then one can determine the corresponding $h^*(Q_n^r,G_n^r)$  so that $D^r(Q_n^r,G_n^r)=S_{h^*(Q_n^r,G_n^r)}(Q_n^r)$, and use the $\max$-norm of $h^*(Q_n^r,G_n^r)$. Therefore, a recommended concrete criterion is given by
 \[ 
 \max_{(s,j)\in {\cal J}_n(C_n)}\mid h^*(s,j)\mid C_n\min_{(s,j)\in {\cal J}_n(C_n)}\pl P_n^r \frac{d}{dQ_n^r}L^r(Q_n^r)(\phi_{s,j})\pl  \approx \sigma_n/(n^{1/2}\log n).\] 
 However, in this case we are aiming to make the undersmoothing criterion tailored for the particular target parameter $\Psi$. Of course, one might as well select $C_n$ so that $P_n^r D^r_n(Q_n^r,G_n^r)\approx \sigma_n/(n^{-1/2}\log n)$, since one already used $G_n^r$ and even aimed to estimate the max-norm of $h^*$. However, we also see that if we set $\frac{\sigma_n}{\pl h^*\pl_{\infty}}$ to some constant $K$ (e.g, $K=1$), then we obtain a global undersmoothing criterion
 \[ \min_{(s,j)\in {\cal J}_n(C_n)}\pl P_n^r \frac{d}{dQ_n^r}L^r(Q_n^r)(\phi_{s,j})\pl  \approx K C_n^{-1}/(n^{1/2}\log n).\] 
 In many censored or causal inference problem, both this max-norm of $h^*$, which is aligned with the sup-norm  of $D^r(Q_n^r,G_{0,n}^r)$, and the standard error $\sigma_n$ of $D^r(Q_n^r,G_{0,n}^r)$ are driven by a positivity assumption and increase as the support for the target parameter decreases. This suggest that it might be quite reasonable to assume that across a large class of target parameters $K$ is uniformly bounded away from 0, so that the above global undersmoothing condition will work well across a large family of target parameters. 

{\bf Proof Theorem \ref{thscoreequationtarget}:}
 We already showed above that $P_n^rD^r_n(Q_n^r,G_{0,n}^r)=o_P(n^{-1/2})$ by the undersmoothing condition. Now we note that 
\begin{eqnarray*}
P_n^r D^r(Q_n^r,G_{0,n}^r)&=&P_n^r\{D^r(Q_n^r,G_{0,n}^r)-D^r_n(Q_n^r,G_{0,n}^r)\}+o_P(n^{-1/2})\\
&=&P_0^r\{{D}^r(Q_n^r,G_{0,n}^r)-{D}^r_n(Q_n^r,G_{0,n}^r)\}+o_P(n^{-1/2}),
\end{eqnarray*}
if, for each $v$, conditional on the training sample (and thus, fixed ${\bf Q}_{n,v}$)
\[
\{{D}^r(Q^r,G_{0,n}^r), D^r_n(Q^r,G_{0,n}^r): Q^r\in {\cal Q}^r(C^u)\}\mbox{ is a $P_0^r$-Donsker class},
\]
and
 \[
 P_0^r\{{D}^r(Q_n^r,G_{0,n}^r)-D^r_n(Q_n^r,G_{0,n}^r)\}^2\rightarrow_p 0.\]
  The Donsker assumption holds since, by (\ref{r1}) (and remark above showing it also applies to $D^r_n$), it consists of $d^r$-variate cadlag functions with universal bound on sectional variation norm. The consistency condition 
  is implied by (\ref{a3}).
 $\Box$
 
 \subsection{Understanding assumption (\ref{a3}).}
Here we  discuss the key condition (\ref{a3}), beyond the undersmoothing condition (\ref{assumptionatarget}). For this purpose, we reparametrize the paths $Q^{r,h}_{n,\epsilon}$ as follows:
 \[
 Q^{r,l(h,Q_n^r)}_{n,\epsilon}(x)=Q_n^r(x)+\epsilon l(h,Q_n^r)(x),
,\]
 where 
 \[
 l(h,Q_n^r)(x)=h(0)Q_n^r(0)+\sum_{s\subset\{1,\ldots,J\}}
 \int_{(0_s,x_s]}h(s,u_s)dQ^r_{n,s}(u_s) .\]  
 Therefore, we could also define the class of paths $\{Q^{r,h}_{n,\epsilon}:\pl h\pl_{\infty}<\infty\}$ as $\{Q^{r,l}_{n,\epsilon}: l\in {\cal F}(Q_n^r)\}$, where the index set is given by ${\cal F}(Q_n^r)=\{l(h,Q_n^r): \pl h\pl_{\infty}<\infty\}$. 
 The set ${\cal F}(Q_n^r)$ is restricted since it consists of the linear span of $\{\phi_{s,j}: (s,j)\in {\cal J}_n(C_n)\}$, that is, the linear span of all basis functions $\phi_{s,u_{s,j}}$ with non-zero coefficient $\beta_n(s,u_{s,j})$ in the fit $Q_n^r=\sum_{(s,j)}\beta_n(s,j)\phi_{s,j}$.
 The scores $S_h(Q_n^r)$ are linear in $l(h,Q_n^r)$ and the set of scores $\{S_h(Q_n^r): \pl h\pl_{\infty}<\infty\}$ can be parametrized accordingly as $\{S_l(Q_n^r): l\in {\cal F}(Q_n^r)\}$.
 We will typically have that the canonical gradient ${D}^r(Q_n^r,G_{0,n}^r)=\left . \frac{d}{d\epsilon}L(Q^{r,l_{0,n}}_{n,\epsilon})\right |_{\epsilon =0}$ for a choice $l_{0,n}=l_0(Q_n^r,G_{0,n}^r)$, generally not an element of ${\cal F}(Q_n^r)$. Let ${\cal F}^{+}(Q_n^r)$ be this richer set so that $l_{0,n}\in {\cal F}^+(Q_n^r)$ and $\{S_l(Q_n^r): l\in {\cal F}^{+}(Q_n^r)\}$ is an augmented set of scores satisfying $P_0^r S_l(Q_{0,n}^r)=0$ for all $l\in {\cal F}^+(Q_{0,n}^r)$.
 So let's make this assumption. 
 Then, we can write ${D}^r(Q_n^r,G_{0,n}^r)=S_{l_{0,n}}(Q_n^r)$. We can define ${D}^r_n(Q_n^r,G_{0,n}^r)$ as the projection of ${D}^r(Q_n^r,G_{0,n}^r)$ onto the finite dimensional linear span $\{S_l(Q_n^r): l\in {\cal F}(Q_n^r)\}$: this actually equals the linear span of $\frac{d}{dQ_n^r}L^r(Q_n^r)(\phi_{s,j})$ across $(s,j)$ with  $\beta_n^r(s,j)\not =0$. Thus, ${D}^r_n(Q_n^r,G_{0,n}^r)=S_{l_n}(Q_n^r)$ for a $l_n\in {\cal F}(Q_n^r)$. Somewhat conservatively, we could define $l_n$ as the projection  of $l_{0,n}$ onto the finite dimensional space $\{\sum_{(s,j)\in {\cal J}_n(C_n)}\alpha(s,j)\phi_{s,j}: \alpha\}$ spanned by the basis functions $\phi_{s,j}$ with a non-zero coefficient $\beta_n(s,j)$ in $Q_n^r$. Let $\pl l_{0,n}-l_{n}\pl_0$ be the chosen Hilbert space norm so that $l_n=\arg\min_{l\in {\cal F}(Q_n^r)}\pl l_{0,n}-l\pl_0$. Since the set of basis functions $\{\phi_{s,j}: (s,j)\in {\cal J}(C_n)\}$ is rich enough (even when we select $C_n=C_{n,cv}$) to approximate $Q_{0,n}^r$ w.r.t. $d_0^{1/2}(f,Q_{0,n}^r)$ at a rate $n^{-1/3}(\log n)^{d^r/2}$, one generally expects that $\pl l_n-l_{0,n}\pl_0$ will also be $O_P(n^{-1/3}(\log n)^{d^r/2})$. However, as argued in main section, if (e.g.) due to $G_{0,n}^r$ being more complex, $l_{0,n}$ is spanned by basis functions that are not needed for $Q_{0,n}^r$, then it might require $C_n>C_{n,cv}$ to obtain this rate of convergence.
 Finally, we note that
 \begin{eqnarray*}
 P_0^r\{{D}^r_n(Q_n^r,G_{0,n}^r)-{D}^r(Q_n^r,G_{0,n}^r)\}&=&
 P_0^r S_{l_n-l_{0,n}}(Q_n^r)\\
 &=& P_0^r\{ S_{l_n-l_{0,n}}(Q_n^r)-S_{l_n-l_{0,n}}(Q_{0,n}^r)\},
 \end{eqnarray*}
 since $P_0^r S_l(Q_{0,n}^r)=0$ for all $l\in {\cal F}^+(Q_{0,n}^r)$. This now proves that the left-hand difference is indeed a second order term that can typically be bounded by $d_0(Q_n^r,Q_{0,n}^r)^{1/2} \pl l_n-l_{0,n}\pl_0$, so that it will be $O_P(n^{-1/3}(\log n)^{d^r/2})\pl l_n-l_{0,n}\pl_0$. Therefore, a sufficient assumption for (\ref{a3}) is that $\pl l_n-l_{0,n}\pl_0=O_P(n^{-1/6-\delta})$ for some $\delta>0$.

\section{Analysis of the Targeted HAL super-learner}\label{AppendixD}
 \subsection{Rate of convergence of T-M-HAL-SL}
Analogue to \citep{vanderLaan&Cai20}, we obtain the same rate of convergence for $d_0(Q_n^{r,*},Q_{0,n}^r)$ as for $d_0(Q_n^r,Q_{0,n}^r)$.
Firstly, we can copy the proof of Lemma \ref{theoremrateQn} by defining ${\cal Q}^r_n\subset{\cal Q}^r$ as the subset of functions $Q^r\in {\cal Q}^r$ for which $\pl P_n^rD^r(Q^r,G_n^r)\pl <r_n$, and $Q_{0,n}^{r,*}$ as the corresponding oracle ensemble. This then proofs $d_0(Q_n^{r,*},Q_{0,n}^{r,*})=O_P(n^{-2/3}(\log n)^{d^r})$. It then remains to show that $d_0(Q_{0,n}^{r,*},Q_{0,n}^r)=O_P(n^{-2/3}(\log n)^{d^r})$. We then construct a local least favorable submodel through $Q_{0,n}^r$ and define a corresponding TMLE update which maps $Q_{0,n}^r$ into a targeted version $\tilde{Q}^r_{0,n}$ that is an element of ${\cal Q}^r_n$. However, in this case this LFM is centered at the true $Q_{0,n}^r$ so that the MLE $\epsilon_n=O_P(n^{-1/2})$, thereby showing that $d_0(Q_{0,n}^r,\tilde{Q}^r_{0,n})=O_P(n^{-1})$. This then shows that $d_0(Q_{0,n}^{r,*},Q_{0,n}^r)\leq d_0(\tilde{Q}^r_{0,n},Q_{0,n}^r)=O_P(n^{-1})$. Therefore, one can conclude that $Q_{0,n}^r$ and $Q_{0,n}^{r,*}$ only differ by a negligible amount so that we indeed have $d_0(Q_n^r,Q_{0,n}^r)=O_P(n^{-2/3}(\log n)^{d^r})$.
This results in the following analogue of Lemma \ref{lemmarateQnr} for this T-HAL-MLE ensemble selector $Q_n^{r,*}$.

\begin{lemma}\label{lemmaratethalmle}
\ \nl
 {\bf Definitions:} Let ${\cal Q}^{r,LFM}(Q_{0,n}^r)\equiv \{Q^r_{0,n,\epsilon,G_n^r}:\epsilon\}\subset {\cal Q}^r(C)$, with $\epsilon\in (-\delta,\delta)$ for some arbitrary small $\delta>0$, be a local least favorable submodel through $Q_{0,n}^r$ at $\epsilon =0$ so that $\frac{d}{d\epsilon }L^r(Q^r_{0,n,\epsilon,G_n^r})=D^r(Q^r_{0,n},G_n^r)$ at $ \epsilon =0$. Note that this parametric model ${\cal Q}^{r,LFM}(Q_{0,n}^r)$ with parameter $\epsilon$ is correctly specified and the true parameter $\epsilon_0=0$. Let $\epsilon_n=\arg\min_{\epsilon}P_n^r L^r(Q^r_{0,n,\epsilon,G_n^r})$ be the MLE of $\epsilon_0$, where $\epsilon$ may vary over larger set than $(-\delta,\delta)$.
 \newline
{\bf T-HAL-MLE:}
Let $C_n(Q^r)\equiv \pl P_n^r D^r(Q^r,G_n^r)\pl$, and consider the T-HAL MLE \[
Q_n^{r,*}=\arg\min_{\pl Q^r\pl_v^*<C,C_n(Q^r)\leq r_n}P_n^r L^r(Q^r).\]

{\bf Assumptions:} 
Assume (\ref{r1}); 
 regularity conditions on the least favorable submodel ${\cal Q}^{r,LFM}(Q_{0,n}^r)$, so that the MLE $\epsilon_n=O_P(n^{-1/2})$, thereby $d_0(Q^r_{0,n,\epsilon_n,G_n^r},Q^r_{0,n})=O_P(n^{-1})$, and  $\pl P_n^r D^r(Q^r_{0,n,\epsilon_n,G_n^r},G_n^r)\pl \leq r_n$ with probability tending to 1.  

{\bf Conclusion:}
We have
\[
d_0^{\bar{V}}(Q_n^*,Q_{0,n})=d_0^r(Q_n^{r,*},Q^r_{0,n})=O_P(n^{-2/3}(\log n)^{d^r}).\] 
\end{lemma}
Since $d_0^{\bar{V}}(Q_{0,n},Q_0)$ is not affected by the targeting,  Theorem \ref{theoremrateQn} applies also to $Q_n^*$.

\section{Treatment Specific Mean Example: Solve Score Equations by Undersmoothing} \label{sec:TSM_undersmoothing}

Since, $Q_n^r$ is an MLE it solves a large class of score equations $P_n^rS_h(Q_n^r)=0$ defined by  $S_h(Q_n^r)=\left . \frac{d}{d\epsilon} L^r\left (\sum_{s,j}(1+\epsilon h(s,j))\beta_n^r(s,j))\phi_{s,j}\right )\right |_{\epsilon =0}$ and $h$ any bounded function satisfying that $r(h,Q_n^r)=0$. The constraint is defined as
$r(h,Q_n^r)\equiv \sum_{s,j}h(s,j)\mid \beta_n^r(s,j) \mid =0$. 

Note that $S_h(Q_n^r)(O^r)=A\sum_{s,j}h(s,j)\beta_n^r(s,j)\phi_{s,j}(W^r)(Y-Q_n^r(W^r))$. 
Thus, this class of scores $\{S_h(Q_n^r):r(h,Q_n^r)=0\}$ across all bounded $h$ with $r(h,Q_n^r)=0$ equals the dimension of the number of non-zero coefficients $\beta_n^r(s,j)$ in its representation $Q_n^r=Q_{\beta_n^r}=\sum_{s,j}\beta_n^r(s,j)\phi_{s,j}$, minus 1 due to the constraint $r(h,Q_n^r)=0$. Therefore, as the $L_1$-norm  $C_n=\pl \beta_n^r\pl_1$ increases, the number of non-zero coefficients grows so that this linear span of these score equations grows accordingly. 
This is the intuitive argument why choosing $C_n$ large enough should give us $P_n^r D^r(Q_n^r,G_{0,n}^r)=o_P(n^{-1/2})$: 
\[
\frac{1}{n}\sum_{i=1}^n \frac{A_i}{G_{0,n}^r(W^r_i)}(Y_i-Q^r_n(W_i^r))=o_P(n^{-1/2}).\]

Formally, we apply Theorem \ref{thscoreequationtarget}.
Let ${\cal S}=\{S_h(Q_n^r): h\}$ the class of scores not enforcing the constraint $r(h,Q_n^r)=0$.
We first need to define an approximation $D_n^r(Q_n^r,G_{0,n}^r)\in {\cal S}$. 
By selecting $C_n$ large enough we will have that we can find an $h$ that makes $A\sum_{s,j}h(s,j)\beta_n(s,j)\phi_{s,j}(W^r)$ approximate $A/G_{0,n}^r(W^r)$. 
Let $h_n^*=h^*(Q_n^r,G_{0,n}^r)\equiv \arg\min_{h} P_0^r(\sum_{s,j}\beta_n^r(s,j)h(s,j)\phi_{s,j}(W^r)-1/G_{0,n}^r(W^r))^2$ be defined as this $L^2(P_0^r)$-projection of the desired $1/G_{0,n}^r$ onto this linear span, where the $\arg\min$ includes any bounded $h$ (not restricting to $r(h,Q_n^r)=0$). We can then define the approximation $D^r_n(Q_n^r,G_{0,n}^r)\equiv 
A\sum_{s,j}h_n^*(s,j)\beta_n^r(s,j)(Y-Q_n^r(W^r))\in {\cal S}$ of $D^r(Q_n^r,G_{0,n}^r)$.

Since $E(A\mid W)>\delta>0$, by iterative conditional expectation, it follows that $G_{0,n}^r=E(A\mid W^r)>\delta>0$ for some $\delta>0$. Then, it follows that the sup-norm of $h_n^*$ is $O_P(1)$, which verified the first condition of Theorem \ref{thscoreequationtarget}.  Consider now the global undersmoothing criterion (\ref{assumptionatarget}) of Theorem \ref{thscoreequationtarget}, and note that it is given by:
 select $C_n$ large enough so that the fit $Q_n^r$ includes  a sparsely supported basis function with $\min_{s,j}\mid P_n \phi_{s,j}\mid$ small enough in the sense that
\begin{equation}\label{undersmoothexample}
\min_{(s,j),\beta_n^r(s,j)\not =0}\mid \frac{1}{n}\sum_{i=1}^n \phi_{s,j}(W^r_i)(Y_i-Q_n^r(W^r_i))\mid =o_P(n^{-1/2}).\end{equation}
Application of Theorem \ref{thscoreequationtarget} proves now that  \[
 P_n^r D^r_n(Q_n^r,G_{0,n}^r)=o_P(n^{-1.2}).\] 
 We assume that $G_{0,n}^r$ is cadlag and has a uniformly bounded sectional variation norm. Since $Q_n^r$ is by definition cadlag with finite sectional variation, this shows that ${\cal D}^r=\{D^r(Q^r,G_{0,n}^r):Q^r\in {\cal Q}^r\}$ is a  class of $d^r$-dimensional real valued cadlag functions with a uniformly bounded sectional variation norm. This verifies the Donsker class conditions in (\ref{r1}).
Condition  (\ref{a3})  of Theorem \ref{thscoreequationtarget} states
 \[
 E_{P_0^r}A \left\{Q^r_{\beta_n^rh_n^*}-1/G_{0,n}^r\right\} (Y-Q_n^r(W^r))=o_P(n^{-1/2}).\]
 So this requires our approximation $Q^r_{\beta_n^rh_n^*}$ of $1/G_{0,n}^r$ to converge fast enough.
By Cauchy-Schwarz inequality, the left-hand side can be bounded by $\pl Q_n^r-Q_{0,n}^r\pl_{P_0}^r\pl Q^r_{\beta_n^r h_n^*}-1/G_{0,n}^r\pl_{P_0^r}$. Given our rate $\pl Q_n^r-Q_{0,n}^r\pl_{P_0}=n^{-1/3}(\log n)^{d^r/2}$ it follows that it suffices that
\begin{equation}\label{examplea3}
\inf_{\beta}\pl \sum_{s,j,\beta_n^r(s,j)\not =0}  \beta(s,j)\phi_{s,j}-\frac{1}{G_{0,n}^r}\pl_{P_0^r}=O_P(n^{-1/6-\delta}),\end{equation}
for some $\delta>0$.
We will assume this to hold.
Since we know that the left-hand side with $1/G_{0,n}^r$ replaced by $Q_{0,n}^r$ would be $O_P(n^{-1/3}(\log n)^{d^r/2})$, even without undersmoothing (i.e., setting $C_n=C_{n,cv})$ this might already hold. On the other hand, if the true $Q_{0,n}^r$ is a relatively simple function spanned by a subset of all possible spline basis functions, while approximating $1/G_{0,n}^r$ requires these basis functions, then undersmoothing will be needed. 
This verifies all conditions of Theorem \ref{thscoreequationtarget} and thus proves the following result. 
\begin{lemma}\label{lemma1example}
Assume  (\ref{examplea3}); and undersmoothing condition (\ref{undersmoothexample}).

Then, $P_n^r D^r(Q_n^r,G_{0,n}^r)=r_1(n)$ with $r_1(n)=o_P(n^{-1/2})$.
\end{lemma}

\section{Treatment Specific Mean Example: Difference Between Targets} \label{sec:TSM_diff}

The following lemma establishes that $\Psi^r(Q_{0,n}^r)-\Psi(Q_0)$ behaves as $d_0(Q_n,Q_0)$ and is thus second order. 
\begin{lemma}\label{lemmaatesecondorder}
{\bf Definitions:}
Recall  $Q_{0,n,v}(W)=Q_{0,n}^r\circ{\bf Q}_{n,v}(W)$.
Define
\begin{eqnarray*}
\tilde{G}_{0,n,v}^{*}(x,y)&\equiv& E_0(A\mid Q_{0,n,v}(W) =x,Q_0(W)=y)\\ 
\tilde{G}_{0,n,v}(x)&\equiv& E_0(A\mid Q_{0,n}(v,W)=x).
\end{eqnarray*}
Due to $G_0>\delta>0$, these two functions are also bounded away from this $\delta$. 
These two functions (where chosen to) satisfy $P_0D^*(Q_{0,n,v},G_0)=P_0D^*(Q_{0,n,v},\tilde{G}_{0,n,v}^*(Q_{0,n,v},Q_0))$, and
$P_0D^*(Q_{0,n,v},\tilde{G}_{0,n,v}(Q_{0,n,v}))=0$. \nl
{\bf Assumptions:}
Assume
\begin{itemize}
\item
$\hat{Q}_1$ is  an HAL-MLE so that, by Theorem \ref{theoremrateQn}
we have $d_0(Q_{0,n},Q_0)=O_P(n^{-2/3}(\log n)^d)$, and, thus also, for each $v$, $d_0(Q_{n,v},Q_0)=O_P(n^{-2/3}(\log n)^d)$. 
\item $\lim\sup_n \sup_{x,y}\mid \frac{d}{dy}\tilde{G}_{0,n,v}^{*}(x,y)\mid <\infty$, where the supremum over $(x,y)\in \openr^2$ is over a support of $(Q_{0,n}(v,W),Q_0(W))$. 
\end{itemize}
Then, we have that $\mid \Psi(Q_{0,n,v})-\Psi(Q_0)\mid $ is bounded by  $\delta^{-2}\pl( \tilde{G}_{0,n,v}^*-\tilde{G}_{0,n,v})(Q_{0,n,v},Q_0)\pl_{P_0}\pl Q_{0,n,v}-Q_0\pl_{P_0}$.
In addition, we have, by Lemma \ref{lemmahelp1} below that $\pl( \tilde{G}_{0,n,v}^*-\tilde{G}_{0,n,v})(Q_{0,n,v},Q_0)\pl_{P_0}=O(\pl Q_{0,n,v}-Q_0\pl_{P_0})$. This proves
 \[
\left\{ \frac{1}{V}\sum_v \Psi(Q^r_{0,n}\circ{\bf Q}_{n,v})-\Psi(Q_0)\right\}=O(d_0(Q_{0,n},Q_0)).\]
Thus, by first bullet point assumption, we have that this is $O_P(n^{-2/3}(\log n)^d)$.
\end{lemma}
The differentiability condition on $G_{0,n,v}^{*}$ is not a bad  assumption since it only concerns the dependence of the conditional expectation of $A$, given $(Q_{0,n,v}(W),Q_0(W))$, on the fixed random variable $Q_0(W)$.

In this example, $d_0(Q_{n,v},Q_0)=O_P(n^{-2/3}(\log n)^d)$ can also be shown directly, instead of as an application of Lemma \ref{lemmarateQnr}.
 We have $Q_{0,n,v}(W)=E_0(Y\mid {\bf Q}_{n,\bar{V}}(W)={\bf Q}_{n,v}(W))$.
 Therefore, $Q_{0,n,v}-Q_0$  requires analyzing $E_0(Y\mid {\bf Q}_{n,v}(W))-Q_0(W)$.
We have that $E_0(Y\mid {\bf Q}_{n,v}(W),A=1)$ is the projection of $Q_0(W)=E_0(Y\mid W,A=1)$ onto the set of functions of ${\bf Q}_{n,v}(W)$ (in the Hilbert space $L^2(P_{0\mid A=1})$). One such candidate function for the projection is given by ${\bf Q}_{n,v,j=1}(W)$, showing that the $L^2(P_{0\mid A=1})$-norm of $Q_0(W)-E_0(Y\mid {\bf Q}_{n,v}(W),A=1)$ is smaller than the $L^2(P_0)$-norm of $Q_0-{\bf Q}_{n,v,1}$, but the latter is $O_P(n^{-1/3}(\log n)^{d/2})$, by assumption.

\ \nl
{\bf Proof of Lemma \ref{lemmaatesecondorder}:}
We have
\begin{equation}\label{z1}
\Psi(Q_{0,n,v})-\Psi(Q_0)=-P_0 D^*(Q_{0,n,v},G_0), \end{equation}
since the second order remainder $R_{20}(Q,G_0,Q_0,G_0)=0$ (due to its double robust structure).
Note now that $P_0D^*(Q_{0,n,v},G_0)=P_0D^*(Q_{0,n,v},\tilde{G}_{0,n,v}^*(Q_{0,n,v},Q_0))$.
This follows since
\[
\begin{array}{l}
E_0  A/G_0(Y-Q_{0,n,v}(W))=E_0  A/G_0(Q_0-Q_{0,n,v})(W) \\
= E_0 (Q_0-Q_{0,n,v})(W)=E_0  A/\tilde{G}_{0,n,v}^{*}(Q_{0,n,v}(W),Q_0(W)) (Q_0-Q_{0,n,v})(W).
\end{array}
\]
So we can replace the right-hand side in (\ref{z1}) by
$-P_0 D^*(Q_{0,n,v},\tilde{G}_{0,n,v}^{*}(Q_{0,n,v},Q_0) )$:
\begin{equation}\label{z2}
\Psi(Q_{0,n,v })-\Psi(Q_0)=-P_0 D^*(Q_{0,n,v},\tilde{G}_{0,n,v}^{*}(Q_{0,n,v},Q_0)).\end{equation}
We now note that $P_0 D^*(Q_{0,n,v},\tilde{G}_{0,n,v}(Q_{0,n,v}) )=0$: due to $E_0(Y\mid A=1,Q_{0,n,v}(W))=Q_{0,n,v}(W)$ it follows that
\[
E_0 \frac{A}{E_0(A\mid Q_{0,n,v}(W))}(Y-Q_{0,n,v}(W))=0.\]
Thus, we have
\[
\begin{array}{l}
\Psi(Q_{0,n,v})-\Psi(Q_0)=P_0\{D^*(Q_{0,n,v},\tilde{G}_{0,n,v}(Q_{0,n,v}))-D^*(Q_{0,n,v},\tilde{G}_{0,n,v}^*(Q_{0,n,v},Q_0))\}\\
=P_0\frac{\tilde{G}_{0,n,v}^*-\tilde{G}_{0,n,v}}{\tilde{G}_{0,n,v}^*\tilde{G}_{0,n,v}}(Q_{0,n,v},Q_0)(Q_0-Q_{0,n,v}).
\end{array}
\]
By Cauchy-Schwarz inequality, we can bound the latter term by $\pl( \tilde{G}_{0,n,v}^*-\tilde{G}_{0,n,v})(Q_{0,n,v},Q_0)\pl_{P_0}\pl Q_{0,n,v}-Q_0\pl_{P_0}$. This shows that it remains to show
$\pl( \tilde{G}_{0,n,v}^*-\tilde{G}_{0,n,v})(Q_{0,n,v},Q_0)\pl_{P_0}=O_P(n^{-1/6-\delta})$ for some $\delta>0$. 
For that we apply Lemma \ref{lemmahelp1}. This lemma shows that we can bound $E_0(A\mid Q_{0,n,v}(W),Q_0(W))-E_0(A\mid Q_{0,n,v}(W) )$ by the $L^2$-norm of $Q_{0,n,v}-Q_0$.  This completes the proof. $\Box$


 \begin{lemma}\label{lemmahelp1}
 For notational convenience, let $X_n(w)=Q_{0,n,v}(w)$, $X(w)=Q_0(w)$, $\tilde{G}_{0,n,v}^*(X_n(w),X(w))=E_0(A\mid X_n(W)=X_n(w),X(W)=X(w))$, and 
$\tilde{G}_{0,n,v}(X_n(w))=E_0(A\mid X_n(W)=X_n(w))$.  
Let $X_n=X_n(W)$ and $X=X(W)$, and let $L^2(P_{X_n,X})$ be the correponding Hilbert space of functions of $(X_n,X)$ with covariance inner product. 
Assume that $\lim\sup_n \sup_{x,y}\mid \frac{d}{dy}\tilde{G}_{0,n,v}^{*}(x,y)\mid <\infty$, where the supremum over $(x,y)$ is over a support of $(X_n(W),X(W))$. 
We note that $\tilde{G}_{0,n,v}$ is the projection of $\tilde{G}_{0,n,v}^*$ onto the subspace $L^2(P_{X_n})$ of functions of $X_n$ only in the Hilbert space $L^2(P_{X_n,X})$.

We have
\[
 \pl \tilde{G}_{0,n,v}^{*}-\tilde{G}_{0,n,v}\pl_{P_{X_n,X}}=O\left(\pl X_n-X\pl_{P_0}\right).\]
 \end{lemma} 
{\bf Proof of Lemma \ref{lemmahelp1}:}
We have that $\tilde{G}_{0,n,v}$ is the projection of $\tilde{G}_{0,n,v}^{*}$ (a function of $X_n,X$) on the subspace $L^2(P_{X_n})$ of all functions that only depend on $X_n$, a subspace of $L^2(P_{X_n,X})$, endowed with the usual covariance as inner product.
Thus, \[
\begin{array}{l}
\pl \tilde{G}_{0,n,v}^{*}-\tilde{G}_{0,n,v}\pl_{P_{X_n,X}}^2
=\inf_{f\in L^2(X_n)}\pl \tilde{G}_{0,n,v}^{*}-f\pl^2_0\\
=\inf_{f \in L^2(X_n)}\int \{\tilde{G}_{0,n,v}^{*}(X_n(w),X(w))-f(X_n(w))\}^2 dP_0(w)\\
=\inf_{f\in L^2(X_n)} \int \left\{\tilde{G}_{0,n,v}^{*}(X_n(w),X(w))-\tilde{G}_{0,n,v}^{*}(X_n(w),X_n(w))\right  . \\
 \left . + \tilde{G}_{0,n,v}^*(X_n(w),X_n(w))-f(X_n(w))\right \}^2 dP_0(w) \\
=\int \{\tilde{G}_{0,n,v}^{*}(X_n(w),X(w))-\tilde{G}_{0,n,v}^{*}(X_n(w),X_n(w))\}^2dP_0(w)\\
+\inf_{f\in L^2(X_n)}\left\{ \int\{\tilde{G}_{0,n,v}^*(X_n(w),X_n(w))-f(X_n(w))\}^2dP_0(w)\right . \\
\left . +2\int (\tilde{G}_{0,n,v}^{*}(X_n,X)-\tilde{G}_{0,n,v}^{*}(X_n,X_n))(\tilde{G}_{0,n,v}^{*}(X_n,X_n)-f(X_n)) dP_{X_n}\right \}
.\end{array}
\]
The  latter infimum over all functions $f$ of $X_n(w)$ is attained at $f=\tilde{G}_{0,n,v}^{*}(X_n(w),X_n(w))$, so we obtain
\[
\pl \tilde{G}_{0,n,v}^{*}-\tilde{G}_{0,n,v}\pl_{P_{X_n,X}}^2=\int \{\tilde{G}_{0,n,v}^{*}(X_n(w),X(w))-\tilde{G}_{0,n,v}^{*}(X_n(w),X_n(w))\}^2dP_0(w).\]
By the assumed differentiability of $\tilde{G}_{0,n,v}^{*}$  in its second coordinate we have
\[
\begin{array}{l}
\tilde{G}_{0,n,v}^{*}(X_n(w),X(w))=\tilde{G}_{0,n,v}^{*}(X_n(w),X_n(w))\\
+\left . \frac{d}{dy}\tilde{G}_{0,n,v}^{*}(X_n(w),y)\right |_{y=\xi(X_n(w),X(w))}(X(w)-X_n(w)),\end{array}
\]
for an intermediate point $\xi(X_n(w),X(w))$ in between $X_n(w)$ and $X(w)$.
By the assumed uniform bound on the derivative we have $\mid \tilde{G}_{0,n,v}^{*}(X_n(w),X(w))-\tilde{G}_{0,n,v}^{*}(X_n(w),X_n(w))\mid <C\mid X_n(w)-X(w)\mid $ for some $C<\infty$, so that we have
\[
\pl \tilde{G}_{0,n,v}^{*}-\tilde{G}_{0,n,v}\pl_{P_{X_n,X}}^2\leq C \int (X_n(w)-X(w))^2dP_0(w)=C\pl Q_{0,n,v}-Q_0\pl_0^2.\]
So this proves that $\pl G_{0,n,v}^*-G_{0,n,v}\pl_{P_0}=O_P(n^{-1/3}(\log n)^{d/2})$.
This proves the lemma. $\Box$

 \section{Relation between undersmoothing criterion (\ref{suffassumptionatarget}) and bound $C_n$ on sectional variation norm}\label{AppendixE}
Consider assumption (\ref{suffassumptionatarget}). In our treatment specific mean example, this states
\begin{equation}\label{suffex}
\min_{s,j,\beta_n^r(s,j)\not =0}P_0^r \phi_{s,j}(Q_n^r-Q_{0,n}^r)=o_P(n^{-1/2}).\end{equation}
The next theorem denotes the left-hand side with $R_n$ and shows that $R_n=o_P(n^{-1/2})$ is generally expected to hold for a selector $C_n$ under which $d_0^r(Q_n^r,Q_{0,n}^r)=O_P(n^{-2/3}(\log n)^{d^r})$. 

\begin{theorem}\label{thsectnorm}
 \ \newline
{\bf Definitions:}
For a given $s$, let  $j^*=j^*_s=\arg\min_j P_0^r\phi_{s,j}$, and let $u_{s,j^*}$ be the corresponding knot point. 
Let $\bar{P}_0^r(s)\equiv \min_j P_0^r\phi_{s,j^*}$ be the probability that $O^r_s\geq u_{s,j^*}$ under $P_0^r$. Let $P_{0,s}^r$ represent the probability distribution $O^r_s=(O^r(j): j\in s)$.
For a cadlag function $Q$, define $\tilde{Q}_s(x_s)=\int_{(0,x_s]} dQ_s(u)$, so that $Q=\sum_s \tilde{Q}_s$. Thus, $Q_n^r=\sum_s \tilde{Q}_{n,s}^r$ and $Q_{0,n}^r=\sum_s \tilde{Q}_{0,n,s}^r$.  
 Let $R_n(s)\equiv \mid P_0^r \phi_{s,j^*}(Q_n^r-Q_{0,n}^r)\mid $ and
$R_n\equiv \min_{s,j,\beta_n^r(s,j)\not =0}\mid P_0^r \phi_{s,j}(Q_n^r-Q_{0,n}^r)\mid $. 
Let $r(n)=n^{-1/3}(\log n)^{d^r/2}\approx n^{-1/3}$.

{\bf Assumptions:}
\begin{itemize}
\item
$d_0(Q_n^r,Q_{0,n}^r)=O_P(r(n))$.
\item The loss-based dissimilarity is equivalent with a square of the $L^2(P_0^r)$-norm: 
$d_0^r(Q_n^r,Q_{0,n}^r)\sim \pl Q_n^r-Q_{0,n}^r\pl_{P_0^r}^2$.
\item For at least one subset $s$, we have that
$\pl \tilde{Q}_{n,s}^r-\tilde{Q}_{0,n,s}^r\pl_{P_{0,s}^r}=O_P(r(n))$; that there exists a $\delta>0$ so that \begin{equation}\label{sens1}
\frac{\int_{u\geq u_{s,j^*}}\mid Q^r_{n,s}(u)-Q^r_{0,n,s}(u)\mid dP_{0,s}^r(u)} {\{ \bar{P}_0^r(s) \}^{(\mid s\mid+1)/\mid s\mid } } >\delta>0\end{equation}
with probability tending to 1;
and, for some $0\leq \alpha\leq 1/2$, 
\begin{eqnarray*}
\pl \phi_{s,j^*}(Q_{n,s}^r-Q_{0,n,s}^r\pl_{1,P_0^r}&\leq&( \bar{P}_0^r(s))^{\alpha}\pl Q_{n,s}^r-Q_{0,n,s}^r\pl_{P_0} \\
\pl \phi_{s,j^*}(Q_n^r-Q_{0,n}^r\pl_{1,P_0^r}&\leq &(\bar{P}_0^r(s))^{\alpha}\pl Q_n^r-Q_{0,n}^r\pl_{P_0^r} .
\end{eqnarray*}
Note, we always can select $\alpha =1/2$ (conservatively), and, if $\pl Q_n^r-Q_{0,n}^r\pl_{\infty}=O_P(r(n))$, then we can set $\alpha=0$.
 Let ${\cal S}_1$ be the collection of subsets $s$ for which these two conditions   hold.
\end{itemize}
{\bf Conclusion:}
We have for each subset $s\in {\cal S}_1$:
 \begin{eqnarray*}
 \bar{P}_0^r(s)&=&O_P\left( (r(n))^{\frac{\mid s\mid}{\alpha \mid s\mid +1}}\right)\\
 &\approx&
 O_P\left ( n^{-\frac{\mid s\mid}{3+3\alpha\mid s\mid}}\right).
 \end{eqnarray*}
We have
\begin{eqnarray*}
R_n(s)&=& O_P\left(  r(n)^{\frac{1+\mid s\mid}{1+\alpha \mid s\mid}}\right ) .\end{eqnarray*}
This implies
\begin{eqnarray*}
R_n&=&O_P\left( \min_{s\in {\cal S}_1} r(n)^{\frac{1+\mid s\mid}{1+\alpha \mid s\mid}}\right ) .
\end{eqnarray*}
If  ${\cal S}_1$ includes a set $s$ with $\mid s\mid \geq 3$, then, even for $\alpha=1/2$, we have $R_n=o_P(n^{-1/2})$.
\end{theorem}
{\bf How does the bound on $R_n$ improve if we would have supnorm convergence:}
Suppose that $\pl Q_n^r-Q_{0,n}^r\pl_{\infty}=O_P(r(n))$ as well. Then we can select $\alpha=1$, so that for we obtain  $R_n(s)=O_P(r(n)^{1+\mid s\mid})\approx n^{-2/3}$ (even for $\mid s\mid =1$). 

{\bf Bounding assumption:}
In this theorem we assumed that for some $s$ $\pl \tilde{Q}_{n,s}^r-\tilde{Q}_{0,n,s}^r\pl_{P_0^r}$ can be bounded by $\pl Q_n^r-Q_{0,n}^r\pl_{P_0^r}$. 
$P_0^r$ describes a random variable on a cube $[0,\tau^r]\subset \openr^{d^r}$. In our example, this would be the distribution of $W^r$. In many applications, one might artificially truncate the covariate space from below and above so that is values are in a cube $[0,\tau^r]$.  In that case, the $s$-specific   edges $E_s=[0_s,\tau^r_s]\times \{0_{-s}\}$ of $[0,\tau^r]$ would have positive mass under $P_0^r$. Then, $\pl Q_n^r-Q_{0,n}^r\pl_{P_0^r}=
\sum_{s} \int_{E_s}(Q_n^r-Q_{0,n}^r)^2(u_s,0_{-s}) dP_{0}^r(u_s,0_{-s})$.
So, in that case, $\pl Q_n^r-Q_{0,n}^r\pl_{P_0^r}=O_P(r(n))$ would imply that the $L^2(P_0^r)$-norm of the difference of the $s$-specific sections, $Q_{n,s}^r-Q^r_{0,n,s}$, converge at same rate. This would naturally imply the same rate for the $s$-specific generalized differences $\tilde{Q}^r_{n,s}-\tilde{Q}^r_{0,n,s}$ of the sections, and thereby verify this bounding condition.
We suspect that this bounding assumption will apply to continuous $P_0^r$ as well (i.e., no mass on the edges $E_s$). Our reasoning is based on  $Q^r_n-Q^r_{0,n}=\sum_s(\tilde{Q}^r_{n,s}-\tilde{Q}^r_{0,n,s})$, and that for $s\not s_1$, the two sets of basis functions in $Q^r_{0,n,s_1}$ and $Q^r_{0,n,s_2}$, respectively,  are largely independent. That is, $Q^r_{0,n}$ represents an additive model (e..g, GAM), where each component $Q^r_{0,n,s}$ is identifiable from the total sum function ( by our definition of $\tilde{Q}^r_s$). It is true that, for example, basis functions $I(X_1>c_1)I(X_2>c_2)$ across knot points $(c_1,c_2)$ (i.e, $s_2=\{1,2\}$) can approximate $I(X_1>c_1)$ (i.e, $s_1=\{1\}$) by letting $c_2\approx 0$, but the $L^1$-norm (i.e., contribution to the  variation norm of $\tilde{Q}^r_{0,n,s}$)  of this small vector of coefficients represent a negligible proportion of the full $L^1$-norm (i.e, full  variation norm of $\tilde{Q}^r_{0,n,s}$) of the coefficients making up $Q^r_{0,n,s}$. 


{\bf Condition (\ref{sens1}):}
Consider one of the subsets $s\in {\cal S}_1$. Note that the numerator in (\ref{sens1}) is the  $L^1(P_0^r)$-norm 
$\pl \phi_{s,j^*}(Q_{n,s}^r-Q_{0,n,s}^r)\pl_{1,P_0^r}$. 
This follows since $\phi_{s,j^*}(u)=I(u\geq u_{s,j^*})$, and since $Q_{n,s}^r-Q_{0,n,s}^r$ is only a function of $W^r(s)$, the expectation w.r.t. $P_0^r$ becomes an expectation w.r.t. its marginal $P^r_{0,s}$.  
One expects that $\phi_{s,j^*}(u)=1$ for all $u\geq u_{s,j^*}$ for most of the basis functions with $\beta_n^r(s,j)\not =0$. So only a few basis functions will have some variation over $u>u_{s,j^*}$. For example, if $s$ is a singleton,  then, for all  $u\geq u_{s,j^*}$, we have $Q_{n,s}^r(u)=Q_{n,s}^r(u_{s,j^*})$ is constant in $u$, due to all basis functions in $Q_{n,s}^r$ being $1$ at such a $u$ (i.e., there are no basis functions $\phi_{s,j}$  in $Q_{n,s}^r$ with knot points larger than $u_{s,j^*}$). So in that case $(Q_{n,s}^r-Q_{0,n,s}^r)(u)=(Q_{n,s}^r(u_{s,j^*})-Q_{0,n,s}^r(u))$ for all $u\geq u_{s,j^*}$.  
Over a cube $A$ in an $\mid s\mid$-dimensional space, the variation in each of the $\mid s\mid$ coordinates is $A^{1/\mid s\mid}$, assuming that the sides of the cube are proportional to each other.
So, $\max_{k\in s}(\tau_s(k)-u_{s,j^*}(k))$  behaves as $(\bar{P}_0^r(s)^{1/\mid s\mid}$. Thus, the integral of $(Q_{0,n,s}^r(u_{s,j^*})-Q_{0,n,s}^r(u)$ behaves as 
$\bar{P}_0^r(s)^{1+1/\mid s\mid}=(\bar{P}_0^r(s))^{(\mid s\mid +1)/\mid s\mid}$.

{\bf Proof of Theorem \ref{thsectnorm}:}
Consider one of the sets $s\in {\cal S}_1$. 
By assumption we have $\pl \phi_{s,j^*}(Q_{n,s}^r-Q_{0,n,s}^r\pl_{1,P_0^r}\geq \{\bar{P}_0^r(s)\}^{\frac{\mid s\mid +1}{\mid s\mid}}$.
Now, use that $\pl \phi_{s,j^*} (Q_{n,s}^r-Q_{0,n,s}^r)\pl_{P_0^r}\leq \bar{P}_0^r(s)^{1-\alpha}\pl Q_{n,s}^r-Q_{0,n,s}^r\pl_{P_0^r}$.
This gives then
\[
\pl Q_{n,s}^r-Q_{0,n,s}^r\pl_{P_0^r}\geq   \bar{P}_0^r(s)^{\frac{\alpha\mid s\mid+1}{\mid s\mid}}
\]
Since the left-hand side is $O_P(r(n))$, this then shows that
\[
\bar{P}_0^r(s)=O_P\left ( r(n)^{\frac{\mid s\mid}{\alpha\mid s\mid+1} }\right).\]

Consider now the term $R_n(s)$ and note we can bound this by $\bar{P}_0^r(s)^{1-\alpha} \pl Q_n^r-Q_{0,n}^r\pl_{P_0^r}$. Combining the bound on $\bar{P}_0^r(s)$ above and $\pl Q_n^r-Q_{0,n}^r\pl_{P_0^r}=O_P(r(n))$, gives then
\[
R_n(s)=O_P \left ( r(n)^{-\frac{1+\mid s\mid}{1+\alpha\mid s\mid}}\right).\]
This implies the bound for $R_n$ by minimizing the latter over all $s\in {\cal S}_1$.
$\Box$

\section{Coordinate-Transformation for NIE} \label{sec:appendix_med}

In this section, we show that it is sufficient to have a two-dimensional coordinate-transformation for the purpose of estimation and conducting influence-curve based inference.

For $O = (W, A, Z, Y) \sim P_0$ and any $P \in \mathcal{M}$, we have 
\begin{align*}
    \frac{p_Z(Z | A = a', W)}{p_Z(Z | A = a, W)} = & \frac{p(A = a' | Z, W) p(Z, W)}{p(A = a' | W) p(W)} \cdot 
    \frac{p(A = a | W) p(W)}{p(A = a | Z, W) p(Z, W)} \\
    = & \frac{p(A = a' | Z, W)}{p(A = a' | W)} \cdot \frac{p(A = a | W)}{p(A = a | Z, W) } \\
    \frac{\indicator{A = a}}{p_A(A = a | W)} \cdot \frac{p_Z(Z | A = a', W)}{p_Z(Z | A = a, W)}
    = & \frac{\indicator{A = a} p(A = a' | Z, W)}{p(A = a' | W) p(A = a | Z, W)}
\end{align*}

Let 
\begin{align*}
    G(A | W) = & p(A | W) \\
    \gamma(A| Z, W) = & p(A | Z, W), \\
    Q_Y(P)(Z, W) = & \ex_P\newbracket{Y | Z, A = a, W}, \\
    Q_Z(P)(W) = & \ex_P\newbracket{ Q_Y(P)(Z, W) |A = a', W}, 
\end{align*}
then (see also \citep{zheng2017longitudinal} and \cite{wang2023targeted})
\begin{align*}
    D^*_Y(P) = & \frac{\indicator{A = a}}{G(A = a' | W)} \frac{\gamma(A = a' | Z, W)}{\gamma(A = a | Z, W)} \newbrace{Y - Q_Y(P)} \\
    D^*_Z(P) = & \frac{\indicator{A = a'}}{G(A = a' | W)} \newbrace{Q_Y(P) - \ex_P\newbrace{ Q_Y(P) | A = a', W}}. 
\end{align*}

Therefore, we can define a dimension-reduced dataset, which constructs a coordinate-transformation from the original data $O$ to 
$$O^r = (W^r, A, Z^r, Y), $$
where $W^r(W) = \newparethensis{P(A = a | W), Q_Z(P)(W)}$ and $Z^r(W, Z) = \newparethensis{P(A = a | Z, W), Q_Y(P)(Z, W)}$. Note that $O^r$ has the same data structure as $O$. For the data-adaptive version with $V$-fold cross-validation, we have
\begin{align*}
    D^r_Y(P^r)(v, O^r(v, O)) = & \frac{\indicator{A = a}}{G^r_v(a' | W^r)} \frac{\gamma^r_v(a' | Z^r(Z, W))}{\gamma^r_v(a | Z^r(Z, W))} \newbrace{Y - Q^r_Y(Z^r, A = a, W^r)} \\
    D^r_Z(P^r)(v, O^r(v, O)) = & \frac{\indicator{A = a'}}{P(A = a' | W)} \newbrace{Q^r_Y(Z^r, A = a, W^r) - Q^r_Z(A = a', W^r)}. 
\end{align*}
It can be verified that 
\begin{align*}
    D^{r}(G^r, \gamma^r, Q_Z^r, Q_Y^r)(v, O^r(v, O)) = & D^{*}(G^r_v, \gamma^r_v, Q^r_Z \circ W^r_v, Q^r_Y \circ Z^r_v). 
\end{align*}

In Section \ref{sec:data}, $W=W^r=\emptyset$, and therefore only a 2-dimensional coordinate-transformation $Z \mapsto Z^r = \newparethensis{P(A = a | Z, W), Q_Y(P)(Z, W)}$ is required for each transfer learning based model $P$. For example, $Q_Y(P)(Z, W)(v, O)$ is an estimated function of the conditional expectation of $Y$ given $Z$ and $A = a$ trained on the $v$-th training sample, which can be a transfer learning application \citep{raina2007self,zeiler2014visualizing,donahue2014decaf,chen2019med3d} of the pretrained model with only the last layer re-trained for predicting the new outcome $Y$.

 \end{appendix}

\bibliographystyle{unsrtnat} 

\bibliography{ref}

\end{document}